\documentclass[10pt,twocolumn]{article}

\usepackage[utf8]{inputenc}
\usepackage{times}
\usepackage{graphicx}
\usepackage{amsmath, amssymb}
\usepackage{url}
\usepackage{hyperref}
\usepackage{authblk}

\usepackage{helvet}  % DO NOT CHANGE THIS
\usepackage{courier}  % DO NOT CHANGE THIS
\usepackage{graphicx} % DO NOT CHANGE THIS
\usepackage{natbib}  % DO NOT CHANGE THIS AND DO NOT ADD ANY OPTIONS TO IT
\usepackage{caption} % DO NOT CHANGE THIS AND DO NOT ADD ANY OPTIONS TO IT
\usepackage{algorithm}
\usepackage{algorithmic}

\usepackage{amsmath, amssymb, amsthm} % Matematica avanzata
\usepackage{stmaryrd}
\usepackage{tikz}
\usetikzlibrary{patterns}
\usepackage{graphicx} % Per includere immagini
\usepackage{caption} %caption
\usepackage[table]{xcolor}
\usepackage{soul,xcolor}
\sethlcolor{lightgray}
\usepackage[font=footnotesize]{caption}

\newtheorem{theorem}{Theorem}[section]
\newtheorem{lemma}{Lemma}[section]
\newtheorem{corollary}{Corollary}[section]
\newtheorem{remark}{Remark}[section]

\newcommand{\vv}{\boldsymbol{v}}
\newcommand{\bfx}{\boldsymbol{x}}
\newcommand{\lpar}{\llbracket}
\newcommand{\rpar}{\rrbracket}
\newcommand{\LL}{{\cal L}}
\newcommand{\EQOR}{\text{EQ1}_{g}^{c}}
\newcommand{\EFOR}{\text{EF1}_{g}^{c}}
\newcommand{\EQOP}{\text{EQ1P}_{g}^{c}}
\newcommand{\EFOP}{\text{EF1P}_{g}^{c}}
\newcommand{\EQXR}{\text{EQX}_{g}^{c}}

\usepackage{newfloat}
\usepackage{listings}
\DeclareCaptionStyle{ruled}{labelfont=normalfont,labelsep=colon,strut=off} % DO NOT CHANGE THIS
\floatstyle{ruled}
\newfloat{listing}{tb}{lst}{}
\floatname{listing}{Listing}
\title{Approximately Envy-free and Equitable Allocations of\\ Indivisible Items for Non-monotone Valuations\thanks{This work was partially supported by: the Horizon EU Framework Programme under Grant
agreement No 101183743 (AGATE); the PNRR MIUR project FAIR - Future AI Research (PE00000013), Spoke 9 - Green-aware AI; the MUR - PNRR IF Agro@intesa; the Project SERICS (PE00000014) under the NRRP MUR program funded by the EU – NGEU; GNCS-INdAM.}}

\author[1]{Vittorio Bilò}
\author[2]{Martin Loebl}
\author[1]{Cosimo Vinci}

\affil[1]{Department of Mathematics and Physics ``Ennio De Giorgi'', University of Salento, Italy\\ \texttt{\{vittorio.bilo,cosimo.vinci\}@unisalento.it}}
\affil[2]{Department of Applied Mathematics, ``Charles'' University of Prague, Czech Republic\\ \texttt{loebl@kam.mff.cuni.cz}}

\date{}

\begin{document}

\maketitle

\begin{abstract}
We revisit the setting of fair allocation of indivisible items among agents with heterogeneous, non-monotone valuations. We explore the existence and efficient computation of allocations that approximately satisfy either envy-freeness or equity constraints. Approximate envy-freeness ensures that each agent values her bundle at least as much as those given to the others, after some (or any) item removal, while approximate equity guarantees roughly equal valuations among agents, under similar adjustments. As a key technical contribution of this work, by leveraging fixed-point theorems (such as Sperner's Lemma and its variants), we establish the existence of {\em envy-free-up-to-one-good-and-one-chore} ($\text{EF1}^c_g$) and {\em equitable-up-to-one-good-and-one-chore} ($\text{EQ1}^c_g$) allocations, for non-monotone valuations that are always either non-negative or non-positive. These notions represent slight relaxations of the well-studied {\em envy-free-up-to-one-item} (EF1) and {\em equitable-up-to-one-item} (EQ1) guarantees, respectively. Our existential results hold even when items are arranged in a path and bundles must form connected sub-paths. The case of non-positive valuations, in particular, has been solved by proving a novel multi-coloring variant of Sperner's Lemma that constitutes a combinatorial result of independent interest. In addition, we also design a polynomial-time dynamic programming algorithm that computes an $\text{EQ1}^c_g$ allocation. For monotone non-increasing valuations and path-connected bundles, all the above results can be extended to EF1 and EQ1 guarantees as well. Finally, we provide existential and computational results for certain stronger {\em up-to-any-item} equity notions under objective valuations, where items are partitioned into goods and  chores.
\end{abstract}


\section{Introduction}
{\em Fair division} \cite{Steinhaus48}, the field that studies how to fairly allocate resources among a set of agents, has numerous applications across a variety of real-life scenarios, such as divorce settlements, credit assignment, and rent and land division, to name a few. Although fair division has been studied for decades in mathematics and economics, the field has attracted increasing attention from the computer science and AI community in recent years, driven by the flourishing of new fairness concepts and the demand for computationally efficient solutions overcoming the inherent impossibility of achieving optimal fairness guarantees.

Two prominently investigated notions of fairness are {\em envy-freeness} \cite{F66} and {\em equitability} \cite{DS61}. An allocation (of items to agents) is {\em envy-free} (EF) if the value that every agent gives to her assigned bundle (of items) is not less than the value she gives to the bundle assigned to any other agent; it is {\em equitable} (EQ) if the value that every agent gives to her assigned bundle is not less than the value that the other agents assign to their respective bundles. So, the two notions coincide when agents have identical valuations.

The nature of valuation functions tremendously impacts the solution of a fair division problem. When an agent's valuation is monotone non-decreasing (resp., non-increasing), items are said to be {\em goods} (resp., {\em chores}) for the agent; when it is non-monotone, items are said to be {\em mixed}. Notable special cases of non-monotone valuations include {\em non-negative} (resp. {\em non-positive}) valuations, where every bundle yields a non-negative (resp. non-positive) value, and {\em objective} valuations, in which items can be partitioned into goods and chores. Valuations, either monotone and non-monotone, are {\em additive}, when each item has a value and the value of a bundle is defined by the sum of the values of its items. 

While objective valuations have been widely studied \cite{ACIW22,BBPP24}, less work has been done for non-negative or non-positive ones, despite their potential applicability in numerous settings. These valuations, for instance, arise when items correspond to nodes in an edge-weighted graph with exclusively positive or negative weights, allocations are interepreted as clusterings, and the value of each bundle/cluster is then determined by factors such as the total weight of internal or cut edges, or other graph-connectivity properties. This class of settings has interesting connections with clustering problems where further fairness guarantees are required (see, e.g., \cite{DinitzSTV22,Chierichetti0LV17,SchwartzZ22}).

Either envy-free or equitable allocations are guaranteed to exist under non-negative or non-positive valuations in the setting of {\em divisible items}, where items can be arbitrarily split among subsets of agents \cite{DS61,S80,W80,S99,CDP13,C17,AD15,BhaskarSV25}. In contrast, in presence of {\em indivisible items} which have to be integrally assigned to any of the agents, existence cannot be guaranteed even for two agents with additive monotone valuations. To overcome this limitation, a number of relaxations have been proposed in the literature. They allow the removal of one item from a bundle when agents perform bundle comparisons. The removal strategy clearly depends on the nature of the considered items. When an agent compares her bundle $A$ against another bundle $B$, she can choose between removing a chore from $A$ or removing a good from $B$. These relaxations have given rise to the notions of {\em envy-freeness-up-to-any-good} (EFX) \cite{CK+16}, {\em envy-freeness-up-to-one-good} (EF1) \cite{LMMS04,B11}, {\em equitability-up-to-any-good} (EQX) \cite{GMT14} and {\em equitability-up-to-one-good} (EQ1) \cite{FSVX19}. By {\em up-to-any-good}, one means that the fairness property holds irrespectively of which item is selected for removal; by {\em up-to-one-good}, instead, the property must hold for at least one removed item. Clearly, fairness {\em up-to-any-good} implies fairness {\em up-to-one-good}.

An interesting and largely studied generalization of fair division assumes the existence of an {\em item graph} modeling proximity relationships among items. Every bundle has to induce a connected subgraph and an item removal is allowed only if it does not disconnect the induced subgraph \cite{BCEIP17,BCFIMPVZ22,I23,MSVV21,S19}. With this respect, the {\em path constraint} assumes that the item graph is a connected path.

\subsubsection{Our Contribution.}

Given the lack of positive results for non-monotone valuations under standard approximate fairness notions, often due to impossibility barriers (see, e.g., \cite{AABFLMVW23,BBPP24}), we study slight relaxations of EF1, EQ1 and EQX, denoted by $\EFOR$ ({\em envy-free-up-to-one-good-and-one-chore}), $\EQOR$ ({\em equitable-up-to-one-good-and-one-chore}), and $\EQXR$ ({\em equitable-up-to-any-good-or-any-chore}), respectively. While EF1 and EQ1 require the envy-freeness and equitability properties, respectively, to hold upon the removal of at most one item from ``some'' bundle, $\EFOR$ and $\EQOR$ allow the properties to hold when removing at most one item from ``each'' bundle (i.e., for each agent, at most one chore from her own bundle and at most one good from others’ bundles). Similarly, while EQX requires the equitability property to hold regardless of which item is removed (i.e., whether it is a good or a chore), $\EQXR$ allows it to hold for the removal of either goods only or chores only (see Section \ref{model}). We obtain positive results on the existence and computation of allocations satisfying the above fairness criteria across several broad classes of non-monotone valuations. 
%We recall that, for identical valuations, all results on approximate equitability naturally extend to the corresponding fairness concepts (and vice-versa), since the two notions become equivalent.

{\em Results for non-negative or non-positive valuations: } Our main contribution concerns the existence and computation of  allocations which are fair {\em up-to-one-good-and-one-chore}, under non-negative or non-positive valuations. In particular, we show that an $\EQOR$ allocation always exists and can be computed in polynomial time, for both non-negative (Theorem \ref{thm1_eq1p} and \ref{thm2_eq1p}) and non-positive valuations (Theorem \ref{thm1_eq1p_nonpos_short_version}).
For $\EFOR$ allocations, we only show  existence (Theorem \ref{thm3_ef1p} and  Theorem~\ref{thm2_ef1p_nonpos_short_version}).
The existence and computation of approximately envy-free or equitable allocations under non-monotone valuations is one of the major open problems in fair division (see, e.g., the surveys by \citet{AABFLMVW23,LiuLSW24}). Our results represent a significant step forward in this direction, due to the generality of the non-monotone valuations we consider (either non-negative or non-positive) and the fairness guarantees achieved (requiring the removal of at most one good and one chore). 

It is worth noting that our results continue to hold even under path constraints, and in this sense, they generalize the findings of \citet{BCFIMPVZ22,I23,MSVV21,S19}, which apply only to monotone non-decreasing valuations, and the results of \citet{BCEIP17,BouveretCL19}, which address the computational problem of finding EF and EQ allocations (that, in general, may not exist). Our existential results are obtained using Sperner's Lemma \cite{Sperner1928} or its variants, and represent a non-trivial generalization of the approaches previously explored in \cite{BCFIMPVZ22,I23}, as handling the non-monotonicity of the valuations poses significant technical challenges (such as the derivation and analysis of the cases described in Figure~\ref{fig:3} of Appendix \ref{app:pictures}). In particular, to address the specific case of non-positive valuations, we introduce and exploit a novel multi-coloring variant of Sperner's Lemma (Theorem \ref{multiSperner}), that constitutes a combinatorial result of independent interest. The computational results have been obtained by means of the dynamic-programming paradigm. Finally, when valuations are monotone (non-decreasing or non-increasing), our results extend to the stronger EQ1 and EF1 guarantees under path constraints, thereby generalizing the results of \citet{I23,BCFIMPVZ22,MSVV21}, which exibith EF1 and EQ1 allocations under monotone non-decreasing valuations and path constraints.

{\em Results for objective valuations:} To complete the picture for non-monotone valuations, we also consider fair allocations under objective valuations, and in most of the cases the obtained results hold even under the {\em up-to-any-good-or-any-chore} approximation guarantee. In particular, we show that an $\EQXR$ allocation always exists and can be computed in pseudo-polynomial time (Theorem~\ref{thm1} of Appendix \ref{sec:objective}), via a simple variant of the local-search approach adopted by \citet{BBPP24}. This result extends to EQX when valuations are monotone non-increasing (Corollary~\ref{thm1_cor1} of Appendix \ref{sec:objective}), thereby generalizing the result of \citet{BBPP24}, which holds only for monotone non-decreasing valuations. For valuations that are both objective and additive, we strengthen the above computational result by showing that an $\EQXR$ allocation can be found through a simple and more efficient greedy algorithm (Theorem \ref{thm3} of Appendix \ref{sec:objective}). This result strengthens the findings of \citet{HS25}, who establish existence and polynomial-time computability of EQ1 allocations under additive objective valuations. With this respect, we also show that a slight generalization of the polynomial-time algorithm proposed by \citet{HS25}, continues to produce EQ1 allocations even for objective valuations that are non-additive (Theorem~\ref{thm2} of Appendix \ref{sec:objective}). It is worth noting that, just as $\EQXR$ is considered a relaxation of EQX, similar relaxations for EFX have been studied for objective \cite{HosseiniSVX23} or identical non-monotone valuations \cite{BercziBBGKKKM24}. However, these works only show non-existence of these relaxed notions.

Most of the technical details on non-negative and non-positive valuations are deffered to Appendix \ref{app:non-negative}-\ref{app:non_positive}, while those regarding objective valuations are all deferred to Appendix \ref{sec:objective}-\ref{app:objective}. Table \ref{tab:senza-newcol2} summarizes both our results and related work on the considered fairness notions and classes of valuations, also specifying the cases that remain unresolved and are therefore posed as open problems.
\begin{table}[ht]
\centering
\footnotesize
\addtolength{\tabcolsep}{-2pt}
\renewcommand{\arraystretch}{0.9}
\resizebox{\linewidth}{!}{%
\begin{tabular}{|c|c|c|c|c|c|c|}
\hline
\textbf{} & \textbf{Gen} & \textbf{NNeg} & \textbf{NPos} & \textbf{NDec} & \textbf{NInc} & \textbf{Obj} \\
\hline
\begin{tabular}{@{}c@{}}\textbf{$\EQOR$} \\ \hline \textbf{$\EFOR$} \end{tabular} &
\begin{tabular}{@{}c@{}} x$_{\text{a}}$ \\ \hline ?$_{\text{a}}$ \end{tabular} & 
\begin{tabular}{@{}c@{}} \hl{$\checkmark ^{\text{P,cns}}$} \\ \hline \hl{$\checkmark^{\text{cns}}$} \end{tabular} & 
\begin{tabular}{@{}c@{}} \hl{$\checkmark^{\text{P,cns}}$} \\ \hline \hl{$\checkmark^{\text{cns}}$} \end{tabular} & 
\begin{tabular}{@{}c@{}} {$\checkmark^{\text{P,cns}}$} \\ \hline {$\checkmark^{\text{cns}}$} \end{tabular} & 
\begin{tabular}{@{}c@{}} \hl{$\checkmark^{\text{P,cns}}$} \\ \hline \hl{$\checkmark^{\text{cns}}$} \end{tabular} & 
\begin{tabular}{@{}c@{}} $\checkmark^{\text{P}}$ \\ \hline $\checkmark^{\text{P}}$ \end{tabular} \\
\hline
\begin{tabular}{@{}c@{}}\textbf{EQ1} \\ \hline \textbf{EF1} \end{tabular} & 
\begin{tabular}{@{}c@{}} x$_{\text{a}}$ \\ \hline ?$_{\text{a}}$ \end{tabular} & 
\begin{tabular}{@{}c@{}} $?$ \\ \hline $?$ \end{tabular} & 
\begin{tabular}{@{}c@{}} $?$ \\ \hline $?$ \end{tabular} & 
\begin{tabular}{@{}c@{}} $\checkmark^{\text{P,cns}}$ \\ \hline $\checkmark^{\text{cns}}$ \end{tabular} & 
\begin{tabular}{@{}c@{}} \hl{$\checkmark^{\text{P,cns}}$} \\ \hline \hl{$\checkmark^{\text{cns}}$} \end{tabular} & 
\begin{tabular}{@{}c@{}} $\checkmark^{\text{P}}$ \\ \hline $\checkmark^{\text{P}}$ \end{tabular} \\
\hline
\textbf{$\EQXR$} & x$_{\text{a}}$ & ? & ? & $\checkmark^{\text{P-}}$ & \hl{$\checkmark^{\text{P-}}$} & \hl{$\checkmark^{\text{P-}},\checkmark^{\text{P}}_{\text{ao}}$} \\
\hline
\begin{tabular}{@{}c@{}}\textbf{EQX} \\ \hline \textbf{EFX} \end{tabular} & 
\begin{tabular}{@{}c@{}} x$_{\text{a}}$ \\ \hline x$_{\text{a}}$ \end{tabular} & 
\begin{tabular}{@{}c@{}} $?$ \\ \hline $?_{\text{a}}$ \end{tabular} & 
\begin{tabular}{@{}c@{}} $?$ \\ \hline $?_{\text{a}}$ \end{tabular} & 
\begin{tabular}{@{}c@{}} $\checkmark^{\text{P-}}$ \\ \hline ?$_{\text{a}}$ \end{tabular} &
\begin{tabular}{@{}c@{}} $\checkmark^{\text{P-}}$ \\ \hline ?$_{\text{a}}$ \end{tabular} & 
\begin{tabular}{@{}c@{}} x$_{\text{a}}$ \\ \hline x$_{\text{a}}$ \end{tabular} \\
\hline
\end{tabular}
}\footnotesize
\caption{Landscape of results for the considered fairness notions. Gen, NNeg, NPos, NDec, NInc, and Obj stand, respectively, for General, Non-negative, Non-positive, Non-decreasing, Non-increasing, and Objective valuations. Gray-highlighted results refer to our findings. $\checkmark$, x and ? mean respectively ``it always exists'', ``it does not generally exist'' and ``existence is an open problem''. Subscript ``a'' (resp., ``ao'') means that the the result of type x or ? (resp., $\checkmark$) holds even (resp., only) for additive valuations. Superscript ``P'' (resp., ``P-'') means that the existence can be obtained via a polynomial (resp., pseudo-polynomial) algorithm, and superscript ``cns'' means that the result holds even under path constraints.}\label{tab:senza-newcol2}
\end{table}

\paragraph{Further Related Work.}
%For divisible items, EQ allocations are always guaranteed to exist \cite{DS61,CDP13,C17,AD15}. Although no algorithm can compute one \cite{PW17}, an approximate EQ allocation is achievable in polynomial time \cite{CP12}.  

\citet{LMMS04} show that an EF1 allocation always exists and can be efficiently computed. Their result has been extended to objective valuations by \citet{ACIW22,BSV21,BercziBBGKKKM24}. Again under obejctive valuations, \cite{HS25} show existence and polynomial time computation of an EQ1 allocation; conversely, they show that, under additive non-objective valuations, EQ1 allocations may not exist. Existence and computation of Pareto optimal EF1 or EQ1 allocations have been studied by \citet{CK+16,FSVX19,FSVX20,GargM24}. EF1 and EQ1 allocations under non-objective valuations have been determined for restricted cases only, and their general existence and computation is a major open problem (see surveys by \citet{AABFLMVW23,LiuLSW24}). 

EQX allocations were first proved to exist for additive monotone non-decreasing valuations in \cite{GMT14}.
%\footnote{We point out that, in most of the literature, monotone valuation means also non-decreasing. In this work, we specify the difference to deal with both non-decreasing and non-increasing valuations.}
 Efficient algorithms computing one have been later designed by \citet{FSVX19,FSVX20}, also covering the non-increasing case. Existence and pseudo-polynomial time computation of an EQX allocation for monotone non-decreasing (and possibly non-additive) valuations has been shown in \cite{BBPP24}. 
This result is complemented by proving that, if one drops the monotonicity assumption, EQX allocations may not exist even for two agents with additive valuations. So, our results  state that, slightly relaxing EQX to $\EQXR$ suffices to recover existence under non-monotone valuations, as long as they are objective (see Table \ref{tab:senza-newcol2}).

%EFX allocations were first proved to exist under identical monotone non-decreasing valuations, but they are hard to compute even in the two-agent case \cite{PR20}. 

The EFX criterion was introduced by \citet{CK+16}, and its existence, under additive non-negative or non-positive valuations, has been addressed only in specific cases and it remains a major open problem in fair division (see surveys by \citet{AABFLMVW23,LiuLSW24}); instead, for objective additive valuations, an EFX allocation might not exist \cite{HosseiniSVX23}. \citet{PR20} show that an EFX allocation exists and can be efficiently computed for non-negative additive valuations when agents assign the same ranking to all items; this result has been also extended to additive objective valuations with equally ranked items \cite{AR20}. Restricting to identical valuations, EFX allocations are known to exist for additive non-decreasing valuations \cite{GMT14} and additive non-increasing valuations \cite{BNV23}, whereas they may not exist for non-monotone valuations \cite{BercziBBGKKKM24}.
%Under general objective additive valuations, \citet{CL20} show the existence of allocations which satisfy a variant of the EFX guarantee.

\section{Model and Definitions}\label{model}
Let  $N = \{1, 2,\ldots, n\}$ be a finite set of $n$ {\em agents} and $M$ be a finite set of $m$ {\em items}. Each agent $i\in N$ has an integral {\em valuation function} $v_i : 2^M \to \mathbb{Z}$ with $v_i(\emptyset)=0$
%\footnote{The normalization constraint $v_i(\emptyset)=0$ is a standard assumption in fair allocation, as it holds without loss of generality for exact and/or approximate envy-freeness, or is somehow necessary in order to satisfy exact and/or approximate equitability criteria.}
for any $i\in N$. We denote by $I=(N,M,(v_i)_{i\in N})$ an {\em allocation instance}. Given an agent $i\in N$, a bundle of items $S\subseteq M$ and an item $x\in S$, we say that $x$ is a {\em good} (resp., a {\em chore}) {\em for $i$ w.r.t. $S$}, if $v_i(S)\geq v_i(S\setminus \{x\})$ (resp., $v_i(S)\leq v_i(S\setminus \{x\})$). We observe that an item $x$ for which $v_i(S)=v_i(S\setminus \{x\})$ is both a good and a chore for $i$ w.r.t. $S$. An item $x$ is a good (resp., a chore) if it is a good (resp., a chore) for any agent $i\in N$ w.r.t. any bundle $S\subseteq M$. An {\em allocation} $\mathcal{A}=(A_1,\ldots, A_n)$ is a partition of $M$ in $n$ (possibly empty) {\em bundles} of items, such that $A_i$ is the bundle assigned to agent $i\in N$. We aim at finding allocations satisfying fairness criteria related to  {\em envy-freeness} and {\em equity}, as described below.   
\paragraph{Envy-freeness.}
An allocation $\mathcal{A}~=~(A_1,\ldots, A_n)$ is:
\begin{itemize}
\item {\em envy-free} (EF) if, for any $i,j\in N$, $v_i(A_i)\geq v_i(A_j)$;

\item {\em envy-free-up-to-any-item} (EFX) if, for any $i,j\in N$ such that $v_i(A_i)< v_i(A_j)$, all the following conditions hold: (i) $v_i(A_i)\geq v_i(A_j\setminus \{g\})$ for any good $g$ for agent $i$ w.r.t. $A_j$; (ii) $v_i(A_i\setminus \{c\})\geq v_i(A_j)$ for any chore $c$ for $i$ w.r.t. $A_i$; (iii) either there exists a good $g$ for $i$ w.r.t. $A_j$, or there exists a chore $c$ for $i$ w.r.t. $A_i$;

%\item {\em envy-free-up-to-any-good-or-any-chore} ($\EFXR$) if, for any $i\in N$, at least one of the following conditions hold: (i) for any \( j \in N \) such that \( v_i(A_i) < v_i(A_j) \), there exists one good for \( i \) w.r.t. \( A_j \), and for each such good \( g \), \( v_i(A_i) \geq v_i(A_j \setminus \{g\}) \); (ii) for any \( j \in N \) such that \( v_i(A_i) < v_i(A_j) \), there exists a chore for \( i \) w.r.t. \( A_i \), and for each such chore \( c \), \( v_i(A_i \setminus \{c\}) \geq v_i(A_j) \);

\item {\em envy-free-up-to-one-item} (EF1) if, for any $i,j\in N$ such that $v_i(A_i)< v_i(A_j)$, there exists $x\in A_i\cup A_j$ such that $v_i(A_i\setminus \{x\})\geq v_i(A_j\setminus \{x\})$;

\item {\em envy-free-up-to-one-good-and-one-chore} ($\EFOR$) if, for any $i,j\in N$ such that $v_i(A_i)<v_i(A_j)$, there exists a subset $X\subseteq M$ with $|A_i\cap X|\leq 1$ and $|A_j\cap X|\leq 1$, such that $v_i(A_i\setminus X)\geq v_i(A_j\setminus X )$.
\end{itemize}

%Under the EFX guarantee, agent $i$ stops envying agent $j$ after removing an arbitrary chore from their own bundle $A_i$ or an arbitrary good from the other agent's bundle $A_j$. The $\EFXR$ condition, instead, is a relaxation of EFX where, for any agent $i$, envy-freeness is achieved by restricting to either chores deletion from $A_i$ or goods deletion from each other's bunlde $A_j$, but not necessarily both. 

For EF1 allocations, agent $i$ stops envying agent $j$ after removing at most one chore from $A_i$ or at most one good from $A_j$, but not both. Similarly, under $\EFOR$ allocations, agent $i$ stops envying agent $j$ after removing at most one chore from $A_i$ and at most one good from $A_j$, even simultaneously.  We observe that $\text{EF}\Rightarrow\text{EFX} \Rightarrow \text{EF1} \Rightarrow \text{$\EFOR$}$. 

\paragraph{Equitability.}
The equitability notions we consider are analogous to the envy-freeness criteria described above, but the comparison each agent $i$ makes is not against the valuation $v_i(A_j)$ that she assigns to the bundle given to any other agent $j$, but rather against the valuation $v_j(A_j)$ that agent $j$ assigns to her own bundle. Thus yields allocations that are {\em equitable} (EQ), {\em equitable-up-to-any-item} (EQX), {\em equitable-up-to-one-item} (EQ1) and {\em equitable-up-to-one-good-and-one-chore} ($\EQOR$) (see Appendix \ref{app:explicit} for their explicit definition). In addition, we consider {\em equitable-up-to-any-good-or-any-chore} ($\EQXR$) allocations, a slight relaxation of EQX that, unlike EQX, allows the equitability to hold after the removal of either goods only or chores only. In particular, an allocation $\mathcal{A}$ is $\EQXR$ if, for any $i\in N$, at least one of the following conditions holds: (i) for any $j\in N$ such that $v_i(A_i)< v_j(A_j)$, there exists at least one good for $j$ w.r.t. $A_j$, and for any such good $g$, $v_i(A_i)\geq v_j(A_j\setminus \{g\})$; (ii) for any $j\in N$ such that $v_i(A_i)< v_j(A_j)$, there exists a chore for $i$ w.r.t. $A_i$, and for any such chore $c$, $v_i(A_i\setminus \{c\})\geq v_j(A_j)$.
%We obtain the notions of {\em equitable} (EQ), {\em equitable-up-to-any-item} (EQX), {\em equitable-up-to-any-good-or-any-chore} ($\EQXR$), {\em equitable-up-to-one-item} (EQ1) and {\em equitable-up-to-one-good-and-one-chore} ($\EQOR$) allocations by replacing, for every $S\subseteq A_j$  written as $A_j,A_j\setminus \{g\}, A_j\setminus X$, $v_i(S)$ with $v_j(S)$ in the definitions of EF, EFX, $\EFXR$, EF1 and $\EFOR$ allocations. This requires defining the removed good $g$ as good for agent $j$ w.r.t. $A_j$ in the definitions of EQX and $\EQXR$ (see Appendix \ref{app:explicit} for an explicit description of these notions).
%In particular, an allocation is {\em equitable} (EQ) if, for any $i,j\in N$, $v_i(A_i)\geq v_j(A_j)$. The approximate variants of EQ are provided by {\em equitable-up-to-any-item} (EQX), {\em equitable-up-to-any-good-or-any-chore} ($\EQXR$), {\em equitable-up-to-one-item} (EQ1) and {\em envy-free-up-to-one-good-and-one-chore} ($\EQOR$) allocations, which are respectively defined as the EF, EFX, $\EFXR$ and EF1 allocations considered above, but with \( v_j(S) \) in place of \( v_i(S) \), where \( S \) is the bundle assigned to \( j \) (i.e., \( A_j \)) or the same bundle after the removal of some items (e.g., \( A_j \setminus \{g\} \)), and where each good removed from \( A_j \) is intended for \( j \) instead of \( i \). 
We observe that $\text{EQ}\Rightarrow \text{EQX} \Rightarrow \text{$\EQXR$} \Rightarrow \text{EQ1} \Rightarrow \text{$\EQOR$}$. 
\paragraph{Classes of Valuations.} We consider the following classes of valuations: (i) {\em Objective:} any item $x$ is either a good or a chore (independently on the considered agents and bundles); in such a case, we can partition \( M \) into a set of goods \( G \) and a set of chores \( C \) (choosing arbitrarily how to classify dummy items that qualify as both); (ii) {\em Non-negative} (resp., {\em Non-positive}): $v_i(S)\geq 0$ (resp. $v_i(S)\leq 0$) for any $i\in N$, $S\subseteq M$; (iii) {\em Monotone non-decreasing} (resp., {\em non-increasing}): each item $x$ is a good (resp., a chore), independently on the considered agents and bundles; (iv) {\em Additive}: $v_i(S)=\sum_{x\in S}v_i(x)$ for any $i$, $S\subseteq M$. 

We will prove our positive results for the most general class of valuations, keeping in mind that monotone non-decreasing (resp. non-increasing) valuations are also non-negative (resp. non-positive), and that
%(ii) the $v_i$s are monotone non-increasing $\Rightarrow$  the $v_i$s are non-positive;
monotone, either non-decreasing or non-increasing, valuations are also objective; instead, objective valuations are not necessarily non-negative or non-positive, and vice versa. Finally, we assume the existence of a {\em constant-time oracle} that, given $i\in N$ and $S\subseteq M$, returns $v_i(S)$ in constant time. 

%Furthermore, since the notion of EF (resp., EFX, $\EFXR$, EF1, $\EFOR$) is equivalent to that of EQ (resp., EQX, $\EQXR$, EQ1, $\EQOR$) under identical valuations, we will avoid specifying that the results holding for the former notion also apply to the latter (and vice-versa) in this special case. 

%this assumption is trivially satisfied for additive valuations, when the value of each item is explicitly provided.
\section{Non-negative Valuations}\label{sec:notwholly}
To address the case of non-negative valuations, we consider a generalization of the fixed-point approach employed in \cite{BCFIMPVZ22,I23}, that is extended in a non-trivial manner to handle the peculiarities of these valuations.

%Indeed, their analysis initially focuses on obtaining approximately fair ``almost'' fractional allocations, in which each item may be split into two parts. The final integral allocation is achieved via a rounding procedure that crucially exploits the integrality of the boundary items in each such fractional bundle. In contrast, in our setting, each item may be divided into three parts, which increases the complexity of the rounding procedure. This added complexity also stems from the need to simultaneously account for both chores and goods.
\subsection{Fairness under Path Constraints}
We say that an allocation instance is {\em path-constrained} if the $m$ items of $M$ are numbered from $1$ to $m$ and organized as a path $P=(1,\ldots, m)$. Given $s,t\in [m]\cup \{0\}$, let $\lpar s,t\rpar$ denote the bundle $\{s,s+1,\ldots, t\}$ if $t\geq s$, and the empty bundle otherwise. A bundle $S$ is {\em connected} if $S=\lpar s,t\rpar$ for some $s\in [m]$ and  $t\in [m]\cup \{0\}$. An allocation $\mathcal{A}$ is {\em connected} if it is made by connected bundles only. Given a connected bundle $S=\lpar s,t\rpar$, let $\partial S=\{s,t\}$ denote the {\em boarder} of $S$ (observe that $\partial S=S$ if $|S|\leq 2$). For a given connected allocation $\mathcal{A}=(A_1,\ldots, A_n)$, we consider the following path-based notions of $\EFOR$ and $\EQOR$: $\mathcal{A}$ is {\em envy-free-up-to-one-good-and-one-chore-over-paths} ($\EFOP$) if, for any $i,j\in N$ such that $v_i(A_i)<v_i(A_j)$, there exists a subset $X\subseteq \partial A_i\cup \partial A_j$ with $|\partial A_i\cap X|\leq 1$ and $|\partial A_j\cap X|\leq 1$, such that $v_i(A_i\setminus X)\geq v_i(A_j\setminus X )$; $\mathcal{A}$ is {\em equitable-up-to-one-good-and-one-chore-over-paths} ($\EQOP$) if, for any $i,j\in N$ such that $v_i(A_i)<v_j(A_j)$, there exists a subset $X\subseteq \partial A_i\cup \partial A_j$ with $|\partial A_i\cap X|\leq 1$ and $|\partial A_j\cap X|\leq 1$, such that $v_i(A_i\setminus X)\geq v_j(A_j\setminus X )$.

Note that $\text{$\EFOP$}\Rightarrow\text{$\EFOR$}$ and $\text{$\EQOP$}\Rightarrow\text{$\EQOR$}$. Furthermore, the notions of $\EQOP$ and $\EFOP$ guarantee that, even after some items are deleted, the bundles remain connected.

\subsubsection{Sperner's Lemma.}
Before presenting our results, we provide a brief overview of the underlying theoretical framework based on Sperner's Lemma. For more details, see, for example, \cite{Flegg1974}.

Let $\text{conv}(\vv_1, \vv_2, \ldots, \vv_n)$ denote the convex hull of the $n$ vectors $\vv_1, \vv_2, \ldots, \vv_n$. An {\em $(n-1)$-simplex} $\Delta$ is an $(n-1)$-dimensional polytope defined as the convex hull of its $n$ (affinely independent) {\em vertices} $\vv_1, \vv_2, \ldots, \vv_n$. Given $k \in [n]$, a {\em $(k-1)$-face} of an $(n-1)$-simplex is the $(k-1)$-simplex obtained as convex hull of a subset of $k-1$ of its vertices. A {\em triangulation} $T$ of a simplex $\Delta$ is a collection of sub-$(k-1)$-simplices (with $k \in [n]$) whose union is $\Delta$, with the property that the intersection of any two sub-simplices in $T$ is either empty or a  face shared by both, which also belongs to $T$. Each sub-simplex $\Delta' \in T$ is referred to as an {\em elementary} simplex. The set of vertices of $T$, denoted as $V(T)$, is the union of the vertices of all the elementary simplices in $T$ (i.e., the union of all the elementary $0$-simplices). 

Now, let $T$ be a fixed triangulation of an $(n-1)$-simplex $\Delta = \text{conv}(\vv_1, \vv_2, \ldots, \vv_n)$. A {\em coloring function} of $T$ is a mapping $L : V(T) \to [n]$ that assigns a number, referred to as a {\em color}, from the set $[n]$ to each vertex of $T$. A coloring function $L$ is called {\em special} if, for any vertex $\bfx \in V(T)$ belonging to the $(n-2)$-face $F_i$ of $\Delta$ that does not include $\vv_i$ (i.e., the face opposite to $\vv_i$, obtained as convex hull of all vertices of $\Delta$ except for $\vv_i$), the condition $L(\bfx) \neq i$ holds. We observe that, if $L$ is a special coloring function, then $L(\vv_i) = i$ holds for any $i \in [n]$. An elementary $(n-1)$-simplex $\Delta^*=\text{conv}(\bfx_1^*,\ldots, \bfx_n^*)\in T$ is said to be {\em fully-colored} under a coloring function $L$ if each of its $n$ vertices is assigned a distinct color by $L$, that is, $L(\bfx_{\sigma(i)}^*)=i$ for any $i\in [n]$, for some permutation $\sigma:[n]\rightarrow [n]$. 

\begin{theorem}[Sperner's Lemma \cite{Sperner1928}]\label{lem:sperner}
Let $T$ be a triangulation of an $(n-1)$-simplex $\Delta$, where $n\geq 2$, and let $L$ be a special coloring function of $T$. Then, there exists a fully-colored elementary $(n-1)$-simplex $\Delta^* \in T$ under $L$; moreover, the number of such simplices is odd.
\end{theorem}
See Figure \ref{fig:1}(a) in Appendix \ref{app:pictures} for an example of the application of Sperner's Lemma with $n=3$. 

Below, we also consider a generalized version of Sperner's Lemma, as presented by \cite{Bapat1989}. In this generalized form, there are $ n $ special coloring functions $ L_1, \dots, L_n $, and we seek an elementary $(n-1)$-simplex that is fully-colored according to a broader definition, which holds simultaneously for all of the coloring functions. Let $ T $ be a triangulation of an $ (n-1) $-simplex, and let $ L_1, \dots, L_n $ be the coloring functions on $ T $. An elementary $(n-1)$-simplex $ \Delta^* = \text{conv}(\bfx_1^*, \bfx_2^*, \ldots, \bfx_n^*) \in T $ is \emph{jointly fully-colored} under $ L_1, \dots, L_n $ if there exist two permutations $ \sigma,\tau : [n] \to [n] $ such that 
$
L_i(\bfx_{\sigma(i)}^*) = \tau(i)$ for any $i \in [n]
$, i.e., each vertex of $\Delta^*$ receives a distinct color under a distinct coloring function.
\begin{theorem}[Generalized Sperner's Lemma \cite{Bapat1989}]\label{lem:sperner_gen}  
Let $T$ be a triangulation of an $(n-1)$-simplex $\Delta$, and let $L_1, \dots, L_{n}$ be special coloring functions of $T$. Then, there exists a jointly-fully-colored elementary $(n-1)$-simplex $\Delta^* \in T$ under $L_1,\ldots, L_{n}$.
\end{theorem}

\subsection{$\EQOP$ Allocations}
Given a path-constrained allocation instance with non-negative valuations, we construct a suitable triangulation $ T $ of an $ n $-simplex $ \Delta $ and define a special coloring $ L $ for $ T $. Each elementary $(n-1)$-simplex $ \Delta^* \in T $ that is fully-colored under $ L $ corresponds to an $\EQOP$ allocation for $ I $. By Sperner's Lemma (Theorem~\ref{lem:sperner}), the existence of such fully-colored simplices is guaranteed, which in turn ensures the existence of an $\EQOP$ allocation. We also design a polynomial-time algorithm, based on dynamic programming, that efficiently computes such an $\EQOP$ allocation.
 
\paragraph{Triangulation.} 
Consider the $(n-1)$-simplex $\Delta = \{ \bfx = (x_1, \ldots, x_{n-1}) \in \mathbb{R}^{n-1} : 0 \leq x_1 \leq x_2 \leq \ldots \leq x_{n-1} \leq m \}$, which is the convex hull $\text{conv}(\vv_1, \vv_2, \ldots, \vv_n)$ of the points $\vv_1,\ldots, \vv_n$, with $\vv_i := (\overbrace{0, 0, \ldots, 0}^{i-1}, \overbrace{m, m, \ldots, m}^{n-i})$ for any $i\in [n]$. 
We observe that each of the $n$ $(n-2)$-faces of $\Delta$ can be defined as $F_i := \left\{ \bfx = (x_1, \ldots, x_{n-1}) \in \Delta : x_{i-1} = x_i \right\}$, where we set $x_0 := 0$ and $x_n := m$. We construct a triangulation $T$ of $\Delta$ whose set of vertices is $V(T) = \{ \bfx \in \Delta : x_i \in \{ 0, \frac 1 3, \frac 2 3, 1, \frac 4 3,\frac 5 3, 2, \frac 7 3, \dots, m-1, m-\frac 2 3,m- \frac 1 2,m \}\ \forall i\in [n-1]\}$ and whose simplicial structure is defined below. Each coordinate $x_i$ of vertices $\bfx\in V(T)$ can be either {\em integral}, or {\em 1-fractional} or {\em 2-fractional}, where integral (resp. 1-fractional, 2-fractional) means $x_i\in \mathbb{Z}$ (resp. $x_i-\frac 1 3 \in \mathbb{Z}$, $x_i-\frac 2 3 \in \mathbb{Z}$); we write $x_i\equiv 0$ (resp. $x_i\equiv 1$, $x_i\equiv 2$) if $x_i$ is integral (resp. 1-fractional, 2-fractional). By leveraging {\em Kuhn's triangulation} \cite{Kuhn1960,Scarf1982,Deng2012}, we construct the triangulation $ T $ such that each elementary $(n-1)$-simplex $\Delta' = \text{conv}(\bfx_1, \bfx_2, \ldots, \bfx_n) \in T$ can be generated by fixing the first vertex $\bfx_1 \in V(T)$ and a permutation $\pi : [n-1] \to [n-1]$, and then iteratively determining the remaining vertices as follows: $
\bfx_{i+1} = \bfx_i + \textstyle \frac{1}{3} \mathbf{e}^{\pi(i)}$ for each $i \in [n-1]$, where $\mathbf{e}^i ~=~ (\overbrace{0, \dots,0}^{i-1}, 1, \overbrace{0,\dots, 0}^{n-i-1})$ is the $i$-th vector of the canonical basis of $\mathbb{R}^{n-1}$. Figure \ref{fig:1}(b) in Appendix \ref{app:pictures} describes Kuhn's triangulation for $n=3$. 

Each vertex $\bfx \in V(T)$ can be understood as a vector representing the positions of $n-1$ knives that divide the interval $[0, m]$ into $n$ connected segments having extremes in $a,b\in [0, m]\cap \{\frac{x}{3}|x\in \mathbb{Z}\}$. Following this interpretation, the $n$ vertices of any elementary $(n-1)$-simplex in $T$ are derived by starting with an initial configuration of $n-1$ cuts (i.e., vertex $\bfx_1$) and sequentially shifting each knife one position to the right (by a length of $1/3$) according to a specific ordering defined by a permutation $\pi$. Refer to Figure \ref{fig:2} in Appendix \ref{app:pictures} to visualize the process of deriving the sequence of fractional allocations from the $ n $ vertices of an $ (n-1) $-dimensional simplex within a triangulation $ T $, which is constructed from an allocation instance with $n=3$ agents.

\paragraph{Coloring Function.} We now construct the coloring function $L : V(T) \to [n]$. Given a vertex $\bfx=(x_1,\ldots, x_{n-1})\in V(T)$, let $\mathcal{\tilde{A}}(\bfx)=(\tilde{A}_1(\bfx), \ldots, \tilde{A}_n(\bfx))$ be the {\em fractional connected allocation} obtained from the partition of $[0,m]$ in $n$ {\em fractional connected bundles}, defined as $\tilde{A}_i(\bfx)=[x_{i-1},x_i]$ for any $i\in [n]$, with $x_0:=0$ and $x_n:=m$; furthermore, each bundle $\tilde{A}_i(\bfx)$ is assigned by default to agent $i$, for any $i\in [n]$ (i.e., bundles are assigned from left to right following the agents order). 

Given $a\in \mathbb{R}_{\geq 0}$, let $a^-:=\lfloor a\rfloor$ and $a^+:=\min\{\lfloor a \rfloor+1,m\}$. Let $\tilde{v}_i$ denote the {\em virtual valuation} of agent $i$, which applies to fractional connected bundles $[a,b]$ (where $a,b\in [0, m]\cap \{\frac{x}{3}|x\in \mathbb{Z}\}$ and $a \leq b$) and returns an integer value $\tilde{v}_i([a,b])$ that is defined as follows:
{\em left-value (LV):}  if $a\equiv 0$, $\tilde{v}_i([a,b]):=v_i(\lpar a^-,b^-\rpar)$; {\em borderline-value (BV):} if $a\equiv 1$, $\tilde{v}_i([a,b])$ is set equal to the middle value among $v_i(\lpar a^-,b^-\rpar)$, $v_i(\lpar a^+,b^-\rpar)$ and $v_i(\lpar a^+,b^+\rpar)$; {\em right-value (RV):} if $a\equiv 2$, $\tilde{v}_i([a,b]):=v_i(\lpar a^+,b^-\rpar)$. Since the original valuations $v_i$s are non-negative, the resulting virtual valuations $\tilde{v}_i$s are also non-negative. Let $ L $ be the coloring function that assigns each vertex $\bfx$ the agent/index $i$ that maximizes the virtual valuation $\tilde{v}_i(\tilde{A}_i(\bfx))$ applied to the fractional connected bundle $\tilde{A}_i(\bfx)$, where ties are broken in favor of agents receiving a non-empty bundle and, in case of further ties, arbitrarily. We observe that $L$ is a special coloring function. Indeed, for any $i\in [n]$, the $(n-2)$-face $F_i$ of $\Delta$, which does not contain $\vv_i$, is such that the fractional allocations $\mathcal{\tilde{A}}(\bfx)$ corresponding to vertices $\bfx \in V(T)$ located on $F_i$ have their $i$-th bundle empty ($\tilde{A}_i(\bfx) = \emptyset$). Due to the non-negativity of the virtual valuations, any empty bundle always has the lowest virtual value, regardless of the agent or allocation being considered. Therefore, by the construction of $L$, we have $L(\bfx) \neq i$ for any $i \in [n]$ and any vertex $\bfx \in V(T)$ located on the $(n-2)$-face $F_i$. Thus, $L$ is a special coloring function. Figure \ref{fig:1}(b) of Appendix \ref{app:pictures} shows an example of special coloring function $L$ derived from an arbitrary non-negative virtual valuation function. 

\paragraph{From the Fully-colored Simplex to the $\EQOP$ Allocation.} 
According to Sperner's Lemma (Theorem~\ref{lem:sperner}), there exists at least one fully-colored elementary $(n-1)$-simplex $\Delta^* = \text{conv}(\bfx_1^*, \ldots, \bfx_n^*) \in T$ under the coloring $L$, where $L(\bfx_{\sigma(i)}^*) = i$ for all $i \in [n]$, for some permutation $\sigma:[n]\rightarrow [n]$. Equivalently, each $i \in [n]$ is among those agents $j$ who maximize the virtual valuation $\tilde{v}_{j}(\tilde{A}_{j}(\bfx_{\sigma(i)}^*))$ in the fractional connected allocation $\tilde{\cal A}(\bfx_{\sigma(i)}^*)$ associated with the $\sigma(i)$-th vertex $\bfx_{\sigma(i)}^*$ of $\Delta^*$, where the sequence of allocations $\tilde{A}(\bfx_1^*), \ldots, \tilde{A}(\bfx_n^*)$ is obtained by moving each knife one at a time from left to right in a specific order, starting from the position of knives determined by $\tilde{A}(\bfx_1^*)$. See Figure \ref{fig:2} of Appendix \ref{app:pictures} for an example. 

%For $ n = 3 $, Figure \ref{fig:2} of Appendix \ref{app:pictures} illustrates how to interpret a fully-colored simplex $ \Delta^* $ in terms of the ranking of bundle values determined by the virtual valuations of the fractional allocations corresponding to the $ n $ vertices of $ \Delta^* $.

Denote by $\mathcal{\tilde{A}}$ the first allocation $\mathcal{\tilde{A}}(\bfx_1^*)$, and refer to it as the {\em main allocation} of $\Delta^*$.
For a bundle $ \tilde{A}_j = [a_j, b_j] $ in the main allocation, say that $ \tilde{A}_j$ is {\em left-first} (resp., {\em right-first}) in $\Delta^*$ if the first allocation $\mathcal{\tilde{A}}(\bfx_h^*)$ associated with $\Delta^*$, for which the bundle $ \tilde{A}_j(\bfx_h^*)= [a_j', b_j'] $ differs from $ \tilde{A}_j$, satisfies $ a_j' = a_j + 1/3 $ and $b_j'=b_j$ (resp., $a_j=a_j'$ and $ b_j' = b_j + 1/3 $).
Equivalently, $\tilde{A}_j$ is {\em left-first} (resp.,  {\em right-first}) in $\Delta^*$ if, among the two knives determining the endpoints of the $j$-th bundle across all allocations associated with $\Delta^*$, the first to move from left to right is the left (resp., right) one; for an example of right-first bundle, see Figure \ref{fig:2} of Appendix \ref{app:pictures}. By appropriately rounding the fractional bundles of $\mathcal{\tilde{A}}$, we will obtain the desired (integral) allocation $ \mathcal{A} $ that satisfies the $\EQOP$ guarantee. The rounding procedure processes all fractional bundles $\tilde{A}_j$s of the main allocation $ \mathcal{\tilde{A}}$ from $j=n$ down to $j=1$, and for each  bundle $ \tilde{A}_j $, it returns the integral bundle $ A_j $ that will form the final (integral) allocation $ \mathcal{A} = (A_1, \ldots, A_n) $. Specifically, once the bundles $ A_{j+1}, \ldots, A_n $ have been determined, the bundle $ A_j$ is obtained by rounding the fractional bundle $ \tilde{A}_j= [a_j, b_j]$ based on the three possible fractionality levels of the two endpoints, $a_j$ and $b_j$, and, if necessary, on whether $\tilde{A}_j$ is left-first or right-first. This rounding process involves considering $ 9 = 3 \times 3 $ possible cases (corresponding to the three fractionality levels for each endpoint) and additional sub-cases, and it is formally described in Figure \ref{fig:3} of Appendix \ref{app:pictures}. The rounding procedure that returns $\mathcal{A}$ is carefully designed so that the statements of the following two lemmas hold. Their proofs, strongly based on the rounding procedure, are deferred to Appendix~\ref{app:non-negative}.
\begin{lemma}\label{lem1_eq1p}
$\mathcal{A}$ is a connected (integral) allocation. 
\end{lemma}
Given $ i \in N $ and a connected (integral) bundle $ S = \lpar s, t \rpar \subseteq [m] $, let  
$
v_i^+(S) = \max\{\lpar s, t \rpar, \lpar s+1, t \rpar, \lpar s, t-1 \rpar\}
$
and
$
v_i^-(S) = \min\{\lpar s, t \rpar, \lpar s+1, t \rpar, \lpar s, t-1 \rpar\};
$
$ v_i^+(S) $ and $ v_i^-(S) $ represent, respectively, the maximum and the minimum valuation that agent $ i $ can obtain from bundle $ S $, after possibly removing one of its endpoint items.
\begin{lemma}\label{lem2_eq1p}
For any $h,i,j\in [n]$, the connected (integral) allocation $\mathcal{A}$ satisfies $v_i^-({A}_j)\leq \tilde{v}_i(\tilde{A}_j(\bfx^*_h))\leq v_i^+({A}_j)$.
\end{lemma}
Using these lemmas, we can show that the allocation $\mathcal{A}$ returned by the rounding procedure is $\EQOP$.
\begin{theorem}\label{thm1_eq1p}
$\mathcal{A}$ is an $\EQOP$ allocation, if valuations are  non-negative.
\end{theorem}
\begin{proof}
First, $ \mathcal{A} $ is a connected allocation by Lemma~\ref{lem1_eq1p}. Next, we show the $\EQOP$ guarantee. As observed above, the full coloring of simplex $\Delta^*$ implies that each $i\in [n]$ is one of the indices $j\in [n]$ that maximize $ \tilde{v}_j(\tilde{A}_j(\bfx_{\sigma(i)}^*)) $ (i.e., agent $i$ has the highest virtual valuation in allocation $\tilde{\mathcal{A}}(\bfx_{\sigma(i)}^*)$). Thus, for any $ i, j \in N $, we have 
$
v_i^+(A_i) \geq \tilde{v}_i(\tilde{A}_i(\bfx_{\sigma(i)}^*)) \geq \tilde{v}_j(\tilde{A}_j(\bfx_{\sigma(i)}^*)) \geq v_j^-(A_j),
$
where the second inequality follows from the above observation, and the first and last inequalities follow from Lemma~\ref{lem2_eq1p}. Since $ v_i^+(A_i) \geq v_j^-(A_j) $ for any $ i, j \in N $, we conclude that $ \mathcal{A} $ satisfies the $\EQOP$ guarantee (i.e., equitability is obtained by removing at most one good from the board of $A_i$ and one chore from the board of $A_j$), and thus the claim holds.
\end{proof}
\paragraph{Efficient Computation.}
The $\EQOP$ allocation guaranteed by Theorem \ref{thm1_eq1p} can be  computed by a polynomial-time algorithm based on dynamic programming. The algorithm first computes the set $C_v$ of the valuations $v_i(S)$ that each agent $i$ has for any bundle $S$ (in $O(nm^2)$ time) and then, by dynamic programming, determines for each $c\in C_v$ if there exists an allocation $\mathcal{A}$ such that $v_i^+(A_i)\geq c\geq v_i^-(A_i)$ for any $i\in [n]$ (in $O(nm^2)$ time), where $v_i^+(S)$ and $v_i^-(S)$ denote the maximum and the minimum valuation that  $i$ can obtain from a bundle $S$ by deleting at most one item from its board; again, we restrict ourselves to allocations where the $i$-th leftmost bundle is assigned to agent $i$. We show that finding such a value $c\in C_v$ satisfying the above condition is equivalent to finding an $\EQOP$ allocation, whose existence is guaranteed by Theorem \ref{thm1_eq1p}. Then, we get the following theorem (full details are deferred to Appendix \ref{app:non-negative}):
\begin{theorem}\label{thm2_eq1p}
If valuations are non-negative, an $\EQOP$ allocation can be found in time $O(n^2m^4)$.
\end{theorem}
%
%\begin{remark}\label{rema2}
%By exploiting Remark \ref{rema}, we have that, once we fix an ordering of the agents, Theorem \ref{thm2_eq1p} can be applied to obtain $(A_1,\ldots, A_n)$ returned by the algorithm are ordered from the left-most to the right-most one. 
%\end{remark}
\subsection{$\EFOP$ allocations}
To show the existence of $\EFOP$ allocations, we employ the same framework as in the approximate equitability case, with minor modifications. We use the same triangulation $ T $ as in the previous case but equip it with $ n $ distinct coloring functions $ L_1, \ldots, L_n $, instead of the single coloring function $ L $ used earlier. Here, each $L_i$ colors any vertex in $V(T)$ with the index $j$ of the bundle that agent $i$ prefers under virtual valuation $\tilde{v}_i$ (defined as in the previous case); we note that each $L_i$ is special, as the empty bundle is the least valuable.

By applying the Generalized Sperner’s Lemma (Theorem~\ref{lem:sperner_gen}), we show the existence of a jointly fully-colored elementary $(n-1)$-simplex $\Delta^*$. As in the previous case, this simplex corresponds to a sequence of $n$ connected fractional partitions, but now the bundles are initially unallocated, and there exist two permutations $\sigma$ and $\tau$ such that, in the $\sigma(i)$-th allocation, agent $i \in [n]$ does not envy any other agent if $i$ receives the $\tau(i)$-th bundle. Then, by applying the same rounding procedure and proof techniques used for the case of $\EQOP$ allocations, the first fractional partition of $\Delta^*$ is transformed into an (integral) $\EFOP$ allocation, where each agent $i$ receives the $\tau(i)$-th bundle. This leads to the following theorem (see Appendix \ref{app:ef1} for the full details).
\begin{theorem}\label{thm3_ef1p}
Under non-negative valuations, an $\EFOP$ allocation always exists.
\end{theorem}
We conjecture that the computation of an $\EFOP$ allocation is a PPAD-complete problem (similarly to the results show in \cite{Deng2012}). We also note that, even without path constraints, the complexity of finding EF1 or $\EFOR$ still remains an open problem. It is worth noting that, in the subclass of monotone non-decreasing valuations, there are goods only. Thus, in this case, Theorems \ref{thm1_eq1p}-\ref{thm3_ef1p} extend to the stronger notions of EF1 and EQ1 under path-constraints, thereby recovering the findings of \citet{BCFIMPVZ22,I23,MSVV21,S19}.

\section{Non-positive Valuations}\label{sec:multisperner}
To address the case of non-positive valuations under path-connectivity constraints, we resort to a novel multi-coloring variant of Sperner's Lemma, where the underlying coloring functions assign, to each vertex $\bfx\in V(T)$, a set of colors (rather than a single color), including the indices $i$ of the $(n-2)$-dimensional faces $F_i$ to which $\bfx$ belongs.
\subsubsection{Multi-coloring Sperner's Lemma.}
Let $T$ be a fixed triangulation of an $(n-1)$-simplex $\Delta=\text{conv}(\vv_1,\ldots, \vv_n)$. A {\em multi-coloring function} of $T$ is a mapping $\LL:V(T)\rightarrow 2^{[n]}\setminus\{\emptyset\}$ that assigns a non-empty subset of colors $\LL(\bfx)\subseteq [n]$ to each vertex of $\bfx\in V(T)$. We recall that $F_i$ is the $(n-2)$-dimensional face of $\Delta$ opposite to vertex $\vv_i$. A multi-coloring function $\LL$ is called {\em special} if, for any vertex \(\bfx \in V(T)\), $\LL(\bfx)\supseteq \{i\in [n]:\bfx\in F_i\}$ holds (i.e., if $\bfx$ is a boundary vertex, the set of colors $\LL(\bfx)$ contains the indices associated with all $(n-2)$-faces of $\Delta$ on which $\bfx$ is located). We observe that, if $\LL$ is a special multi-coloring function and $F$ is a $(k-1)$-face of $\Delta$ spanned by vertices $\vv_{i_1},\ldots, \vv_{i_k}$, it holds that $\LL(\bfx) \supseteq [n]\setminus\{i_1,\ldots, i_k\}$ for any vertex $\bfx\in V(T)$ located on $F$. An elementary $(n-1)$-simplex \(\Delta^*=\text{conv}(\bfx_1^*,\ldots, \bfx_n^*)\in T\) is said to be {\em fully-colored} under a multi-coloring function $\LL$ if there exists a permutation $\sigma:[n]\rightarrow [n]$ such that $i\in \LL(\bfx^*_{\sigma(i)})$ for any $i\in [n]$ (that is, a distinct color $i$ appears in the set $\LL(\bfx^*_{\sigma(i)})$ associated with a distinct vertex $\bfx^*_{\sigma(i)}$). An example of a special multi-coloring function \( \LL \) applied to the triangulation \( T \) of a \( 2 \)-simplex is provided in Figure \ref{fig:4} of Appendix \ref{app:pictures}.

\begin{theorem}[Multi-coloring Sperner's Lemma]\label{multiSperner}
Let $T$ be a triangulation of an $(n-1)$-simplex $\Delta$, where $n\geq 2$, and let $\LL$ be a special multi-coloring function of $T$. Then, there exists a fully-colored elementary $(n-1)$-simplex $\Delta^* \in T$ under multi-coloring function $\LL$.
\end{theorem}
The Multi-coloring Sperner’s Lemma can be viewed as a dual to the standard Sperner’s Lemma. In the classical version, color $i$ is ``prohibited'' from appearing at any vertex $\bfx$ located on the face $F_i$ opposite to the vertex $\vv_i$. In contrast, the multi-coloring version requires that color $i$ ``must appear'' in the set of assigned colors $\LL(\bfx)$ at each vertex $\bfx \in F_i$. To show the Multi-coloring Sperner's Lemma, we first assume, w.l.o.g., that \( \LL(\bfx) = \{i \in [n] : \bfx \in F_i\} \) holds for any vertex $\bfx$ located on the boundary of $\Delta$, and that \( | \LL(\bfx) | = 1 \) for any internal vertex \( \bfx \) not located on the boundary. Then, we define the {\em minimal restriction} of $\LL$ as the standard coloring function that assigns the color $L(\bfx) = \min\{i\in \LL(\bfx)\}$ to each vertex $\bfx$. We show by induction on $n\geq 2$ that the number of fully-colored simplices with respect to the minimal restriction \( L \) is odd, which implies the existence of at least one such simplex (full details are deferred to Appendix \ref{app:non_positive}). Finally, we note that each fully-colored simplex with respect to \( L \) is also fully-colored with respect to the original multi-coloring function \( \LL \), thereby proving the claim of the theorem. The full proof is deferred to Appendix \ref{app:non_positive}.
It is worth noting that, unlike the standard Sperner's Lemma, the structure of our multi-coloring functions required additional intermediate steps and ad hoc topological transformations of the simplicial structure in order to carry out the inductive argument on $n$.

\subsection{$\EQOP$ and $\EFOP$ Allocations}
To show the existence of an $\EQOP$ allocation, we employ the same framework as in the case of non-negative valuations, with minor modifications. We use the same triangulation \( T \) as in that case, but we equip $T$ with the multi-coloring function $\LL$ that assign to each vertex $\bfx\in V(T)$ the set $\LL(\bfx)$ of indices $j$ which maximize $\tilde{v}_j(\tilde{A}_j(\bfx))$, where $\tilde{v}_j$ is the virtual valuation defined as in Section \ref{sec:notwholly}. Differently from the case of non-negative valuations, under non-positive valuations the empty bundle is always the best one for each agent. Thus, $\LL$ is a special multi-coloring function and, by the Multi-coloring Sperner's lemma (Theorem \ref{multiSperner}), there exists an elementary \((n-1)\)-simplex \(\Delta^* = \text{conv}(\bfx^*_1, \ldots, \bfx^*_n)\) that is fully-colored under the multi-coloring function \(\LL\). As in the case of non-negative valuations, this means that there exists a permutation $\sigma$ such that each agent $i\in [n]$ is the happiest (among all others) in the allocation $\mathcal{A}(\bfx^*_{\sigma(i)})$ (where each agent $j$ receives the $j$-th bundle). 
From this point onward, we can apply the same approach used for non-negative valuations to transform $\Delta^*$ into an $\EQOP$ allocation (full details are deferred to Appendix \ref{app:non_positive}).
\begin{theorem}\label{thm1_eq1p_nonpos_short_version}
Under non-positive valuations, an $\EQOP$ allocation always exists and can be computed in polynomial time. 
\end{theorem}
As in the case of non-negative valuations, to show the existence of $\EFOP$ allocations, we consider a generalization of the multi-coloring Sperner's Lemma that deals with $n$ distinct multi-coloring functions, each one modelling the virtual valuation of each agent $i$. We then obtain a jointly fully-colored simplex $\Delta^*$, in which each vertex corresponds to an allocation where a distinct agent prefers a distinct bundle. Then, by resorting to the usual rounding procedure we obtain the desired $\EFOP$ allocation (full details are deferred to Appendix~\ref{app:non_positive}):
\begin{theorem}\label{thm2_ef1p_nonpos_short_version}
Under non-positive valuations, an $\EFOP$ allocation always exists.
\end{theorem}
The following corollary holds since, in the case of monotone non-increasing valuations, there are chores only.
\begin{corollary}\label{thm1_eq1p_chore}
Under non-increasing valuations, EQ1 and EF1 allocations always exist, even under path constraints, with the former being computable in polynomial time.
\end{corollary}
%%TODOe

\section{Conclusions}
As the main contribution of this work, we established the existence of allocations that are (approximate) equitable ($\EQOR$) or envy-free ($\EFOR$) up to the removal of one good or chore from each bundle, even for non-positive valuations and under path constraints. Furthermore, efficient computation can be achieved under the approximate equitability guarantee. With these results, we made significant progress on the general problem of establishing the existence and computation of allocations that are fair “up to some items” for general non-monotone valuations. However, the existence of allocations satisfying the stronger EQ1 and EF1 guarantees remains open, as does the case of more general non-monotone valuations beyond the settings of the non-negative or non-positive valuations. It would be interesting to explore whether our techniques can be extended to address these cases.

Additionally, the time complexity of finding approximately fair allocations for general non-negative or non-positive valuations remains an open question, even in cases where existence has been established. Indeed, under path constraints, we conjecture that finding $\EFOP$ allocations is PPAD-complete for both non-negative and non-positive valuations. Furthermore, we conjecture that PPAD-completeness also holds for finding EF1 allocations under monotone non-decreasing or non-increasing valuations.

Finally, the existence of EQX allocations for objective valuations remains an open problem, even in the additive case. Furthermore, moving to the approximate envy-freeness guarantee, we note that the existence of EFX allocations is a major open question in fair division, even for additive non-negative valuations.

\subsection*{AI Use Declaration}
ChatGPT (OpenAI) was used solely for language polishing and figure layout formatting. All scientific content and data were prepared by the authors.
\bibliographystyle{plainnat}
\bibliography{ijcai25}
\appendix
\newpage\ \newpage
\section*{Supplementary Material of Paper "Approximately Envy-free and Equitable Allocations of Indivisible Items \\for Non-monotone Valuations" and Reproducibility Checklist}
The supplementary material is organized in seven distinct appendices (A-H). For a better organization, each distinct appendix starts in a new page. The reproducibility checklist is provided at the end of the document. 
\section{Explicit Description of the Approximate Equitability Notions (from Section \ref{model})}\label{app:explicit}
An allocation $\mathcal{A}=(A_1,\ldots, A_n)$ is:
\begin{itemize}
\item {\em equitable} (EQ): if, for any $i,j\in N$, $v_i(A_i)\geq v_j(A_j)$.
\item {\em equitable-up-to-any-item} (EQX): if, for any $i,j\in N$ such that $v_i(A_i)< v_j(A_j)$, all the following conditions hold: (i) $v_i(A_i)\geq v_j(A_j\setminus \{g\})$ for any good $g$ for $j$ w.r.t. $A_j$; (ii) $v_i(A_i\setminus \{c\})\geq v_j(A_j)$ for any chore $c$ for $i$ w.r.t. $A_i$; (iii) either there exists a good $g$ for $j$ w.r.t. $A_j$, or there exists a chore $c$ for $i$ w.r.t. $A_i$.
%\item {\em equitable-up-to-any-good-or-any-chore} ($\EQXR$): if, for any $i\in N$, at least one of the following conditions holds: (i) for any $j\in N$ such that $v_i(A_i)< v_j(A_j)$, there exists at least one good for $j$ w.r.t. $A_j$, and for any such good $g$, $v_i(A_i)\geq v_j(A_j\setminus \{g\})$; (ii) for any $j\in N$ such that $v_i(A_i)< v_j(A_j)$, there exists a chore for $i$ w.r.t. $A_i$, and for any such chore $c$, $v_i(A_i\setminus \{c\})\geq v_j(A_j)$.
\item {\em equitable-up-to-one-item} (EQ1): if, for any $i,j\in N$ such that $v_i(A_i)< v_j(A_j)$, there exists $x\in A_i\cup A_j$ such that $v_i(A_i\setminus \{x\})\geq v_j(A_j\setminus \{x\})$.
\item {\em envy-free-up-to-one-good-and-one-chore} ($\EQOR$): if, for any $i,j\in N$ such that $v_i(A_i)<v_j(A_j)$, there exists a subset $X\subseteq M$ with $|A_i\cap X|\leq 1$ and $|A_j\cap X|\leq 1$, such that $v_i(A_i\setminus X)\geq v_j(A_j\setminus X )$.
\end{itemize}
\newpage \ \newpage
\section{Missing Proofs from Section \ref{sec:notwholly}:\\part on equitability}

\subsection{Proof of Lemma \ref{lem1_eq1p}}\label{app:non-negative}
The rounding procedure described in Figure \ref{fig:3} of Appendix \ref{app:pictures} verifies the following properties: (i) the right endpoint $t_n$ of bundle \( A_n=\lpar s_n,t_n\rpar \) is \( m \); (ii) for any \( j \in [n-1] \), the right endpoint \( t_j \) of bundle \( A_j=\lpar s_j,t_j\rpar \) and the left endpoint $s_{j+1}$ of bundle \( A_{j+1} \) satisfy \( t_j = s_{j+1} - 1 \); (iii) the left endpoint $s_1$ of bundle \( A_1 \) is \( 1 \). Based on these observations, all the connected bundles \( A_j \) are disjoint and their union covers all items in \( M \).

\subsection{Proof of Lemma \ref{lem2_eq1p}}
For the rounding procedure, refer to Figure \ref{fig:3} in Appendix \ref{app:pictures}. Fix $h,i,j\in [n]$ and let $\mathcal{\tilde{A}}'$ denote allocation $\mathcal{\tilde{A}}(\mathbf{x}_h^*)$. Let \(\tilde{A}_j(L,L) := [a_j,b_j]\), \(\tilde{A}_j(L,R) := [a_j, \min\{b_j+1/3, m\}]\), \(\tilde{A}_j(R,L) := [\min\{a_j+1/3, b_j\}, b_j]\), and \(\tilde{A}_j(R,R) := [\min\{a_j+1/3, b_j\}, \min\{b_j+1/3, m\}]\). Equivalently: \(\tilde{A}_j(L,L)\) represents the $j$-th fractional bundle \(\tilde{A}_j\) of the allocation $\mathcal{\tilde{A}}=\mathcal{\tilde{A}}(\bfx^*)$ associated with the first vertex of $\Delta^*$; \(\tilde{A}_j(L,R)\) is the fractional bundle obtained from \(\tilde{A}_j\) by moving the right knife one position (i.e., 1/3 of the size of an item) to the right, ensuring it remains before the last possible position; \(\tilde{A}_j(R,L)\) is the fractional bundle obtained by moving the left knife one position to the right (while keeping it before the right knife); \(\tilde{A}_j(R,R)\) is the fractional bundle obtained by moving both knives one position to the right, ensuring the left knife remains before the right knife. We observe that each possible configuration of the \(j\)-th bundle in $\mathcal{\tilde{A}}$ must be one of \(\tilde{A}_j(L,L)\), \(\tilde{A}_j(L,R)\), \(\tilde{A}_j(R,L)\), and \(\tilde{A}_j(R,R)\). Thus, to obtain the claim, it is sufficient to show that $v_i^-(A_j)\leq \tilde{v}_i(\tilde{A}_j(Y,Z))\leq v_i^+(A_j)$ for any $Y,Z\in \{L,R\}$. 

At this point, the remainder of the proof is straightforward but tedious, as it requires verifying the claim by systematically analyzing all nine cases and their respective sub-cases in the rounding procedure. In particular, for each of the nine cases, one must verify that inequality $v_i^-(A_j)\leq \tilde{v}_i(\tilde{A}_j(Y,Z))\leq v_i^+(A_j)$ holds for any $Y,Z\in \{L,R\}$.

%, that is, there are two subsets $B,C\subseteq A_j$ such that obtained by removing at most one item $x\in \partial A_j$ from $\partial A_j$, such that $v_i^-(A_j)\leq v_i(B)\leq \tilde{v}_i(\tilde{A}_j(Y,Z))\leq v_i(C)\leq v_i^+(A_j)$ (where the first and the last inequality hold by the definitions of of $v_i^-$ and $v_i^+$, respectively). 

As an example, we describe the analysis of cases 2(i), 2(ii), 6(i) and 6(ii), since the other cases are less complex and can be automatically verified by following analogous and simpler arguments. 

We recall that the virtual value $\tilde{v}_i([a,b])$ of a fractional bundle $[a,b]$ corresponds to one of the following three cases:
\begin{itemize}
\item {\em left-value (LV):}  if $a\equiv 0$, $\tilde{v}_i([a,b]):=v_i(\lpar a^-,b^-\rpar)$; 
\item {\em borderline-value (BV):} if $a\equiv 1$, $\tilde{v}_i([a,b])$ is set equal to the middle value among $v_i(\lpar a^-,b^-\rpar)$, $v_i(\lpar a^+,b^-\rpar)$ and $v_i(\lpar a^+,b^+\rpar)$; 
\item {\em right-value (RV):} if $a\equiv 2$, $\tilde{v}_i([a,b]):=v_i(\lpar a^+,b^-\rpar)$. 
\end{itemize}
Furthermore, we recall that the $j$-th bundle $\tilde{A}_j$ of the main allocation $\mathcal{\tilde{A}}$ is {\em left-first} (resp.,  {\em right-first}) in $\Delta^*$ if, among the two knives determining the endpoints of the $j$-th bundle across all allocations associated with $\Delta^*$, the first to move from left to right is the left (resp., right) one.

In the following, we proceed with the analysis of Cases 2(i), 2(ii), 6(i) and 6(ii):
\begin{itemize}
\item {\em Case 2(i)}: In this case, we have $b_j\in A_{j+1}$ and we set $A_j\gets \lpar a_j^-,b_j^-\rpar$. 

First, let us focus  on fractional bundle $\tilde{A}_j(L,Z)$ for any $Z\in \{L,R\}$; we observe that the virtual valuation applied to $\tilde{A}_j(L,Z)$ is left-value (LV), that is, $\tilde{v}_i(\tilde{A}_j(L,Z))$ is equal to $v_i(\lpar a_j^-,b_j^-\rpar)$. Thus, we have $v_i^-(A_j)\leq v_i(\lpar a_j^-,b_j^-\rpar)=\tilde{v}_i(\tilde{A}_j(L,Z))=v_i(\lpar a_j^-,b_j^-\rpar) \leq v_i^+(A_j)$, where the first (resp., the last) inequality holds since $v_i^-(A_j)$ (resp., $v_i^+(A_j)$) is the highest (resp., lowest) valuation obtainable by removing at most one boundary item, or none at all, from $A_j$.

Now, we consider a fractional bundle of type $\tilde{A}_j(R,Z)$ with $Z\in \{L,R\}$; the virtual valuation applied to $\tilde{A}_j(R,Z)$ is borderline-value (BV), that is, $\tilde{v}_i(\tilde{A}_j(R,Z))$ is equal to the middle value among $v_i(\lpar a_j^-,b_j^-\rpar)$, $v_i(\lpar a_j^+,b_j^-\rpar)$ and $v_i(\lpar a_j^+,b_j^+\rpar)$. Thus, we have $v_i^-(A_j)\leq \min\{v_i(\lpar a_j^-,b_j^-\rpar),v_i(\lpar a_j^+,b_j^-\rpar )\}\leq \tilde{v}_i(\tilde{A}_j(R,Z))\leq \max\{v_i(\lpar a_j^-,b_j^-\rpar),v_i(\lpar a_j^+,b_j^-\rpar )\} \leq v_i^+(A_j)$, where: the first (resp., the last) inequality holds again since $v_i^-(A_j)$ (resp., $v_i^+(A_j)$) is the highest (resp., lowest) valuation obtainable by removing at most one boundary item, or none at all, from $A_j$; the second (resp., third) inequality holds since the minimum (resp., maximum) among the two values $v_i((a_j^-,b_j^-))$ and $v_i((a_j^+,b_j^-))$ must be less than or equal to (resp., greater than or equal to) the middle value among the three values $v_i((a_j^-,b_j^-))$, $v_i((a_j^+,b_j^-))$ and $v_i((a_j^+,b_j^+))$.

We conclude that $v_i^-(A_j)\leq \tilde{v}_i(\tilde{A}_j(Y,Z))\leq v_i^+(A_j)$ holds for any $Y,Z\in \{L,R\}$, and this shows the claim in case~2(i). 

\item {\em Case 2(ii):} In this case, we have $b_j\notin A_{j+1}$ and we set $A_j=\lpar a_j^-,b_j^+\rpar$. 

First, let us focus  on fractional bundle $\tilde{A}_j(L,Z)$ for any $Z\in \{L,R\}$; we observe that the virtual valuation applied to $\tilde{A}_j(L,Z)$ is left-value (LV), that is, $\tilde{v}_i(\tilde{A}_j(L,Z))$ is equal to $v_i(\lpar a_j^-,b_j^-\rpar)$. Thus, we have $v_i^-(A_j)\leq \min\{v_i(\lpar a_j^-,b_j^+\rpar),v_i(\lpar a_j^-,b_j^-\rpar )\}\leq \tilde{v}_i(\tilde{A}_j(L,Z))\leq \max\{v_i(\lpar a_j^-,b_j^+\rpar),v_i(\lpar a_j^-,b_j^-\rpar )\} \leq v_i^+(A_j)$, where the first and the last inequality holds for the usual reasons, while the second and third inequalities hold since $\tilde{v}_i(\tilde{A}_j(L,Z))=v_i(\lpar a_j^-,b_j^-\rpar)$. 

Now, we consider each fractional bundle $\tilde{A}_j(R,Z)$ with $Z\in \{L,R\}$; the virtual valuation applied to $\tilde{A}_j(R,Z)$ is borderline-value (BV), that is, $\tilde{v}_i(\tilde{A}_j(R,Z))$ is equal to the middle value among $v_i(\lpar a_j^-,b_j^-\rpar)$, $v_i(\lpar a_j^+,b_j^-\rpar)$ and $v_i(\lpar a_j^+,b_j^+\rpar)$. Thus, we have $v_i^-(A_j)\leq \min\{v_i(\lpar a_j^-,b_j^-\rpar ),v_i(\lpar a_j^+,b_j^+\rpar )\}\leq \tilde{v}_i(\tilde{A}_j(R,Z))\leq \max\{v_i(\lpar a_j^-,b_j^-\rpar ),v_i(\lpar a_j^+,b_j^+\rpar )\} \leq v_i^+(A_j)$, where the first and the last inequality holds for the usual reasons, while, similarly to case 2(i), the second (resp., third) inequality holds since the minimum (resp., maximum) among the two values $v_i((a_j^-,b_j^-))$ and $v_i((a_j^+,b_j^+))$ must be less than or equal to (resp., greater than or equal to) the middle value among the three values $v_i((a_j^-,b_j^-))$, $v_i((a_j^+,b_j^-))$ and $v_i((a_j^+,b_j^+))$.

We conclude that $v_i^-(A_j)\leq \tilde{v}_i(\tilde{A}_j(Y,Z))\leq v_i^+(A_j)$ holds for any $Y,Z\in \{L,R\}$, and this shows the claim in case~2(ii).

\item {\em Case 6(i)}: In this case, we have $A_j=\lpar a_j^+,b_j^+\rpar$. As the fractional bundle $[a_j,b_j]$ is left-first in $\Delta^*$, the $j$-th bundle will never be equal to $\tilde{A}_j(L,R)$ over all allocations associated with vertices in $\Delta^*$. Thus, it is sufficient showing that $v_i^-(A_j)\leq \tilde{v}_i(\tilde{A}_j(Y,Z))\leq v_i^+(A_j)$ holds for any $(Y,Z)\in \{(L,L),(R,L),(R,R)\}$. 

First, let us focus  on fractional bundle $\tilde{A}_j(L,L)$; the virtual valuation applied to $\tilde{A}_j(L,L)$ is borderline-value (BV), that is, $\tilde{v}_i(\tilde{A}_j(L,L))$ is equal to the middle value among $v_i(\lpar a_j^-,b_j^-\rpar)$, $v_i(\lpar a_j^+,b_j^-\rpar)$ and $v_i(\lpar a_j^+,b_j^+\rpar)$. Thus, for the usual reasons, we have $v_i^-(A_j)\leq \min\{v_i(\lpar a_j^+,b_j^-\rpar ),v_i(\lpar a_j^+,b_j^+\rpar )\}\leq \tilde{v}_i(\tilde{A}_j(L,L))\leq \max\{v_i(\lpar a_j^+,b_j^-\rpar ),v_i(\lpar a_j^+,b_j^+\rpar )\} \leq v_i^+(A_j)$. 

Now, we consider bundle $\tilde{A}_j(R,L)$; the virtual valuation applied to $\tilde{A}_j(R,L)$ is right-value (RV), that is, $\tilde{v}_i(\tilde{A}_j(R,Z))$ is equal to $v_i(\lpar a_j^+,b_j^-\rpar)$. Thus, we have $v_i^-(A_j)\leq \min\{v_i(\lpar a_j^+,b_j^-\rpar ),v_i(\lpar a_j^+,b_j^+\rpar )\}\leq \tilde{v}_i(\tilde{A}_j(R,Z))\leq \max\{v_i(\lpar a_j^+,b_j^-\rpar ),v_i(\lpar a_j^+,b_j^+\rpar )\} \leq v_i^+(A_j)$, where the first and the second inequalities hold for the usual reasons, and the other inequalities hold since $\tilde{v}_i(\tilde{A}_j(R,Z))=v_i(\lpar a_j^+,b_j^-\rpar)$.

Finally, we consider bundle $\tilde{A}_j(R,R)$; the virtual valuation applied to $\tilde{A}_j(R,R)$ is again right-value (RV), but now the right knife has position $b^+$, and then $\tilde{v}_i(\tilde{A}_j(R,R))=v_i(\lpar a_j^+,b_j^+\rpar)$. Thus, for the usual reasons, we have $v_i^-(A_j)\leq v_i(\lpar a_j^+,b_j^+\rpar )=\tilde{v}_i(\tilde{A}_j(R,R))=v_i(\lpar a_j^+,b_j^+\rpar )\leq v_i^+(A_j)$. 

We conclude that $v_i^-(A_j)\leq \tilde{v}_i(\tilde{A}_j(Y,Z))\leq v_i^+(A_j)$ holds for any $(Y,Z)\in \{(L,L),(R,L),(R,R)\}$, and this shows the claim in case 6(i). 

\item {\em Case 6(ii)}: In this case, we have $A_j=\lpar a_j^-,b_j^+\rpar$. As the fractional bundle $[a_j,b_j]$ is right-first in $\Delta^*$, the $j$-th bundle will never be equal to $\tilde{A}_j(R,L)$ over all allocations associated with vertices in $\Delta^*$. Thus, it is sufficient showing that $v_i^-(A_j)\leq \tilde{v}_i(\tilde{A}_j(Y,Z))\leq v_i^+(A_j)$ holds for any $(Y,Z)\in \{(L,L),(L,R),(R,R)\}$. 

First, let us focus on fractional bundle $\tilde{A}_j(L,L)$; the virtual valuation applied to $\tilde{A}_j(L,L)$ is borderline-value (BV), that is, $\tilde{v}_i(\tilde{A}_j(L,L))$ is equal to the middle value among $v_i(\lpar a_j^-,b_j^-\rpar)$, $v_i(\lpar a_j^+,b_j^-\rpar)$ and $v_i(\lpar a_j^+,b_j^+\rpar)$. Thus, for the usual reasons, we have $v_i^-(A_j)\leq \min\{v_i(\lpar a_j^-,b_j^-\rpar ),v_i(\lpar a_j^+,b_j^+\rpar )\}\leq \tilde{v}_i(\tilde{A}_j(L,L))\leq \max\{v_i(\lpar a_j^-,b_j^-\rpar ),v_i(\lpar a_j^+,b_j^+\rpar )\} \leq v_i^+(A_j)$. 

Now, we consider bundle $\tilde{A}_j(L,R)$; the virtual valuation applied to $\tilde{A}_j(R,L)$ is again borderline-value (BV), but the right knife is now placed in position $b^+$, that is, $\tilde{v}_i(\tilde{A}_j(L,R))$ is equal to the middle value among $v_i(\lpar a_j^-,b_j^+\rpar)$, $v_i(\lpar a_j^+,b_j^+\rpar)$ and $v_i(\lpar a_j^+,(b_j+1)^+\rpar)$. Thus, for the usual reasons, we have  $v_i^-(A_j)\leq \min\{v_i(\lpar a_j^-,b_j^+\rpar ),v_i(\lpar a_j^+,b_j^+\rpar )\}\leq \tilde{v}_i(\tilde{A}_j(L,L))\leq \max\{v_i(\lpar a_j^-,b_j^+\rpar ),v_i(\lpar a_j^+,b_j^+\rpar )\} \leq v_i^+(A_j)$. 

Finally, we consider bundle $\tilde{A}_j(R,R)$; the virtual valuation applied to $\tilde{A}_j(R,L)$ is right-value (RV) and the right knife is again in position $b^+$, that is, $\tilde{v}_i(\tilde{A}_j(R,R))=v_i(\lpar a_j^+,b_j^+\rpar)$. Thus, for the usual reasons, we have $v_i^-(A_j)\leq \min\{v_i(\lpar a_j^-,b_j^+\rpar ),v_i(\lpar a_j^+,b_j^+\rpar )\}\leq \tilde{v}_i(\tilde{A}_j(R,R))\leq \max\{v_i(\lpar a_j^-,b_j^+\rpar ),v_i(\lpar a_j^+,b_j^+\rpar )\} \leq v_i^+(A_j)$. 

We conclude that $v_i^-(A_j)\leq \tilde{v}_i(\tilde{A}_j(Y,Z))\leq v_i^+(A_j)$ holds for any $(Y,Z)\in \{(L,L),(L,R),(R,R)\}$, and this shows the claim in case 6(ii). 

\end{itemize}

\subsection{Full Details on Efficient Computation of $\EQOP$ Allocations (Theorem \ref{thm2_eq1p_app}).}
Full details on design, correctness and complexity of the dynamic programming algorithm introduced in Section \ref{sec:notwholly} are provided in the proof of Theorem \ref{thm2_eq1p_app} below (for the pseudocode, see Algorithm \ref{alg4} of Appendix \ref{app:algorithms}). 

We say that a connected allocation $\mathcal{A}$ is {\em well-ordered} if the $i$-th leftmost bundle is assigned to agent $i$, for any $i\in [n]$; we note the allocation $\mathcal{A}$ constructed in Theorem \ref{thm1_eq1p} to show the existence of $\EQOP$ allocations was well-ordered. 
In this section, we again restrict our attention to well-ordered allocations, and  this assumption is without loss of generality (see Remark \ref{rema}). Let $C_v$ be the set of the valuations $v_i(S)$ that each agent $i$ has for any bundle $S$, and, given $c\in C_v$, we say that an allocation $\mathcal{A}=(A_1,\ldots, A_n)$ is {\em $c$-feasible} if, for each $i\in [n]$, $v_i^+(A_i(c))\geq c\geq v_i^-(A_j(c))$ holds.
\begin{lemma}\label{lem_eq1p_comp}
$\mathcal{A}$ is a well-ordered $c$-feasible allocation for some $c\in C_v$ $\Leftrightarrow$ $\mathcal{A}$ is a well-ordered $\EQOP$ allocation.
\end{lemma}

\begin{proof}
$\Rightarrow$): Let $\mathcal{A}$ be a $c$-feasible well-ordered allocation for some $c\in C_v$. Then, $v_i^+(A_i)\geq c\geq v_j^-(A_j)$ holds for any $i,j\in N$, that is, $\mathcal{A}$ is $\EQOP$.

$\Leftarrow$): Given an $\EQOP$ well-ordered allocation $\mathcal{A}$, let $c^+:=\min_{j\in N}v_j^+(A_j)$ and $c^-:=\max_{j\in N}v_j^-(A_j)$. Then, we have that $v_i^+(A_i)\geq c^+\geq c^-\geq v_i^-(A_i)$ holds for any $i\in N$, where the intermediate inequality follows from the $\EQOP$ guarantee. Let $j$ be the index minimizing $v_j^+(A_j)$, that is, $c^+=v_j^+(A_j)$. We observe that $v^+_j(A_j)=v_j(S)$, for some bundle $S$ obtained from $A_j$ by possibly removing one item from the board of $A_j$. Thus, by setting $c:=c^+$, we have $c\in C_v$. Then, by the previous inequalities, we have $v_i^+(A_i)\geq c\geq v_i^-(A_i)$ for any $i\in [n]$, that is, $\mathcal{A}$ is a well-ordered $c$-feasible allocation with $c\in C_v$.
\end{proof}

Then, we obtain Theorem \ref{thm2_eq1p}, that is restated and proven below:
\begin{theorem}[Claim of Theorem \ref{thm2_eq1p}]\label{thm2_eq1p_app}
Under non-negative valuations, an $\EQOP$ allocation can be found in polynomial time $O(n^2m^4)$.
\end{theorem}
\begin{proof}
By Lemma~\ref{lem_eq1p_comp}, to find an $\EQOP$ allocation \( \mathcal{A} \), it is sufficient to check, for each \( c \in C_v \), whether there exists a well-ordered $c$-feasible allocation. Furthermore, the existence of a well-ordered $\EQOP$ allocation was guaranteed by Theorem~\ref{thm1_eq1p} (whose proof, indeed, yielded a well-ordered allocation), so it is certain that such a \( c \in C_v \) can be found. 

For any \( c \in C_v \), the problem of determining whether there exists a well-ordered $c$-feasible allocation and, if so, finding such an allocation can be solved using a dynamic programming algorithm in polynomial time \( O(n m^2) \). Indeed, one can iteratively solve the sub-problem \( B_c[i][j] \) which determines whether there exists a partial well-ordered $c$-feasible allocation that allocates the first \( j \) items to the first \( i \) agents.

Specifically, for any \( i \in [n] \) and \( j \in [m] \cup \{0\} \), let \( B_c[i][j] \) be the boolean value that is TRUE (resp. FALSE) if there exists a (resp. if there is no) partial $c$-feasible well-ordered allocation $\mathcal{A}$ of the first \( j \) items to the first \( i \) agents. The solution to the original problem is then obtained by evaluating \( B_c[n][m] \). The values \( B_c[i][j] \) satisfy the following recurrence relation:
\begin{itemize}
\item {\em Base Case} (\( i = 1 \), \( j \in [m] \cup \{0\} \)):  
  \( B_c[1][j] = \text{TRUE} \) if the bundle $\lpar 1, j\rpar$ satisfies \( v_1^+(\lpar 1, j\rpar) \geq c \geq v_1^-(\lpar 1, j\rpar) \);  
  otherwise, \( B_c[1][j] = \text{FALSE} \).
  \item {\em Recursive Step} (\( i \in [n] \setminus \{1\} \), \( j \in [m] \cup \{0\} \)):  
  \( B_c[i][j] = \text{TRUE} \) if there exists \( \ell \in [j+1] \) such that the bundle $\lpar \ell, j\rpar$ satisfies both  
  \( v_i^+(\lpar \ell, j\rpar) \geq c \geq v_i^-(\lpar \ell, j\rpar\)  
  and \( B_c[i-1][\ell-1] = \text{TRUE} \);  
  otherwise, \( B_c[i][j] = \text{FALSE} \).
\end{itemize}

The above recurrence relation for \( B_c[i][j] \) can be exploited by a dynamic-programming algorithm to compute all the values \( B_c[i][j] \), for \( i \in [n]\) and \( j \in [m] \cup \{0\} \), in \( O(nm^2) \) time (i.e., \( O(m) \) to compute each cell \( B_c[i][j] \), multiplied by the total number of cells, $O(mn)$)\footnote{For $i = n$, it is sufficient to find the value $B_c[n][j]$ for $j = m$ only, since the bundle assigned to agent $n$ in the complete allocation is either empty or must include the last item $m$. Nonetheless, for the sake of simplicity, we compute all values $B_c[n][j]$, as this does not impact the asymptotic time complexity.}. Thus, the overall execution of this dynamic-programming approach for all \( c \in C_v \) requires \( O((nm^2) \cdot (nm^2)) = O(n^2 m^4) \) steps.

A pseudo-code of the procedure outlined above is given in Algorithm \ref{alg4} of Appendix \ref{app:algorithms}.
\end{proof}

We observe that, under identical valuations, the cardinality of the set \( C_v \) is \( O(m^2) \), as the value of each connected bundle does not depend on the agent evaluating it. 
In this case, the complexity of the algorithm outlined in the proof of Theorem~\ref{thm2_eq1p} decreases by a factor of \( n \). Moreover, since all the considered approximate envy-freeness and equitability guarantees coincide under identical valuations, we obtain the following corollary as a direct consequence of Theorems \ref{thm1_eq1p} and \ref{thm2_eq1p_app} (restated above as Theorem \ref{thm2_eq1p_app}).
\begin{corollary}\label{coro2_eq1p}
Under identical non-negative valuations, $\EQOP$ (or, equivalently, $\EFOP$) allocations always exist and can be efficiently computed in polynomial time \( O(n^2m^3) \).
\end{corollary}

\begin{remark}\label{rema}
We observe that the $\EQOP$ allocations \( \mathcal{A} \) considered in Theorem \ref{thm2_eq1p} (as well as that of Theorem \ref{thm1_eq1p}) is well-ordered, that is, it assigns the left-most bundle to agent 1, the second left-most bundle to agent~2, and so forth. Anyway, the existence and computation of $\EQOP$ allocations can be easily generalized to accommodate any fixed ordering of the agents, that is, we can decide in advance which agent is assigned the $i$-th leftmost bundle. Indeed, by Theorem \ref{thm2_eq1p}, for any fixed ordering $\rho:[n]\rightarrow [n]$ of the agents, we can compute in polynomial time another $\EQOP$ allocation that is well-ordered w.r.t. the new agents order, that is, it assigns the left-most bundle to agent $\rho(1)$, the second left-most bundle to agent $\rho(2)$, and so forth.
\end{remark}
\newpage 
\section{Missing Proofs from Section \ref{sec:notwholly}:\\part on envy-freeness}\label{app:ef1}
\subsection{Full details on the proof of Theorem \ref{thm3_ef1p}}
To show the existence of an $\EFOP$ allocation, we employ the same framework as in the approximate equitability case, with minor modifications. Specifically, we use the same triangulation $ T $ as in the previous case but equip it with $ n $ distinct coloring functions $ L_1, \ldots, L_n $, instead of the single coloring function $ L $ used earlier. 

Here, each vertex $\mathbf{x} \in V(T)$ corresponds to a partition $\mathcal{\tilde{S}}(\mathbf{x}) = (\tilde{S}_1(\bfx), \ldots, \tilde{S}_n(\bfx))$ of the interval $[0, m]$, referred to as a \emph{fractional connected partition}. This partition divides the interval into $n$ connected sub-intervals (fractional connected bundles), arranged sequentially from the left-most to the right-most. Each fractional bundle is temporarily unassigned to any agent, unlike in the equitability case (where the $ i $-th left-most fractional bundle was automatically assigned to agent $ i $).

Let $ \tilde{v}_i $ be the virtual valuation defined as in the equitability case. For any $ i \in [n] $, consider a distinct coloring function $ L_i $ that assigns to each vertex $ \bfx \in V(T) $ the index $ j $ of the bundle $ \tilde{S}_j(\bfx) $ in $ \mathcal{\tilde{S}}(\bfx) $ that maximizes the virtual valuation $ \tilde{v}_i(\tilde{S}_j(\bfx)) $ of agent $ i $, with ties broken in favor of non-empty bundles and, in case of further ties, arbitrarily. Also in this case, each $ L_i $ is a special coloring function (since the empty bundle is the least valuable for any agent).

By the Generalized Sperner's Lemma (Theorem~\ref{lem:sperner_gen}), there exists at least one jointly fully-colored elementary $(n-1)$-simplex $\Delta^* = \text{conv}(\bfx_1^*, \ldots, \bfx_n^*) \in T$ under the coloring functions $L_1, \ldots, L_n$. Consequently, by exploiting the construction of these coloring functions, for any $i\in [n]$ and for some permutations $\sigma,\tau: [n] \to [n]$ (independent on $i$), the fractional bundle indexed by $\tau(i)$ is the most valuable bundle for agent~$i$ in the fractional partition $\mathcal{\tilde{S}}(\bfx_{\sigma(i)}^*)$ derived from  the ${\sigma(i)}$-th vertex $\bfx_{\sigma(i)}^*$ of $\Delta^*$, under virtual valuation $\tilde{v}_i$. Starting from the fractional partition $\mathcal{\tilde{S}}:=\mathcal{\tilde{S}}(\bfx_1^*)$ associated with the first vertex $\bfx_1^*$ of $\Delta^*$, we use the same rounding procedure as in the equitability case, to obtain a partition $\mathcal{S}=(S_1,\ldots, S_n)$ of $[m]$ made of integral connected bundles, referred to as {\em connected integral partition}. The claims of Lemma \ref{lem1_eq1p} and \ref{lem2_eq1p} obviously apply to $\mathcal{\tilde{S}}$, and can be restated as follows:
\begin{lemma}\label{lem1_ef1p}
$\mathcal{S}$ is a connected integral partition.
\end{lemma}
\begin{lemma}\label{lem2_ef1p}
For any $h,i,j\in [n]$, the connected (integral) partition $\mathcal{S}$ satisfies $v_i^-({S}_j)\leq \tilde{v}_i(\tilde{S}_j(\bfx^*_h))\leq v_i^+({S}_j)$.
\end{lemma}
Now, we can determine the final connected (integral) allocation $ \mathcal{A} = (A_1, \ldots, A_n) $ by assigning, for each $ i \in [n] $, the bundle $ S_{\tau(i)} $ from the connected partition $ \mathcal{S} $ to agent $ i $ (i.e., $ A_i := S_{\tau(i)} $). We can then prove Theorem \ref{thm3_ef1p}, restated below for completeness:
\begin{theorem}[Statement of Theorem \ref{thm3_ef1p}]
Under non-negative valuations, $\mathcal{A}$ is an $\EFOP$ allocation.
\end{theorem}
\begin{proof}
First of all, $ \mathcal{A} $ is a connected allocation by Lemma~\ref{lem1_ef1p}, as it is obtained by permuting the bundles of $\mathcal{S}$ trough permutation $\tau$ ($A_i=S_{\tau(i)}$ for any $i\in [n]$). Now, we show the $\EFOP$ guarantee. 
By the construction of the jointly fully-colored simplex $ \Delta^* = \text{conv}(\bfx_1^*, \ldots, \bfx_n^*) $ and the coloring functions $ L_1, \ldots, L_n $, we have that  $\tau(i)$ is one of the indices $j$ that maximize the valuation $ \tilde{v}_i(\tilde{S}_j(\bfx_{\sigma(i)}^*)) $, for any $ i \in [n] $ (i.e., the $ \tau(i) $-th bundle of partition $ \mathcal{\tilde{S}}(\bfx_{\sigma(i)}^*) $ maximizes the valuation of agent $ i $ among all bundles of that partition). Thus, for any $ i, j \in N $, we have 
$
v_i^+(A_i)=v_i^+(S_{\tau(i)}) \geq \tilde{v}_i(\tilde{S}_{\tau(i)}(\bfx_{\sigma(i)}^*)) \geq \tilde{v}_i(\tilde{S}_{\tau(j)}(\bfx_{\sigma(i)}^*))\geq v_i^-(S_{\tau(j)})= v_i^-(A_j),
$
where the second inequality follows from the above observation, and the first and last inequalities follow from Lemma~\ref{lem2_eq1p}. Since $ v_i^+(A_i) \geq v_i^-(A_j) $ for any $ i, j \in N $, we conclude that $ \mathcal{A} $ satisfies the $\EFOP$ condition, and thus the claim holds.
\end{proof}
\begin{remark}
As noted in Remark~\ref{rema}, in the case of equitability one can decide in advance which agent is assigned the $i$-th leftmost bundle. However, unlike in that case, for $\EFOP$ allocations this assignment cannot be predetermined, since the ordering of agents to whom the bundles are assigned from left to right is determined by the jointly fully-colored simplex returned by the Generalized Sperner’s Lemma in the proof of Theorem \ref{thm3_ef1p}. This constitutes the main obstacle preventing the extension of the dynamic programming approach used for equitability in Theorem~\ref{thm2_eq1p} to compute $\EFOP$ allocations. Indeed, such an approach would first require knowing the ordering of agents to whom the bundles are assigned from left to right in the future $\EFOP$ allocation, and then applying a similar computational method as in Theorem~\ref{thm2_eq1p} to obtain such an $\EFOP$ allocation. With this respect, we point out that, differently from the case of equitability, it is not guaranteed that, for any agents ordering to whom the bundles are assigned from left to right, there exists an $\EFOP$ allocation consistent with that assignment.
\end{remark}

\newpage\ \newpage

\section{Missing Proofs from Section \ref{sec:multisperner}}\label{app:non_positive}
\subsection{Proof of Theorem \ref{multiSperner} (Multi-coloring Sperner's Lemma)}
We first assume, without loss of generality, that for any vertex \( \bfx \) located on the boundary of \( \Delta \), the condition \( \LL(\bfx) = \{i \in [n] : \bfx \in F_i\} \) holds, and that \( | \LL(\bfx) | = 1 \) for any internal vertex \( \bfx \) not located on the boundary. Indeed, if the claim holds under this restriction, it also holds in the more general case where \( \LL(\bfx) \supseteq \{i \in [n] : \bfx \in F_i\} \) for boundary vertices \( \bfx \) and \( | \LL(\bfx) | \geq 1 \) for the other vertices.

Let \( L \) be the standard coloring function associated with the multi-coloring one $\LL$, that assigns to each vertex \( \bfx \in V(T) \) the color  
$L(\bfx) = \min\{i\in \LL(\bfx)\}$;  we refer to $L$ as the {\em minimal restriction} of $\LL$.
We will show, by induction on $n\geq 2$, that the number of fully-colored simplices with respect to the minimal restriction \( L \) is odd, which implies the existence of at least one such simplex. By the construction of \( L \), any fully-colored simplex with respect to \( L \) is also fully-colored with respect to the original multi-coloring function \( \LL \), thereby proving the claim of the theorem.  

{\em Base Case ($n=2$):} 
In this case, \( \Delta \) is the \( 1 \)-simplex \( [\vv_1, \vv_2] \), whose \( 0 \)-dimensional faces are the two endpoint vertices \( \vv_1 \) and \( \vv_2 \). The triangulation \( T \) can be viewed as a path consisting of vertices \( V(T) \), connected by edges that correspond to the (contiguous) elementary \( 1 \)-simplices $[\bfx,\mathbf{y}]$ of \( T \). We have \( L(\bfx) \in [2] \) for every \( \bfx \in V(T) \). Moreover, since \( \LL \) is special, and $\vv_1$ (resp., $\vv_2$) represents the $(n-2)$-face opposite to $\vv_2$ (resp., $\vv_1$), it follows that \( L(\vv_1) = 2 \) (resp., \( L(\vv_2) = 1 \)). Thus, since $L$ assigns colors in $\{1,2\}$ and the colors assigned by $L$ to the endpoints $\vv_1$ and $\vv_2$ are different, we immediately have that the number of elementary $1$-simplices $[\bfx,\mathbf{y}]$ with $L(\bfx)\neq L(\mathbf{y})$ (i.e., those which are fully-colored w.r.t. $L$) must be odd. 

{\em Inductive Step:} Assume that the claim holds for $k=n-1$, and let us show it for $k=n$. Let $F'=\bigcup_{i\in [n-1]}F_i$ be the union of all the $(n-2)$-faces of $\Delta$, except for the $n$-th face. The remainder of the proof proceeds with the following two sub-steps:
\begin{itemize}
\item {\em Sub-step 1 (Recovering the inductive hypothesis)}: We have that the triangulation \( T' \) induced on \( F' \) by \( T \) is topologically equivalent\footnote{Given two topological spaces \( X, Y \), a {\em homeomorphism} \( f: X \to Y \) is a bijective continuous function whose inverse is also continuous.  Two triangulations \( T_X \) and \( T_Y \) of topological spaces \( X \) and \( Y \), respectively, are {\em topologically equivalent} if there exists a homeomorphism \( f: X \to Y \) that maps each simplex of \( T_X \) to a simplex of \( T_Y \), preserving the adjacency relation, i.e., $ f(\Delta_1 \cap \Delta_2) = f(\Delta_1) \cap f(\Delta_2)$ for any simplices  $\Delta_1, \Delta_2 \in T_X.$} to the triangulation of an \( (n-2) \)-simplex ( $\Delta' = \text{conv}(\vv_1', \ldots, \vv_{n-1}')$. This is established via a homeomprphism \( f \) that maps each vertex \( \vv_i \) to \( \vv_i' \) for \( i \in [n-1] \), sends the vertex \( \vv_n \) into the interior of \( \Delta' \), and maps each \( (n-3) \)-face \( F_i \cap F_n \) to an \( (n-3) \)-face of \( \Delta' \). 
In particular, such a homeomorphism \( f \) can be obtained by projecting each point located on $F'$ onto the hyperplane containing the face \( F_n \), parallely to the axis connecting vertex $\vv_n$ to the barycenter $\vv_n'$ of $F_n$. See Figure \ref{fig:pir} for a visualization of the projection \( f \) in the case \( n = 4 \).

Thus, we can topologically regard \( F' \) as an \( (n-2) \)-simplex \( \Delta' \) and interpret \( T' \) as its associated triangulation.  

Let \( \LL' : V(T') \to 2^{[n-1]} \setminus \{\emptyset\} \) be the multi-coloring function that assigns to each vertex \( \bfx \in V(T') \) the set \( \LL(\bfx) \setminus \{n\} \), that is, the restriction of \( \LL \) to the first \( n-1 \) colors. By  relying on the fact that $\LL$ is a special multi-coloring function for $T$, we have that \( \LL' \) is also a special multi-coloring function for \( T' \), and its minimal restriction \( L' \) coincides with the minimal restriction \( L \) of $\LL$, when restricted to vertices in \( V(T') \) (i.e., located on $F'$). Thus, we can apply the inductive hypothesis to the triangulation \(T'\) and the multi-coloring function \(L'\) on \(V(T')\), since \(\dim(\Delta') = n-1\). This allows us to conclude that the number of fully-colored elementary \((n-2)\)-simplices in \(T'\) with respect to the coloring function \(L'\) and using colors from \([n-1]\) is odd. Equivalently, we have that the number of fully-colored elementary \((n-2)\)-simplices in $F'$, with respect to the coloring function \(L\)  and using colors from $[n-1]$, is odd. 
 
\item {\em Sub-step 2 (From dimension $n-2$ to $n-1$ through a parity argument on graphs)}: Let \( G \) be the undirected graph whose nodes represent the elementary \( (n-1) \)-simplices of \( \Delta \), with edges connecting pairs of \( (n-1) \)-simplices that share an \( (n-2) \)-dimensional face that is fully-colored with respect to \( L \) using colors from \( [n-1] \).  
By leveraging the construction of the edges, we see that no node in \( G \) has a degree greater than 2. This means that \( G \) consists of a union of undirected paths or cycles (excluding isolated nodes). Consequently, the number of nodes with degree 1 in \( G \) must be even.  
Moreover, we observe that a node has degree \( 1 \) in \( G \) if and only if it corresponds to an elementary \((n-1)\)-simplex of one of the following types:  
 (i) an {\em almost} fully-colored \( (n-1) \)-simplex on the {\em boundary}, where each color \( i \in [n-1] \) appears exactly once, except for one color that appears twice, and that has a face touching \( F' \); (ii) a fully-colored elementary \( (n-1) \)-simplex (where each color $i\in [n]$ appears exactly once). The simplices of type (i) are bijectively associated with the elementary \( (n-2) \)-simplices established in Sub-step 1, and thus they appear in an odd number. Since the total number of nodes of degree 1 in \( G \) is even, it follows that the number of nodes of type (ii), i.e., the desired fully-colored \( (n-1) \)-simplices, must be odd, and this concludes the proof of the theorem.
\end{itemize}
The Multi-coloring Sperner's lemma, applied to an instance with \( n = 3 \), along with part of its proof (in particular, Sub-step 2), is illustrated in Figure \ref{fig:4} of Appendix~\ref{app:pictures}.

\subsection{Full Details on Existence of $\EQOP$ allocations.}
To show the existence of an $\EQOP$ allocation, we employ the same framework as in the case of non-negative valuations, with minor modifications. We use the same triangulation \( T \) as in that case, but we equip $T$ with the multi-coloring function $\LL$ that assign to each vertex $\bfx\in V(T)$ the set $\LL(\bfx)$ of indices $j$ which maximize $\tilde{v}_j(\tilde{A}_j(\bfx))$, where $\tilde{v}_j$ is the virtual valuation defined as in Section \ref{sec:notwholly}. 

Differently from the case of non-negative valuations, under non-positive valuations the empty bundle is always the best one for each agent. Thus, given $i\in [n]$, since each vertex $\bfx$ located on the $(n-2)$-face $F_i$ opposite to $\vv_i$ corresponds to  an allocation having its $i$-th bundle empty, we have that $i\in \LL(\bfx)$. Then, $\LL$ is a special multi-coloring function and, by the Multi-coloring Sperner's lemma (Theorem \ref{multiSperner}), there exists an elementary \((n-1)\)-simplex \(\Delta^* = \text{conv}(\bfx^*_1, \ldots, \bfx^*_n)\) that is fully-colored under the multi-coloring function \(\LL\). This means there exists a permutation \(\sigma: [n] \rightarrow [n]\) such that \(i \in \LL(\bfx^*_{\sigma(i)})\) for each \(i \in [n]\). Let \(\mathcal{\tilde{A}}(\bfx^*_1), \ldots, \mathcal{\tilde{A}}(\bfx^*_n)\) be the sequence of \(n\) fractional allocations determined from \(\Delta^*\) as in the case of non-negative valuations (see Section \ref{sec:notwholly}). By leveraging the construction of \(\LL\) and the above sequence of fractional allocations, we have that for each \(i \in [n]\), \(i\) is one of the indices \(j\)  that maximize \(\tilde{v}_j(\tilde{A}_j(\bfx^*_{\sigma(i)}))\) (that is, the $i$-th bundle of allocation  $\tilde{\mathcal{A}}(\bfx^*_{\sigma(i)})$ maximizes the virtual valuation $\tilde{v}_j$ among all bundles $\tilde{A}_j(\bfx^*_{\sigma(i)})$ of that allocation). Thus, by exploiting the same reasoning as in Theorem~\ref{thm1_eq1p}, we have that the integral connected allocation $\mathcal{A}$ obtained by rounding the first fractional allocation $\mathcal{\tilde{A}}(\bfx^*_1)$ as in Section \ref{sec:notwholly}, satisfies the desired approximate equitability guarantees.
\begin{theorem}\label{thm1_eq1p_nonpos}
$\mathcal{A}$ is an $\EQOP$ connected allocation, if the agents' valuations are non-positive.
\end{theorem}
Furthermore, by exploiting the same algorithmic framework as in Theorem \ref{thm2_eq1p},
we also obtain the following computational result. 
\begin{theorem}\label{thm2_eq1p_nonpos}
An $\EQOP$ allocation for instances with non-positive valuations can be found in polynomial time $O(n^2m^4)$.
\end{theorem}
If valuations are monotone non-increasing, there are chores only. Thus, the $\EQOP$ guarantee is equivalent to the stronger {\em equitability-up-to-one-chore-over-paths} (EQ1P), where each agent \( i \) can remove at most one outer chore \( c \) from her own bundle \( A_i \) to obtain a valuation at least as large as that achieved by each other agent. Then, we obtain the following corollary of Theorem \ref{thm2_eq1p_nonpos}:
\begin{corollary}
An EQ1P allocation for instances with monotone non-increasing valuations always exists and can be found in polynomial time. 
\end{corollary}
This existential and computational result complement the result of \cite{MSVV21}, that holds for monotone non-decreasing valuations only. 

Finally, analougsly to the case of non-negative valuations, we have the following corollary of Theorem \ref{thm2_eq1p_nonpos} and a similar remark to Remark \ref{rema}.

\begin{corollary}\label{coro2_eq1p_nonpos}
Under identical non-positive valuations, $\EQOP$ (or, equivalently, $\EFOP$) allocations always exist and can be efficiently computed in polynomial time \( O(n^2m^3) \).
\end{corollary}

\begin{remark}\label{rema_nonpos}
Under non-positive valuations, for any fixed ordering $\rho:[n]\rightarrow [n]$ of the agents, we can compute in polynomial time another $\EQOP$ allocation that is well-ordered w.r.t. the new agents order, that is, it assigns the left-most bundle to agent $\rho(1)$, the second left-most bundle to agent $\rho(2)$, and so forth.
\end{remark}

\subsection{Full Details on Existence of $\EFOP$ allocations.}

Before showing the existence of $\EFOP$ allocations, we first provide another multi-coloring variant of  Sperner's lemma, instantiated for the specific Kuhn's triangulation, but applied to $n$ multi-coloring functions, as in the Generalized Sperner's Lemma (Theorem \ref{lem:sperner_gen}). Specifically, let $\Delta=\{ \bfx = (x_1, \ldots, x_{n-1}) \in \mathbb{R}^{n-1} : 0 \leq x_1 \leq x_2 \leq \ldots \leq x_{n-1} \leq m \}$ be the $(n-1)$-simplex and $T$ be the related Kuhn's triangulation as defined in Section \ref{sec:notwholly}. Furthermore, let $\LL_1,\ldots, \LL_n$ be $n$ special multi-coloring functions of $V(T)$, as defined in the previous analysis of the $\EQOP$ guarantee. An elementary $(n-1)$-simplex $\Delta^*=\text{conv}(\bfx^*_1,\ldots, \bfx^*_n)\in T$ is {\em jointly fully-colored} under $\LL_1,\ldots, \LL_n$ if there exist two permutations $\sigma,\tau:[n]\rightarrow [n]$ such that $\tau(i)\in \LL_i(\bfx^*_{\sigma(i)})$ for any $i\in [n]$. 
\begin{theorem}[Generalized Multi-coloring Sperner's Lemma]\label{multiSperner_super}
Let $\LL_1, \dots, \LL_{n}$ be $n$ special multi-coloring functions of the Kuhn's triangulation $T$ of the $(n-1)$-simplex $\Delta$. Then, there exists a jointly-fully-colored elementary $(n-1)$-simplex $\Delta^* \in T$ under multi-coloring functions $\LL_1,\ldots, \LL_{n}$.
\end{theorem}
\begin{proof}

Given two vertices $\bfx$ and $\textbf{y}$, a {\em simple path} $P$ from $\bfx$ to $\textbf{y}$ in $T$ is a sequence of adjacent elementary $1$-simplices of $T$ connecting $\bfx$ and $\textbf{y}$, and its {\em length} is given by the number of such simplices.
Let $\Phi:V(T)\rightarrow [n]$ be an auxiliary standard coloring function such that $\Phi(\bfx)=1+(d(\bfx,\textbf{0}_{n-1})\text{ mod } n)$, where $d(\bfx,\textbf{y})$ is the length of the shortest path connecting $\bfx$ and $\textbf{y}$,  $\textbf{0}_{n-1}=(0,\ldots, 0)$ is the origin vertex, and ``mod $n$'' denote the remainder operator modulo $n$. Let $\overline{\LL}$ be a new {\em aggregated multi-coloring function} that assigns color $\overline{\LL}(\bfx):=\LL_{\Phi(\bfx)}(\bfx)$ to each vertex $\bfx\in V(T)$, where $\LL_1,\ldots, \LL_n$ are the input special multi-coloring functions. We observe that, since each $\LL_i$ is special, then $\overline{\LL}$ is special, too. Then, by relying on the standard Multi-coloring Sperner's lemma (Theorem \ref{multiSperner}), there exists an elementary $(n-1)$-simplex $\Delta^*=\text{conv}(\bfx_1^*,\ldots, \bfx^*_n)\in T$ that is fully-colored under $\overline{\LL}$, i.e., there exists a permutation $\pi:[n]\rightarrow [n]$ such that $j\in \overline{\LL}(\bfx_{\pi(j)}^*)$ for any $j\in [n]$. 

It remains to show that $\overline{\LL}$ refers to a distinct multi-coloring function $\LL_i$ at each distinct vertex $\bfx^*_{\sigma(i)}$, for some permutation $\sigma:[n]\rightarrow [n]$. Together with the fact that $\Delta^*$ is fully-colored w.r.t. $\overline{\LL}$, this will imply that $\Delta^*$ is jointly fully-colored w.r.t. the multi-coloring functions $\LL_1,\ldots, \LL_n$.

By construction of $\Phi$, we have that $\Phi$ assigns a distinct color in $[n]$ to each vertex of $\Delta^*$. Indeed, for any vertices $\bfx,\textbf{y}\in V(T)$, the distance $d(\bfx,\textbf{y})$ between $\bfx$ and $\textbf{y}$ (used to describe $\Phi$) can be equivalently defined as the scaled Manhattan distance\footnote{Given two vectors $(x_1,\ldots, x_k),(y_1,\ldots, y_k)\in \mathbb{R}^k$, the {\em Manhattan distance} between $\bfx$ and $\textbf{y}$ is given by $\sum_{h=1}^k|x_h-y_h|$.} $3\sum_{h=1}^{n-1}|x_h-y_h|$. Furthermore, by exploiting the structure of the elementary $(n-1)$-simplices in the Kuhn's triangulation \( T \) and the definition of Manhattan distance, we have that, given an elementary \((n-1)\)-simplex of \( T \), after an appropriate reordering of its vertices, each vertex is at distance 1 from the previous one according to the above scaled Manhattan distance, and the $i$-th-vertex in this ordering is at distance $i-1$ from the first vertex. Thus, $\Phi$ attains distinct values in $[n]$ at each of the $n$ vertices of $\Delta^*$, then allowing to define a permutation $\sigma:[n]\rightarrow [n]$ such that $\Phi(\bfx^*_{\sigma(i)})=i$ for any $i\in  [n-1]$. 

We recall that 
$j\in \overline{\LL}(\bfx^*_{\pi(j)})=\LL_{\Phi(\bfx^*_{\pi(j)})}(\bfx^*_{\pi(j)})$ holds for any $j\in [n]$. Let $\tau:=\pi^{-1}\circ \sigma$, where $\sigma$ is the permutation such that $\Phi(\bfx^*_{\sigma(i)})=i$ for any $i\in  [n-1]$. By using $\tau(i)$ in place of $j$ in the above membership relation involving $\overline{\LL}$, we have $\tau(i)\in \LL_{\Phi(\bfx^*_{\pi(\tau(j))})}(\bfx^*_{\pi(\tau(i))})=\LL_{\Phi(\bfx^*_{\sigma(i)})}(\bfx^*_{\sigma(i)})=\LL_{i}(\bfx^*_{\sigma(i)})$ for any $i\in [n]$, where the first equality holds by $\tau=\pi^{-1}\circ \sigma$ and the last one holds by $\Phi(x^*_{\sigma(i)})=i$. We conclude that $\Delta^*$ is jointly fully-colored under $\LL_1,\ldots, \LL_n$.
\end{proof}

To establish the existence of an $\EFOP$ allocation, we adopt the same framework as in the approximate equitability case, with minor modifications. Specifically, we use the same triangulation \( T \) as before but equip it with \( n \) distinct multi-coloring functions, \( \LL_1, \ldots, \LL_n \) (defined below), instead of the single multi-coloring function \( \LL \) used earlier.

As in the case of the approximate envy-freeness guarantee for non-negative valuations, each vertex \( \bfx \in V(T) \) now corresponds to a fractional connected partition \( \tilde{\mathcal{S}}(\bfx) = (\tilde{S}_1(\bfx), \ldots, \tilde{S}_n(\bfx)) \) of the interval \([0, m]\), where each fractional bundle is initially unassigned to any agent (unlike in the equitability case). For any \( i \in [n] \), we consider a distinct multi-coloring function \( \LL_i \) that assigns to each vertex \( \bfx \in V(T) \) the set \( \LL_i(\bfx) \) of indices \( j \) corresponding to the bundles \( \tilde{S}_j(\bfx) \) that maximize the valuation \( \tilde{v}_i(\tilde{S}_j(\bfx)) \) of agent \( i \), where \( \tilde{v}_i \) is the virtual valuation considered above. As in the case of equitability for non-positive valuations, each \( \LL_i \) is again a special multi-coloring function, since the empty bundle is the most valuable for any agent.

By Theorem~\ref{multiSperner_super} (Generalized Multi-coloring Sperner's lemma), there exists at least one jointly fully-colored elementary \((n-1)\)-simplex \(\Delta^* = \text{conv}(\bfx_1^*, \ldots, \bfx_n^*) \in T\) under the coloring functions \(\LL_1, \ldots, \LL_n\). Consequently, by exploiting the construction of these coloring functions, for some permutations \(\sigma,\tau:[n] \to [n]\) and for any $i\in [n]$, we have that $\tau(i)$ is one of the indices $j$ maximizing $\tilde{v}_i(\tilde{S}_j(\bfx^*_{\sigma(i)}))$ (that is, the fractional bundle indexed by \(\tau(i)\) is the most valuable bundle for agent~\(i\) in allocation $\mathcal{\tilde{S}}(\bfx^*_{\sigma(i)})$, under virtual valuation $\tilde{v}_i$). 
Then, we can apply the same rounding procedure as in the case of non-negative valuations (Section \ref{sec:notwholly}) to the partition \( \tilde{\mathcal{S}} \) to obtain a connected integral allocation \( \mathcal{A} = (A_1, \ldots, A_n) \) of \([m]\), obtained by assigning, for each $ i \in [n] $, the bundle $ S_{\tau(i)} $ from the connected partition $ \mathcal{S} $ to agent $ i $ (i.e., $ A_i := S_{\tau(i)} $). Then, we obtain the following theorem, analogous to Theorem \ref{thm3_ef1p}:

\begin{theorem}\label{thm3_ef1p_nonpos}
Under non-positive valuations, $\mathcal{A}$ is an $\EFOP$ allocation.
\end{theorem}

If valuations are monotone non-increasing, the $\EFOP$ guarantee is equivalent to the stronger {\em envy-freeness-up-to-one-chore-over-paths} (EF1P), where each agent \( i \) can remove at most one outer chore \( c \) from her own bundle \( A_i \) to obtain a valuation at least as large as that achieved for the other bundles. Then, we obtain the following corollary:
\begin{corollary}
An EF1P allocation for instances with monotone non-increasing valuations always exists. 
\end{corollary}
This result complements the findings of \citet{BCFIMPVZ22,I23}, that hold for monotone non-decreasing valuations only.

\newpage

\section{Objective Valuations}\label{sec:objective}
We recall that, under objective valuations, any item $x$ is either a good or a chore (independently on the considered agents and bundles); in such a case, we can partition \( M \) into a set of goods \( G \) and a set of chores \( C \), and we choose arbitrarily how to classify dummy items that qualify as both. 

For the remainder of this section, we assume w.l.o.g. that $m \geq n-1$. Indeed, the following remark explains how to handle the case $m<n-1$.
\begin{remark}\label{rema_compl_agents}
If the initial instance $I$ had $n-1 > m$, we could remove $n-m-1$ arbitrary players, thereby obtaining a new instance $I'$ with $n' := m+1 > 0$ players. Let $\mathcal{A}'$ be an allocation in $I'$ that satisfies one of the considered approximate equitability criteria.\footnote{We note that these arguments do not, in general, apply under exact or approximate envy-freeness.} We observe that the allocation $\mathcal{A}$ obtained from $\mathcal{A}'$ by assigning empty bundles to all $n-m-1$ excluded players also satisfies the same criterion. Since $I'$ has $m$ items and $m+1$ players, at least one player $i$ in $\mathcal{A}'$ necessarily receives the empty bundle. Therefore, when extending $\mathcal{A}'$ to $\mathcal{A}$, the fairness guarantees established in $\mathcal{A}'$ are preserved, as the excluded players also receive empty bundles and can be treated analogously to player $i$.

In light of this observation, any polynomial or pseudo-polynomial algorithm for computing approximately equitable allocations in instances with $m \geq n-1$ can be adapted to instances $I$ with $n-1 > m$ by applying it to the reduced instance $I'$ with $n' = O(m)$ players. Consequently, in this case, the dependence of the running time on $n$ can be replaced by a dependence on $m = O(n)$.
\end{remark}
In Algorithm \ref{alg1} of Appendix \ref{app:algorithms}, we present a {\em local-search algorithm} that returns an $\EQXR$ allocation in pseudo-polynomial time under objective valuations; this algorithm is a variant of the local-search framework of \cite{BBPP24}. Algorithm \ref{alg1} starts from an allocation assigning all items to the first agent. Then, it executes either the while-loop at lines 4–7 if $v_1(M) > 0$, or the while-loop at lines 12–15 if $v_1(M) < 0$, and simply returns the initial allocation otherwise (i.e., if $v_1(M)=0$). The former while-loop, referred to as {\em good-moving}, repeatedly moves goods from the bundle of some envied agent (in terms of inequity) to the least valuable bundle. The latter, referred to as {\em chore-moving}, repeatedly moves chores from the bundle of some envious agent to the bundle of the most envied one. Both loops terminate when the $\EQXR$ condition is satisfied\footnote{
In the full proof, for completeness, we analyze both the cases $v_1(M) > 0$ and $v_1(M) < 0$, corresponding to the execution of the good-moving while-loop and the chore-moving while-loop, respectively. Anyway, to handle the case $v_1(M) < 0$, one could directly appeal to a mirroring argument with respect to the case $v_1(M) > 0$. Indeed, consider a fair allocation instance with valuations $v_i'$ obtained by multiplying the original valuations $v_i$ by $-1$. We observe that: the case $v_1'(M) > 0$ corresponds to $v_1(M) < 0$; executing the good-moving while-loop with respect to $v_i'$ is equivalent to executing the chore-moving while-loop with respect to the original valuations $v_i$ (since inequalities reverse, indices $i,j$ are exchanged, and goods become chores); and property (i) of $\EQXR$ under $v_i'$ is equivalent to property (ii) of $\EQXR$ under $v_i$. Thus, when $v_1(M) < 0$, the ability of the good-moving while-loop to compute allocations satisfying property (i) of the $\EQXR$ guarantee under valuations $v_i'$ can be directly translated into the ability of the chore-moving while-loop to compute allocations satisfying property (ii) of the $\EQXR$ guarantee under the original valuations $v_i$.}.

\begin{theorem}\label{thm1}
Given an allocation instance $I=(N,M=G\cup C,(v_i)_{i\in N})$ with objective valuations, the local-search algorithm returns an $\EQXR$ allocation in pseudo-polynomial time $O(V_{max}\cdot n\cdot m^2)$, with $V_{max}:=\max_{i\in N}\max\{v_i(G),|v_i(C)|\}$.
\end{theorem}
\begin{proof}[Sketch of the proof]
The full proof of the theorem is deferred to Appendix~\ref{app:objective}; here, we provide the main intuitions only. If $v_1(M)=0$, the allocation returned by the algorithm is trivially EQ, and then $\EQXR$. If $v_1(M)>0$ or $v_1(M)<0$, one can first show that the violation of the while-condition of each while-loop in the local-search algorithm is equivalent to finding an $\EQXR$ allocation (Lemma~\ref{lem3} of Appendix~\ref{app:objective}).   By this observation, to prove that the algorithm returns an $\EQXR$ allocation, it suffices to show that the executed while-loop (good-moving or chore-moving) terminates. To establish termination, we use a potential function argument. Specifically, we define two ad-hoc potential functions, one for each while-loop. Each potential function depends only on the current allocation, returns a triplet of values, and increases in every iteration (refer to Lemma \ref{lem4} and \ref{lem5} of Appendix~\ref{app:objective}), according to a lexicographic order $\succ$ of the triplet values (i.e., \( T_1 \succ T_2 \) if and only if the first differing component between triplets \( T_1 = (x_1, y_1, z_1) \) and \( T_2 = (x_2, y_2, z_2) \) has a higher value in \( T_1 \)). Since each potential function is necessarily bounded (e.g., by the maximal triplet achievable across all possible allocations), it follows that each while-loop must terminate. For the good-moving while-loop (i.e., case $v_1(M)>0$), we consider the potential function $\Phi$ that assigns, to each allocation $\mathcal{A}$, the triplet  $(x(\mathcal{A}),y(\mathcal{A}),z(\mathcal{A}))$, such that $x(\mathcal{A})$ is the minimum valuation in $\mathcal{A}$, $y(\mathcal{A})$ is minus the number of agents obtaining the minimum valuation, and  $z(\mathcal{A})$ is the number of items allocated to such agents. For the chore-moving while loop (i.e., case $v_1(M)<0$), we consider the potential function $\Xi$ that assigns, to each allocation $\mathcal{A}$, the triplet $(x'(\mathcal{A}),y'(\mathcal{A}),z'(\mathcal{A}))$, such that $x'(\mathcal{A})$ is minus the maximum valuation in $\mathcal{A}$, $y'(\mathcal{A})$ is minus the number of agents obtaining the maximum valuation, and $z'(\mathcal{A})$ is the number of items allocated to such agents.

To estimate the time complexity of the algorithm, we observe that it is equal to the time required for each iteration of the while-loops, multiplied by the total number of iterations. The while-condition can be checked in time $O(m+n)=O(m)$, as it requires to check all bundles, and all items for each bundle\footnote{$O(m+n)=O(m)$ holds since we assumed w.l.o.g. that $n=O(m)$ (by Remark \ref{rema_compl_agents} of Appendix \ref{sec:objective}).}.

To bound the number of iterations, we first observe that the potential function increases at each iteration. Consequently, the number of iterations is bounded by the maximum value of the potential functions \( \Phi \) and \( \Xi \), that is \( O(V_{\text{max}} \cdot n \cdot m) \) (as established in Lemma \ref{lem_thm1_comple} in Appendix~\ref{app:objective}).

Thus, the overall time complexity is \( O(V_{\text{max}} \cdot n \cdot m^2) \)
%
%Thus, the existence of 
\end{proof}
%
%We first observe that the instructions $\arg\max$ and $\arg\min$ of Algorithm \ref{alg1} are well-defined, as a consequence of the foll

By observing that, under monotone non-decreasing (resp. non-increasing) valuations, property (i) (resp. (ii)) of the $\EQXR$ guarantee is equivalent to the stronger EQX notion, we obtain the following corollary of Theorem~\ref{thm1}.
\begin{corollary}\label{thm1_cor1}
The local-search algorithm returns an EQX allocation in pseudo-polynomial time, under monotone (non-decreasing and non-increasing) valuations. 
\end{corollary}
%By exploiting the equivalence between equity and fairness under identical valuations, we have the following corollary:
%\begin{corollary}\label{thm1_cor2}
%Under identical objective valuations, the local-search algorithm returns an EFX* allocation in pseudo-polynomial time. Furthermore, it returns an EFX allocation under identical monotone (non-decreasing and non-increasing) valuations. 
%\end{corollary}
In the following, we present a {\em greedy algorithm} that returns in polynomial time an EQ1 allocation under objective valuations (for the pseudo-code, see Algorithm \ref{alg2} in Appendix~\ref{app:objective}). This algorithm can be seen as a  straightforward generalization of the greedy approach introduced by \citet{HS25}, which was specifically designed for the case of objective additive valuations.

Starting from the empty allocation, the algorithm first processes all goods, and repeatedly assigns each good (in an arbitrary order) to the least valuable bundle of the current iteration. Then, it processes all chores, and repeatedly assigns each chore to the most valuable bundle of the current iteration. We observe that the initial local-search algorithm clearly returns an EQ1 allocation but operates in pseudo-polynomial time, whereas the greedy algorithm typically runs faster, operating in polynomial time.
\begin{theorem}\label{thm2}
Given an allocation instance $I=(N,M=G\cup C,(v_i)_{i\in N})$ with objective valuations, the greedy algorithm returns an EQ1 allocation in polynomial time $O(m\log n)$.
\end{theorem}
\begin{proof}[Sketch of the proof]
It can be shown by induction that the partial allocation \( \mathcal{A}^t \) obtained after assigning the first \( t \) items is EQ1, and consequently, the final allocation \( \mathcal{A} = \mathcal{A}^t \) is also EQ1. Full details are deferred to Appendix~\ref{app:objective}.

Regarding the time complexity, the algorithm performs at most $O(m)$ item insertions. For each insertion, it identifies and updates either the least or the most valuable bundle among $n$, which can be done in $O(\log n)$ time using a heap-tree data structure. Therefore, the overall time complexity is $O(m \log n)$.
\end{proof}

%We have the following corollary of Theorem~\ref{thm2}. 
%\begin{corollary}\label{thm2_cor1}
%The greedy algorithm returns an EQ1 allocation in polynomial time, under monotone (non-decreasing or non-increasing) valuations. 
%\end{corollary}
Finally, we consider a variant of the above greedy algorithm, called {\em strongly-greedy algorithm}, that is specialized for objective valuations that are also additive (for the pseudo-code, refer to Algorithm \ref{alg3} in Appendix~\ref{app:objective}). 
This algorithm operates like the standard greedy algorithm, but in each iteration, when a good (or a chore) is assigned to an agent, it selects the one that maximizes (or minimizes) that agent's valuation. Under objective additive valuations, the strongly-greedy algorithm returns an $\EQXR$ allocation, ensuring stronger equity than the EQ1 guarantee provided by \citet{HS25} or the standard greedy algorithm (Algorithm \ref{alg2} in Appendix \ref{app:algorithms}).
\begin{theorem}\label{thm3}
Given an allocation instance $I=(N,M=G\cup C,(v_i)_{i\in N})$ with additive objective valuations, the strongly-greedy algorithm returns an $\EQXR$ allocation in polynomial time $O(mn\log m)$.
\end{theorem}
Theorem~\ref{thm3} can be shown in a manner similar to Theorem~\ref{thm2}, with a sketch of the proof provided in Appendix~\ref{app:objective}.

%Finally, we show that, if the considered valuations are also non-zero-marginal, the complexity of Algorithm \ref{alg1} can be lowered. 
%\begin{corollary}
%Under objective non-zero-marginal valuations, Algorithm \ref{alg1} returns an $\EQXR$ allocation in pseudo-polynomial time $O(V_{max}\cdot n\cdot m)$. 
%\end{corollary}
\newpage
\section{Missing Proofs from Appendix \ref{sec:objective}}\label{app:objective}
\subsection{Full Proof of Theorem \ref{thm1} - Correctness}
We first show the correctness of the algorithm. By relying on the two technical lemmas provided below (Lemmas~\ref{lem1} and \ref{lem2}), we will first show that the violation of the while-condition of each while-loop is equivalent to finding an $\EQXR$ allocation (Lemma~\ref{lem3}). 
\begin{lemma}\label{lem1}
Under objective valuations, for any $i\in N$ and  non-empty bundle $S\subseteq M$ such that $S$ contains goods (resp. chores) only, we have $v_i(S)\geq 0$ (resp. $v_i(S)\leq 0$). Equivalently, if $i$ and $S$ satisfy $v_i(S)>0$ (resp. $v_i(S)< 0$), then there exists at least one good (resp. one chore) in $S$.
\end{lemma}
The proof of Lemma \ref{lem1} is straightforward and directly follows from the definition of goods and chores. 
%\begin{proof}[Proof of Lemma~\ref{lem1}]
%We prove only the first claim, as the second can be trivially established by contradiction from the first.
%Let \(i \in N\) and \(S = \{x_1, \ldots, x_k\} \subseteq M\) be an arbitrary non-empty bundle of items. We first assume that $S$ contains goods only. Let \(S(h) = \{x_1, \ldots, x_h\}\) denote the set of the first \(h\) elements of \(S\). We have $v_i(S)=\sum_{h=1}^k(v_i(S(h+1))-v_i(S(h)))=\sum_{h=1}^k(v_i(S(h)\cup \{x_h\})-v_i(S(h)))\geq 0$, where the last inequality holds since each $x_h$ is a good. 
%The proof for the case in which $S$ contains only chores is symmetric. Indeed, in this case we have $v_i(S)=\sum_{h=1}^k(v_i(S(h+1))-v_i(S(h)))=\sum_{h=1}^k(v_i(S(h+1))-v_i(S(h+1)\setminus\{x_h\}))\leq 0$, where the last inequality holds since each $x_h$ is a chore. 
%\end{proof}

\begin{lemma}\label{lem2}
Assume that $v_1(M)> 0$ (resp. $v_1(M)< 0$), and let $\mathcal{A}^t=(A_1^t,\ldots, A_n^t)$ denote the allocation obtained after each iteration $t\geq 1$ of the good-moving (resp. chore-moving) while-loop, with $\mathcal{A}^0$ denoting the initial allocation. Then, for any $t\geq 0$, we have that $A_1^t$ contains at least one good (resp. chore) and each non-empty bundle $A_h^t$ with $h\neq 1$ contains only goods (resp. chores). 
\end{lemma}
\begin{proof}[Proof of Lemma~\ref{lem2}]
We first show the claim in the case $v_1(M)> 0$, and we proceed by induction on $t\geq 0$. 
We start with $t=0$ (base case). We have $v_1(A_1^0)=v_1(M)> 0$, since $A_1^0=M$ and we are in the case $v_1(M)> 0$. Thus, by Lemma \ref{lem1}, we have that $A_1^0$ contains at least one good. Since $A_h^0=\emptyset$ for any $h\neq 1$, the above property is sufficient to establish the claim for $t=0$. 

Now, assume that the claim holds for an arbitrary $t$ (inductive hypothesis), and let us show it for $t+1$ (inductive step). For all agents $h$ such that $A_h^t=A_h^{t+1}$ (i.e., whose bundle remains unchanged in iteration $t+1$) the claim trivially holds by the inductive hypothesis. Next, we focus on agents who either receive or lose one good in each iteration $t+1$. Let $h$ be any such an agent. If $h\neq 1$, by the inductive hypothesis, we have that $A_h^{t}$ contains goods only. So, if $h$ receives a good, then $A_h^{t+1}$ still contains good only; if $h$ loses a good, then $A_h^{t+1}$, if not empty, also contains goods only. Now, assume that $h=1$. By the inductive hypothesis, we have that $A_1^{t}$ contains at least one good. If agent $1$ receives a good in iteration $t+1$, we then have that $A_1^{t+1}$ contains at least one good as well. If agent $1$ loses a good in iteration $t+1$, by lines 4 and 5 of Algorithm \ref{alg1}, there exist an agent $i$ and a good $g\in A_1^t$ such that $v_1(A_1^t\setminus\{g\})>v_i(A_i^t)$. We derive $v_1(A_1^{t+1})=v_1(A_1^t\setminus\{g\})>v_i(A_i^t)\geq 0$, where the last inequality holds since $A_i^t$ contains goods only (by Lemma \ref{lem1} and the inductive hypothesis). Then, since $v_1(A_1^{t+1})> 0$, Lemma \ref{lem1} implies that $A_1^{t+1}$ contains at least one good. We conclude that the inductive step holds for any agent $h$, and the claim follows in the case $v_1(M)\geq 0$.

For the case $v_1(M)< 0$, the claim can be shown by using an argument symmetric to that of the previous case. Let $v_i'$ be the valuation obtained by multiplying the value of $v_i$  by minus one, that is, $v'_i(S)=-v_i(S)$ for any $i\in N$ and $S\subseteq M$. We observe that $v'_1(M)=-v_1(M)>0$.  Thus, under the valuations $v_i'$, the previous analysis implies that, for any $t \geq 0$, the bundle $A_1^t$ contains at least one good, while each non-empty bundle $A_h^t$ with $h \neq 1$ contains only goods. Returning to the original valuations $v_i$, this implies that, for any $t \geq 0$, $A_1^t$ contains at least one chore, while each non-empty bundle $A_h^t$ with $h \neq 1$ contains only chores.
\end{proof}
\begin{lemma}\label{lem3}
Assume that $v_1(M)>0$ (resp. $v_1(M)< 0$). If $\mathcal{A}$ is an allocation found by the good-moving (resp. chore-moving) while-loop that violates the corresponding while-condition, then $\mathcal{A}$ is $\EQXR$.
\end{lemma}

\begin{proof}[Proof of Lemma~\ref{lem3}]
We first assume that $v_1(M)> 0$. We will show that, if the while-condition is violated, property (i) of the $\EQXR$ guarantee is satisfied for any $i\in N$ (that is, the inequity an agent \(i\) feels toward the bundle assigned to another agent \(j\) can be eliminated by removing any good from \(j\)'s bundle). 
To establish this, it suffices to show that every non-empty bundle in \(\mathcal{A} = (A_1, \ldots, A_n)\) contains at least one good, as the remaining part of property (i) of the $\EQXR$ guarantee is ensured by the violation of the while-condition. This fact is immediately guaranteed by Lemma~\ref{lem2}, and thus the claim follows.

To show the claim when \(v_1(M) < 0\), we use an argument symmetric to the previous case. Specifically, we will show that, if the while-condition is violated, property (ii) of the $\EQXR$ guarantee is satisfied for any \(i \in N\) (that is, the inequity an agent \(i\) feels toward the bundle assigned to another agent \(j\) can be eliminated by removing any chore from \(i\)'s bundle). To establish this, it suffices to show that every non-empty bundle in \(\mathcal{A} = (A_1, \ldots, A_n)\) contains at least one chore, as the remaining part of property (ii) of the $\EQXR$ guarantee is ensured by the violation of the while-condition. Again, this fact is guaranteed by Lemma~\ref{lem2}, thus proving the claim.
\end{proof}

We now move to the main part of the proof. If $v_1(M)=0$, the allocation returned by the algorithm assigns all items to the first agent. Such an allocation is trivially EQ, and therefore $\EQXR$, as each agent receives zero value. Thus, in the remainder of the proof we focus on the cases $v_1(M)>0$ and $v_1(M)<0$. By Lemma~\ref{lem3}, to prove that Algorithm~\ref{alg1} returns an $\EQXR$ allocation in these cases, it suffices to show that the executed while-loop (good-moving or chore-moving) terminates. To establish termination, we use a potential function argument. Specifically, we define two ad-hoc potential functions, one for each while-loop. Each potential function depends only on the current allocation, returns a triplet of values, and increases in every iteration, following a lexicographic order of the triplet values. Since each potential function is necessarily bounded (e.g., by the maximal triplet achievable across all possible allocations), it follows that the while-loop must terminate. We treat separately the cases $v_1(M)> 0$ and $v_1(M)< 0$, each corresponding to a distinct potential function, though the proof arguments are analogous. 

\paragraph{Case $v_1(M)> 0$: }
We first consider the case $v_1(M)> 0$.  Let $\Phi$ be the potential function that assigns, to each allocation $\mathcal{A}$, the triplet  $(x(\mathcal{A}),y(\mathcal{A}),z(\mathcal{A}))$, where $x(\mathcal{A})=\min_{i\in N}v_i(A_i)$, $y(\mathcal{A})=-|\{i\in N:v_i(A_i)=x(\mathcal{A})\}|$ and $z(\mathcal{A})=\sum_{i\in N:v_i(A_i)=x(\mathcal{A})}|A_i|$. Specifically,  $x(\mathcal{A})$ is the minimum valuation in $\mathcal{A}$, $-y(\mathcal{A})$ is the number of agents obtaining the minimum valuation, and  $z(\mathcal{A})$ is the number of items allocated to such agents. We consider the total order over triplets such that $(x_1,y_1,z_1)\succ (x_2,y_2,z_2)$ iff $(x_1,y_1,z_1)$ is {\em lexicographically higher} than $(x_2,y_2,z_2)$, i.e., if $x_1>x_2 \vee (x_1=x_2\wedge y_1>y_2)\vee (x_1=x_2\wedge y_1=y_2\wedge z_1>z_2)$. Assume by contradiction that the good-moving while-loop never terminates. Let $\mathcal{A}^t$ denote the allocation computed after each iteration $t\geq 0$, where $\mathcal{A}^{0}$ denotes the initial allocation, and $t=0$ denotes the state prior to the first iteration. We have the following lemma:
\begin{lemma}\label{lem4}
Assume that $v_1(M)>0$ and let $t\geq 0$ be an iteration of the good-moving while-loop such that the allocation $\mathcal{A}_t$ obtained at the end of the iteration satisfies again the while-condition. Then, we have $\Phi(\mathcal{A}^{t+1})\succ\Phi(\mathcal{A}^{t})$. 
\end{lemma}
\begin{proof}[Proof of Lemma \ref{lem4}]
Let $t\geq 0$ be the iteration satisfying the hypothesis of the lemma. Then, there exist two agents $i,j\in N$ and a good $g\in A_j$ such that $i$ has the least valuable bundle in $\mathcal{A}^t$, $v_j(A_j^t\setminus \{g\})>v_i(A_i^t)$, $\mathcal{A}^{t+1}$ is obtained from $\mathcal{A}^{t}$ by moving good $g$ from the bundle of $j$ to the bundle of $i$, and the other bundles do not vary. One of the following three cases can occur:
\begin{itemize}
\item (i) $v_i(A_i^{t+1})=v_i(A_i^t)$: By $v_j(A_j^{t+1})=v_j(A_j^t\setminus \{g\})>v_i(A_i^t)=v_i(A_i^{t+1})$, the bundle of $i$ (resp. $j$) is (resp. is not) the least valuable bundle in both $\mathcal{A}^t$ and $\mathcal{A}^{t+1}$. 
Thus, the value of the least valuable bundle does not vary, i.e., \(x(\mathcal{A}^{t+1}) = x(\mathcal{A}^{t})\), the number of least valuable bundles remains the same, i.e., \(-y(\mathcal{A}^{t+1}) = -y(\mathcal{A}^{t})\), but the overall cardinality of the least valuable bundles increases by one, i.e., \(z(\mathcal{A}^{t+1}) = z(\mathcal{A}^{t})+1>z(\mathcal{A}^{t})\) (since \(|A_i^{t+1}| = 1 + |A_i^{t}|\), and the other eventual least valuable bundles do not change). We conclude that $\Phi(\mathcal{A}^{t+1})\succ\Phi(\mathcal{A}^{t})$ holds in case (i). 
\item (ii) $v_i(A_i^{t+1})>v_i(A_i^t)$ and $-y(\mathcal{A}^t)>1$: Given an arbitrary agent $h\neq i$ possessing one of the least valuable bundles in $\mathcal{A}^t$, we have   $v_i(A_i^t)=v_h(A_h^t)=v_h(A_h^{t+1})$, where the last equality holds since the bundle of $h$ does not vary in iteration $t+1$; furthermore, such an agent $h\neq i$ exists, since the number of least valuable bundles in $\mathcal{A}^t$ is $-y(\mathcal{A})>1$. Then, by $v_j(A_j^{t+1})=v_j(A_j^t\setminus \{g\})>v_i(A_i^t)=v_h(A_h^{t+1})$ and $v_i(A_i^{t+1})>v_i(A_i^t)=v_h(A_h^{t+1})$, the following facts hold: the bundle of $j$ is not the least valuable bundle in both $\mathcal{A}^t$ and $\mathcal{A}^{t+1}$; the bundle of $i$ is the least valuable bundle in $\mathcal{A}^t$ but not in $\mathcal{A}^{t+1}$; the bundle of $h$ is the least valuable bundle in both $\mathcal{A}^t$ and $\mathcal{A}^{t+1}$. By the arbitrariness of the choice of  agent $h\neq i$ possessing one of the least valuable bundles in $\mathcal{A}^t$, we conclude that the minimum value $x(\mathcal{A}^{t+1})$ among all bundles in $\mathcal{A}^{t+1}$ continues to be equal to $x(\mathcal{A}^t)=v_i(A_i^t)$, and is achieved by all agents $h$ obtaining such value in $\mathcal{A}^t$ but $i$, i.e., $-y(\mathcal{A}^{t+1})=-y(\mathcal{A}^t)-1\geq 1$. We conclude that $x(\mathcal{A}^{t+1})=x(\mathcal{A}^{t})$ and $y(\mathcal{A}^{t+1})=y(\mathcal{A}^t)+1>y(\mathcal{A}^t)$, that is, $\Phi(\mathcal{A}^{t+1})\succ\Phi(\mathcal{A}^{t})$ holds in case (ii). 
\item (iii) $v_i(A_i^{t+1})>v_i(A_i^t)$ and $-y(\mathcal{A}^t)=1$: Since the number of least valuable bundles in $\mathcal{A}^t$ is $-y(\mathcal{A}^t)=1$, we have that $A_i^t$ is the unique bundle minimizing the valuation in $\mathcal{A}^t$. Then, by $v_i(A_i^{t+1})>v_i(A_i^t)$ and $v_j(A_j^{t+1})=v_j(A_j^t\setminus \{g\})>v_i(A_i^t)$, we necessarily have that that the minimum valuation $x(\mathcal{A}^{t+1})$ in $\mathcal{A}^{t+1}$ is strictly higher than the minimum valuation $x(\mathcal{A}^t)=v_i(A_i^t)$ in $\mathcal{A}^t$, that is, $\Phi(\mathcal{A}^{t+1})\succ\Phi(\mathcal{A}^{t})$ holds in case (iii). 
\end{itemize}
In any case, we have $\Phi(\mathcal{A}^{t+1})\succ\Phi(\mathcal{A}^{t})$, thus showing the claim.
\end{proof}
Under the hypothesis of non-termination of the good-moving while-loop, there exists at least one allocation that is visited twice. In other words, there are at least two iterations $t,t'\geq 0$ with $t<t'$ such that $\mathcal{A}^t=\mathcal{A}^{t'}$, and hence $\Phi(\mathcal{A}^{t})=\Phi(\mathcal{A}^{t'})$. Thus, by this observation and Lemma~\ref{lem4}, we have $\Phi(\mathcal{A}^{t})=\Phi(\mathcal{A}^{t'})\succ\Phi(\mathcal{A}^{t'-1})\succ\ldots \succ\Phi(\mathcal{A}^{t+1})\succ \Phi(\mathcal{A}^{t})$, that leads to the contradiction $\Phi(\mathcal{A}^{t})\succ \Phi(\mathcal{A}^{t})$. We conclude that the good-moving while-loop necessarily terminates, and then, by Lemma~\ref{lem3}, it returns an $\EQXR$ allocation if $v_1(M)>0$. 

\paragraph{Case $v_1(M)< 0$: } Now, let us consider the case $v_1(M)< 0$. Let $\Xi$ be the potential function that assigns, to each allocation $\mathcal{A}$, the triplet $(x'(\mathcal{A}),y'(\mathcal{A}),z'(\mathcal{A}))$, where $x'(\mathcal{A})=-\max_{i\in N}v_i(A_i)$, $y'(\mathcal{A})=-|\{i\in N:v_i(A_i)=-x'(\mathcal{A})\}|$ and $z'(\mathcal{A})=\sum_{i\in N:v_i(A_i)=-x'(\mathcal{A})}|A_i|$. Specifically,  $-x'(\mathcal{A})$ is the maximum valuation in $\mathcal{A}$, $-y'(\mathcal{A})$ is the number of agents obtaining the maximum valuation, and $z'(\mathcal{A})$ is the number of items allocated to such agents.  Let $\succ$ be the lexicographic order among triplets defined above. Assume by contradiction that the chore-moving while-loop never terminates. Let $\mathcal{A}^t$ denote the allocation computed after each iteration $t\geq 0$, where $\mathcal{A}^{0}$ denotes the initial allocation, and $t=0$ denotes the state prior to the first iteration. We have the following lemma (the proof has an argument symmetric to that of Lemma~\ref{lem4}):
\begin{lemma}\label{lem5}
Assume that $v_1(M)<0$ and let $t\geq 0$ be an iteration of the chore-moving while-loop such that the allocation $\mathcal{A}_t$ obtained at the end of the iteration satisfies again the while-condition. Then, we have $\Xi(\mathcal{A}^{t+1})\succ\Xi(\mathcal{A}^{t})$. 
\end{lemma}
\begin{proof}[Proof of Lemma \ref{lem5}]
Let $t\geq 0$ be the iteration satisfying the hypothesis of the lemma. Then, there exist two agents $i,j\in N$ and a chore $c\in A_i$ such that $j$ has the most valuable bundle in $\mathcal{A}^t$, $v_i(A_i^t\setminus \{c\})<v_j(A_j^t)$, $\mathcal{A}^{t+1}$ is obtained from $\mathcal{A}^{t}$ by moving chore $c$ from the bundle of $i$ to the bundle of $j$, and the other bundles do not vary. One of the following three cases can occur: 
\begin{itemize}
\item (i) $v_j(A_j^{t+1})=v_j(A_j^t)$: By $v_i(A_i^{t+1})=v_i(A_i^t\setminus \{c\})<v_j(A_j^t)=v_j(A_j^{t+1})$, the bundle of $j$ (resp. $i$) is (resp. is not) the most valuable bundle in both $\mathcal{A}^t$ and $\mathcal{A}^{t+1}$. 
Thus, the value of the most valuable bundle does not vary, i.e., \(x'(\mathcal{A}^{t+1}) = x'(\mathcal{A}^{t})\), the number of most valuable bundles remains the same, i.e., \(-y'(\mathcal{A}^{t+1}) = -y'(\mathcal{A}^{t})\), but the overall cardinality of the most valuable bundles increases by one, i.e., \(z'(\mathcal{A}^{t+1}) = z'(\mathcal{A}^{t})+1>z'(\mathcal{A}^{t})\) (since \(|A_j^{t+1}| = |A_j^{t}|+1\), and the other eventual most valuable bundles do not vary). We conclude that $\Xi(\mathcal{A}^{t+1})\succ\Xi(\mathcal{A}^{t})$ holds in case (i). 
\item (ii) $v_j(A_j^{t+1})<v_j(A_j^t)$ and $-y'(\mathcal{A}^t)>1$: Given an arbitrary agent $h\neq j$ possessing one of the most valuable bundles in $\mathcal{A}^t$, we have   $v_j(A_j^t)=v_h(A_h^t)=v_h(A_h^{t+1})$, where the last equality holds since the bundle of $h$ does not vary in iteration $t+1$; furthermore, such an agent $h\neq j$ exists, since the number of most valuable bundles in $\mathcal{A}^t$ is $-y'(\mathcal{A})>1$. Then, by $v_i(A_i^{t+1})=v_i(A_i^t\setminus \{c\})<v_j(A_j^t)=v_h(A_h^{t+1})$ and $v_j(A_j^{t+1})<v_j(A_j^t)=v_h(A_h^{t+1})$, the following facts hold: the bundle of $i$ is not the most valuable bundle in both $\mathcal{A}^t$ and $\mathcal{A}^{t+1}$; the bundle of $j$ is the most valuable bundle in $\mathcal{A}^t$ but not in $\mathcal{A}^{t+1}$; the bundle of $h$ is the most valuable bundle in both $\mathcal{A}^t$ and $\mathcal{A}^{t+1}$. By the arbitrariness of the choice of  agent $h\neq j$ possessing one of the most valuable bundles in $\mathcal{A}^t$, we conclude that the maximum value $-x'(\mathcal{A}^{t+1})$ among all bundles in $\mathcal{A}^{t+1}$ continues to be equal to $-x'(\mathcal{A}^t)=v_j(A_j^t)$, and is achieved by all agents $h$ obtaining such value in $\mathcal{A}^t$ but $j$, i.e., $-y'(\mathcal{A}^{t+1})=-y'(\mathcal{A}^t)-1\geq 1$. We conclude that $x'(\mathcal{A}^{t+1})=x'(\mathcal{A}^{t})$ and $y'(\mathcal{A}^{t+1})=y'(\mathcal{A}^t)+1>y'(\mathcal{A}^t)$, that is, $\Xi(\mathcal{A}^{t+1})\succ\Xi(\mathcal{A}^{t})$ holds in case (ii). 
\item (iii) $v_j(A_j^{t+1})<v_j(A_j^t)$ and $-y'(\mathcal{A}^t)=1$: Since the number of most valuable bundles in $\mathcal{A}^t$ is $-y'(\mathcal{A}^t)=1$, we have that $A_j^t$ is the unique bundle maximizing the valuation in $\mathcal{A}^t$. Then, by $v_j(A_j^{t+1})<v_j(A_j^t)$ and $v_i(A_i^{t+1})=v_i(A_i^t\setminus \{c\})<v_j(A_j^t)$, we necessarily have that that the maximum valuation $-x'(\mathcal{A}^{t+1})$ in $\mathcal{A}^{t+1}$ is strictly lower than the maximum valuation $-x'(\mathcal{A}^t)=v_j(A_j^t)$ in $\mathcal{A}^t$, that is, $\Xi(\mathcal{A}^{t+1})\succ\Xi(\mathcal{A}^{t})$ holds in case (iii). 
\end{itemize}
In any case, we have $\Xi(\mathcal{A}^{t+1})\succ\Xi(\mathcal{A}^{t})$, thus showing the claim.
\end{proof}
Under the hypothesis of non-termination of the chore-moving while-loop, there are two iterations $t,t'\geq 0$ with $t<t'$ such that $\mathcal{A}^t=\mathcal{A}^{t'}$, that is, $\Xi(\mathcal{A}^{t})=\Xi(\mathcal{A}^{t'})$. Thus, as in the previous case, we have $\Xi(\mathcal{A}^{t})=\Xi(\mathcal{A}^{t'})\succ\Xi(\mathcal{A}^{t'-1})\succ\ldots \succ\Xi(\mathcal{A}^{t+1})\succ \Xi(\mathcal{A}^{t})$, that leads to the contradiction $\Xi(\mathcal{A}^{t})\succ \Xi(\mathcal{A}^{t})$. We conclude that the chore-moving while-loop necessarily terminates, and then, by Lemma~\ref{lem3}, it returns an $\EQXR$ allocation in case $v_1(M)< 0$, too. 

\subsection{Full Proof of Theorem \ref{thm1} - Complexity}
The complexity of Algorithm \ref{alg1} is proportional to the maximum number of iterations performed by the while-loops, multiplied by the number of steps required to check the while-condition. 

The while-condition can be checked in time $O(m+n)=O(m)$, as it requires to check all bundles, and all items for each bundle. 

The number of iterations of the while-loops is bounded by the maximum number of increases that the potential functions in Lemmas \ref{lem4} and \ref{lem5} can undergo, which is $O(V_{max} \cdot n \cdot m)$, by the following lemma:
\begin{lemma}\label{lem_thm1_comple}
The number of iterations of each while-loop (good-moving and chore-moving) is at most $O(V_{max} \cdot n \cdot m)$.
\end{lemma}
\begin{proof}[Proof of Lemma \ref{lem_thm1_comple}]
by Lemma \ref{lem4},  the potential function $\Phi$ lexicographically increases after each iteration. Thus, as $\Phi$ returns a triplet in the set $[0,\max_{i\in N}v_i(G)]\times \{-n,\ldots, -1\}\times [m]$, we have that after at most $\max_{i\in N}v_i(G)\cdot n\cdot m\leq V_{max}\cdot n\cdot m$ iterations the good-moving while loop necessarily terminates. Regarding the chore-moving while-loop, Lemma~\ref{lem5} implies that the potential function $\Xi$ lexicographically increases after each iteration. Again, as $\Xi$ returns a triplet in the set $[0,\max_{i\in N}|v_i(C)|]\times \{-n,\ldots, -1\}\times [m]$, we have that after at most $\max_{i\in N}|v_i(C)|\cdot n\cdot m\leq V_{max}\cdot n\cdot m$ iterations the chore-moving while loop necessarily terminates. 
\end{proof}
 We conclude that the complexity of Algorithm \ref{alg1} is $O(V_{max}\cdot n\cdot m^2)$. 
\subsection{Full Proof of Theorem \ref{thm2}}
We first focus on the correctness of the algorithm. Let $\mathcal{A}^t$ denote the partial allocation obtained after assigning $t$ items. By exploiting the greedy choice of the algorithm, it can be shown by induction on $t \geq 0$ that each partial allocation $\mathcal{A}^t$ is EQ1.

For $t=0$ the claim trivially holds. Assume that the claim holds for iteration $t$, and let us show it for $t+1$ (inductive step). To prove the EQ1 guarantee, it is sufficient showing that, given $i,j\in N$ such that $v_i(A_i^{t+1})<v_j(A_j^{t+1})$, $v_i(A_i^{t+1}\setminus\{x\})\geq v_j(A_j^{t+1}\setminus \{x\})$ holds for some item $x\in A_i\cup A_j$. Let $i,j$ be two distinct agents such that $v_i(A_i^{t+1})<v_j(A_j^{t+1})$. If the item assigned at step $t+1$ is a good, denoted by $g$, it is assigned to the agent $h$ having the lowest valuation in $\mathcal{A}^t$. If $h=j$, we have that $v_i(A_i^{t+1})=v_i(A_i^{t})\geq v_h(A_h^{t})=v_j(A_j^{t})=v_j(A_j^{t+1}\setminus \{g\})$, where the inequality holds by the minimality of $v_h(A_h^t)$. If $h\neq j$, the bundle assigned to $j$ remains unchanged after iteration $t+1$ and the bundle assigned to $i$ either remains unchanged (case $h\notin \{i,j\}$) or receives an additional good $g$ (case $h=i$). By the inductive hypothesis, there exists $x\in A_i^t\cup A_j^t$ such that  $v_i(A_i^{t}\setminus\{x\})\geq v_j(A_j^{t}\setminus \{x\})$. Thus, by the previous observations, we have $v_i(A_i^{t+1}\setminus{x}) \geq v_i(A_i^{t+1}\setminus\{x,g\})=v_i(A_i^{t}\setminus\{x\}) \geq v_j(A_j^{t}\setminus {x})=v_j(A_j^{t+1}\setminus {x})$, and therefore the EQ1 guarantee continues to hold in $\mathcal{A}^{t+1}$.

Now, assume that the item assigned in iteration $t+1$ is a chore, denoted by $c$. In such a case, $c$ is assigned to the agent $h$ having the highest value in the current iteration. If $h=i$, we have that $v_i(A_i^{t+1}\setminus \{c\})=v_i(A_i^{t})=v_h(A_h^{t})\geq v_j(A_j^t)=v_j(A_j^{t+1})$, where the inequality holds by the maximality of $v_h(A_h^t)$. If $h\neq i$, the bundle assigned to $i$ remains unchanged after iteration $t+1$ and the bundle assigned to $j$ either remains unchanged (case $h\notin \{i,j\}$) or receives an additional chore $c$ (case $h=j$). By the inductive hypothesis, there exists $x\in A_i^t\cup A_j^t$ such that  $v_i(A_i^{t}\setminus\{x\})\geq v_j(A_j^{t}\setminus \{x\})$. Thus, by the previous observations, we have $v_i(A_i^{t+1}\setminus\{x\})=v_i(A_i^{t}\setminus\{x\})\geq v_j(A_j^{t}\setminus \{x\})=v_j(A_j^{t+1}\setminus \{x,c\})\geq v_j(A_j^{t+1}\setminus \{x\})$.

In any case, $\mathcal{A}^{t+1}$ is EQ1, and this shows the inductive step. 
We conclude that each allocation $\mathcal{A}^t$, and a fortiori the final allocation $\mathcal{A} = \mathcal{A}^m$ obtained after all items have been assigned, satisfies EQ1.

Regarding the time complexity, we observe that the algorithm performs at most $O(m)$ item insertions. During each insertion, it identifies and updates either the least or the most valuable bundle among $n$. This operation can be executed in $O(\log n)$ time using a heap-tree data structure, that can be preliminary constructed in time $O(n)$. Thus, the overall running time is $O(m\log n+n)=O(m\log n)$\footnote{$O(m\log n+n)=O(m\log n)$ holds since we assumed w.l.o.g. that $m\geq n-1$.}.

\subsection{Proof of Theorem \ref{thm3}}
We first show the correctness of the algorithm. Let $\mathcal{A}^t$ denote the partial allocation obtained after the assignment of $t$ items. Mimicking the proof of Theorem \ref{thm2}, one can easily show by induction on \( t \geq 0 \) that \(\mathcal{A}^t\) satisfies the following approximate equity guarantee: for any \( i, j \in N \) such that \( v_i(A_i^t) < v_j(A_j^t) \), we have \( v_i(A_i^t\setminus \{x\}) \geq v_j(A_j^t \setminus \{x\}) \), where $x$ is the last chore assigned to $A_i^t$ if $A_i^t$ contains at least one chore, and $x$ is the last good assigned to $A_j^t$ otherwise. By the greedy selection of each agent and the imposed ordering of items, if bundle \( A_j^t \) (resp. \( A_i^t \)) contains at least one good (resp. chore), the last good (resp. chore) assigned to the bundle is the least (resp. most) valuable good (resp. chore) within the bundle. This fact, combined with above approximate equity guarantee, implies the following: for any \( i, j \in N \) such that \( v_i(A_i^t) < v_j(A_j^t) \), we have \( v_i(A_i^t\setminus \{x\}) \geq v_j(A_j^t) \) for any chore $x\in A_i^t$ if $A_i^t$ contains at least one chore, and \( v_i(A_i^t) \geq v_j(A_j^t \setminus \{x\}) \) for any good $x\in A_j^t$ otherwise. This final approximate equity guarantee ensures that every allocation $\mathcal{A}^t$ is $\EQXR$, and, in particular, that the final allocation $\mathcal{A} := \mathcal{A}^m$ also satisfies the $\EQXR$ guarantee.

Regarding the complexity, the algorithm can be implemented by creating two heap-tree data structures that keep track of the agent with the minimum and maximum valuation (this can be done in time \( O(n) \)), and, for each agent \( i \in [n] \), two heap-tree data structures that keep track of the best good and the worst chore (this can be done in time \( O(mn) \)). Excluding the above initialization, the execution time of the algorithm is proportional to the number of iterations, i.e., \( m \), multiplied by the complexity of each iteration. To execute each iteration, the algorithm must find the agent \( i \) having the minimum or maximum valuation, and determine the best available good or the worst available chore for \( i \). Using the aforementioned heap-trees, these operations can be executed in time \( O(\log n + n \log m) \) per iteration: \( O(\log n) \) to find the agent maximizing or minimizing the valuation, \( O(\log m) \) to extract the best good or worst chore and \( O(n\log m) \) to update the $n$ heap-trees associated with the agents (in particular, to ensure that items already extracted in other heaps are treated as unusable in each heap). Therefore, the overall running time is 
$
O(n + mn + m(\log n + n \log m)) = O(mn \log m).
$

\newpage
\section{Missing Figures}\label{app:pictures}
\subsection{Figure \ref{fig:1}}
\begin{figure}[h!]
\centering
%scale 0.8 for ijcai
\begin{tikzpicture}[scale=0.8]
%figura 1

\node at (2.7,1.5) {\small (a)};
% Definiamo i vertici del triangolo grande (simplesso 2D)
\coordinate (A) at (0,0);
\node[below left] at (A) {\small $\vv_2$}; % Etichetta del vertice A (verde)

\coordinate (B) at (3,0);
\node[below right] at (B) {\small $\vv_3$}; % Etichetta del vertice B (blu)

\coordinate (C) at (1.5,2.598);

\node[above] at (C) {\small $\vv_1$}; % Etichetta del vertice C (rosso)

% Definiamo i vertici della triangolazione e coloriamoli
\coordinate (q) at (1.2,0.6);
\coordinate (u) at (0.5,0.5);
\coordinate (v) at (0.6,1.0392);
\coordinate (s) at (1,1.732);
\coordinate (w) at (1.5,1.5);
\coordinate (x) at (1.5,0);
\coordinate (y) at (2.3,0);
\coordinate (z) at (1.7,1);
\coordinate (t) at (2.4,1.0392);

% Colleghiamo tutti i punti per formare la triangolazione
\draw[thick] (A) -- (u);
\draw[thick] (A) -- (v);
\draw[thick] (A) -- (x);
\draw[thick] (u) -- (v);
\draw[thick] (u) -- (x);
\draw[thick] (u) -- (w);
\draw[thick] (v) -- (w);
\draw[thick] (v) -- (s);
\draw[thick] (w) -- (z);
\draw[thick] (w) -- (s);
\draw[thick] (w) -- (t);
\draw[thick] (w) -- (q);
\draw[thick] (w) -- (C);
\draw[thick] (s) -- (C);
\draw[thick] (t) -- (C);
\draw[thick] (t) -- (z);
\draw[thick] (t) -- (B);
\draw[thick] (z) -- (x);
\draw[thick] (z) -- (y);
\draw[thick] (z) -- (B);
\draw[thick] (y) -- (B);
\draw[thick] (y) -- (x);
\draw[thick] (q) -- (x);
\draw[thick] (q) -- (u);
\draw[thick] (q) -- (z);

% Coloriamo internamente i fully-colored symplices con pattern a strisce grige
 \usetikzlibrary{patterns}

\fill[lightgray] (q) -- (z) -- (w) -- cycle; % Triangolo (u, q, w)
\fill[lightgray] (w) -- (C) -- (t) -- cycle; % Triangolo (w, C, t)
\fill[lightgray] (u) -- (x) -- (q) -- cycle; % Triangolo (u, x, q)
\draw[thick] (q) -- (z) -- (w) -- cycle;
\draw[thick] (w) -- (C) -- (t) -- cycle;
\draw[thick] (u) -- (x) -- (q) -- cycle;

% Coloriamo i vertici del triangolo grande (simplesso 2D)
\fill[green] (A) circle (3pt); % Vertice A in verde
\fill[blue] (B) circle (3pt); % Vertice B in blu
\fill[red] (C) circle (3pt); % Vertice C in rosso

% Coloriamo i vertici della triangolazione

\fill[red] (q) circle (3pt); % Vertice q in rosso

\fill[green] (u) circle (3pt); % Vertice u in verde

\fill[green] (v) circle (3pt); % Vertice v in verde

\fill[red] (s) circle (3pt); % Vertice s in rosso

\fill[green] (w) circle (3pt); % Vertice w in verde

\fill[blue] (x) circle (3pt); % Vertice x in blu

\fill[green] (y) circle (3pt); % Vertice y in verde

\fill[blue] (z) circle (3pt); % Vertice z in blu

\fill[blue] (t) circle (3pt); % Vertice t in blu

%%%%%FIGURA 2
\node at (7,1.5) {\small (b)};

\begin{scope}[shift={(0.7,0)}]  % Trasla la figura verso destra di 0.5 cm

% Definiamo il punto di origine degli assi cartesiani
\coordinate (origin) at (3.6, -0.1);

% Disegno degli assi cartesiani a (3.7,0) con lunghezze maggiori
\draw[->, gray] (origin) -- (6.8, -0.1) node[right] {\small $x_1$};  % Asse x allungato
\draw[->, gray] (origin) -- (3.6, 3.1) node[above] {\small $x_2$};  % Asse y allungato

% Definiamo il triangolo di base
\coordinate (P1) at (3.7,0);
\coordinate (P2) at (6.5,2.8);
\coordinate (P3) at (3.7,2.8);

% Disegno del triangolo principale
\draw[thick] (P1) -- (P2) -- (P3) -- cycle;

% Parametri per la suddivisione
\def\n{9} % Numero di suddivisioni
\def\step{0.31} % Distanza tra le linee

\fill[lightgray] (3.7 + 3*\step, 6*\step) -- (3.7 + 3*\step, 7*\step) -- (3.7 + 4*\step, 7*\step) -- cycle; % Triangolo (u, q, w)

% Suddividere il triangolo
\foreach \i in {1,...,\n} {
    % Linee orizzontali
    \draw[thick] (3.7, \i*\step) -- (3.7 + \i*\step, \i*\step);
    % Linee verticali
    \draw[thick] (3.7 + \i*\step, \i*\step) -- (3.7 + \i*\step, 2.8);
    % Linee oblique
    \draw[thick] (3.7, \i*\step) -- (6.5 - \i*\step, 2.8);
}

% Creare i vertici dei triangolini risultanti
\foreach \i in {0,...,\n} {
    \foreach \j in {0,...,\i} {
        % Calcolo delle coordinate dei vertici
        \pgfmathsetmacro{\x}{3.7 + \j*\step}
        \pgfmathsetmacro{\y}{\i*\step}
        
        % Cambiare colore sui vertici del lato verticale alternando tra verde e blu
        \ifdim \x pt = 3.7pt
            % Alternare tra blu e verde sui vertici del lato verticale
            \pgfmathparse{mod(\i,2) == 0 ? "green" : "blue"}
            \edef\vertexcolor{\pgfmathresult}
        \else
            % Cambiare colore sui vertici del lato orizzontale alternando tra verde e rosso
            \ifdim \y pt = 0pt
                % Alternare tra verde e rosso sui vertici del lato orizzontale
                \pgfmathparse{mod(\j,2) == 0 ? "green" : "red"}
                \edef\vertexcolor{\pgfmathresult}
            \else
                % Altri vertici (rimangono come prima)
                \pgfmathparse{mod(\j*\j+\i,3) == 0 ? "blue" : (mod(\j*\j+\i,3) == 1 ? "red" : "green")}
                \edef\vertexcolor{\pgfmathresult}
            \fi
        \fi
        
        % Disegna il vertice
        \fill[\vertexcolor] (\x, \y) circle (2pt);
    }
}

%modifiche
\fill[green] (3.7 + 3*\step, 9*\step) circle (2pt);
\fill[green] (3.7 + 5*\step, 9*\step) circle (2pt);
\fill[green] (3.7 + 5*\step, 8*\step) circle (2pt);
\fill[green] (3.7 + 6*\step, 9*\step) circle (2pt);

\fill[red] (3.7 + 1*\step, 1*\step) circle (2pt);
\fill[red] (3.7 + 4*\step, 4*\step) circle (2pt);
\fill[red] (3.7 + 6*\step, 6*\step) circle (2pt);
\fill[red] (3.7 + 7*\step, 7*\step) circle (2pt);
\fill[red] (3.7 + 7*\step, 8*\step) circle (2pt);

% Aggiunta dei vertici principali
\fill[blue] (P1) circle (3pt);
\fill[red] (P2) circle (3pt);
\fill[green] (P3) circle (3pt);

% Etichette per i vertici principali
\node[below left] at (P1) {\small $\vv_3$};
\node[above right] at (P2) {\small $\vv_1$};
\node[above left] at (P3) {\small $\vv_2$};

\end{scope}
\end{tikzpicture}

\caption{The figure on the left (a) illustrates an application of Sperner's Lemma to a triangulation $ T $ of a 2-simplex $ \Delta = \text{conv}(\vv_1, \vv_2, \vv_3) $. The coloring function assigns a color from set \{1 (Red), 2 (Green), 3 (Blue)\} to each vertex $ \bfx \in V(T) $, and it is special. Specifically, for each $i\in [3]$, the color $ i $ does not appear on the 1-face of $ \Delta $ that does not contain $ \vv_i $ (i.e., on the edge opposite to vertex $ \vv_i $). Sperner's Lemma guarantees the existence of an odd number of  fully-colored simplices. In this case, there are three such simplices (that is, there exists at least one such simplex), and they are highlighted in gray. On the right (b), we show Kuhn's triangulation of the 2-simplex $ \{(x_1, x_2) : 0 \leq x_1 \leq x_2 \leq 3\} $, which is based on an allocation instance with $ n = 3 $ agents and $ m = 3 $ items. The colors are assigned to vertices based on the special coloring $L$ derived from an arbitrary non-negative virtual valuation functions. One of the fully-colored elementary $(n-1)$-simplices in $T$, whose existence is guaranteed by Sperner's Lemma, is highlighted in gray.}\label{fig:1}
\end{figure}
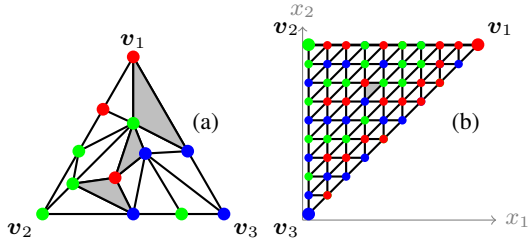
\newpage
\subsection{Figure \ref{fig:2}}
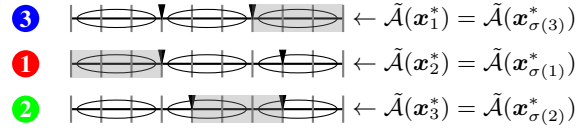
\begin{figure}[h!]
\centering
%scale 0.5 for ijcai
\begin{tikzpicture}[scale=0.6]

% Cerchi colorati sulla sinistra
% Cerchio verde (label 2, figura in basso)
\fill[green] (-1, 0) circle (0.3);
\node[white] at (-1, 0) {\textbf{\small 2}};

% Cerchio rosso (label 1, figura media)
\fill[red] (-1, 1) circle (0.3);
\node[white] at (-1, 1) {\textbf{\small 1}};

% Cerchio blu (label 3, figura in alto)
\fill[blue] (-1, 2) circle (0.3);
\node[white] at (-1, 2) {\textbf{\small 3}};

% (2) figura in basso - Linea orizzontale
\draw[thick] (0,0) -- (6,0);

% Aggiungere 10 linee verticali grige equidistanti per la prima figura
\foreach \x in {0,1,2,3,4,5,6,7,8,9} 
    \draw[gray, thick] (\x*6/9, -0.3) -- (\x*6/9, 0.3);

\foreach \x in {1.5*6/9, 4.5*6/9, 7.5*6/9} {
    \draw (\x, 0) ellipse (1.3*6/9 and 0.2);
}

% Aggiungere i coltelli (triangoli neri) alle posizioni richieste
\draw[fill=black] (4*6/9, 0) -- (4.1*6/9, 0.3) -- (3.9*6/9, 0.3) -- cycle;
\draw[fill=black] (7*6/9, 0) -- (7.1*6/9, 0.3) -- (6.9*6/9, 0.3) -- cycle;

% (1) figura media - Linea orizzontale
\draw[thick] (0,1) -- (6,1);

% Aggiungere 10 linee verticali grige equidistanti per la seconda figura
\foreach \x in {0,1,2,3,4,5,6,7,8,9} 
    \draw[gray, thick] (\x*6/9, 1 - 0.3) -- (\x*6/9, 1 + 0.3);

\foreach \x in {1.5*6/9, 4.5*6/9, 7.5*6/9} {
    \draw (\x, 1) ellipse (1.3*6/9 and 0.2);
}

\draw[fill=black] (3*6/9, 1) -- (3.1*6/9, 1.3) -- (2.9*6/9, 1.3) -- cycle;
\draw[fill=black] (7*6/9, 1) -- (7.1*6/9, 1.3) -- (6.9*6/9, 1.3) -- cycle;

% (3) figura in alto - Linea orizzontale
\draw[thick] (0,2) -- (6,2);

% Aggiungere 10 linee verticali grige equidistanti per la terza figura
\foreach \x in {0,1,2,3,4,5,6,7,8,9} 
    \draw[gray, thick] (\x*6/9, 2 - 0.3) -- (\x*6/9, 2 + 0.3);

\foreach \x in {1.5*6/9, 4.5*6/9, 7.5*6/9} {
    \draw (\x, 2) ellipse (1.3*6/9 and 0.2);
}

\draw[fill=black] (3*6/9, 2) -- (3.1*6/9, 2.3) -- (2.9*6/9, 2.3) -- cycle;
\draw[fill=black] (6*6/9, 2) -- (6.1*6/9, 2.3) -- (5.9*6/9, 2.3) -- cycle;

% (3) figura in alto - Rettangolo senza bordi, riempito con linee oblique grige, 
\fill[gray,opacity=0.3] (6*6/9,2 - 0.3) rectangle (9*6/9, 2+0.3);

% (1) figura media- Rettangolo senza bordi, riempito con linee oblique grige, 
\fill[gray,opacity=0.3] (0,1 - 0.3) rectangle (3*6/9, 1+0.3);

% (2) figura in basso- Rettangolo senza bordi, riempito con linee oblique grige, 
\fill[gray,opacity=0.3] (4*6/9,0 - 0.3) rectangle (7*6/9, 0+0.3);

\node[right] at (9*6/9, 2) {\small $\leftarrow \mathcal{\tilde{A}}(\bfx_{1}^*)=\mathcal{\tilde{A}}(\bfx_{\sigma(3)}^*)$};

\node[right] at (9*6/9, 1) {\small $\leftarrow \mathcal{\tilde{A}}(\bfx_{2}^*)=\mathcal{\tilde{A}}(\bfx_{\sigma(1)}^*)$};

\node[right] at (9*6/9, 0) {\small $\leftarrow \mathcal{\tilde{A}}(\bfx_{3}^*)=\mathcal{\tilde{A}}(\bfx_{\sigma(2)}^*)$};

\end{tikzpicture}

\caption{The figure represents the fractional allocations $\mathcal{\tilde{A}}(\bfx_1^*), \mathcal{\tilde{A}}(\bfx_2^*), \mathcal{\tilde{A}}(\bfx_3^*)$ associated with the vertices $\bfx_1^*, \bfx_2^*, \bfx_3^*$ of the fully-colored simplex $\Delta^*$ in Figure \ref{fig:1}(b) (highlighted in gray), which are colored blue (label 3), red (label 1), and green (label 2) according to the coloring function $L$; the permutation $\sigma:[3]\rightarrow [3]$ such that $L(\bfx_{\sigma(i)})=i$ is defined by $(1,2,3)\xrightarrow{\sigma}(2,3,1)$. The integral items are represented by ellipses,  the fractionality levels of each item (integral, 1-fractional, 2-fractional) are determined by the gray vertical lines, that cut each item in three parts, and the knives that determine the fractional allocations are represented by the black triangles. In such example, as indicated by the vertex colors, the virtual valuation determining $L$ is maximized by the third (resp. first, resp. second) fractional bundle from the left, of the allocation associated with the blue (resp. red, resp. green) vertex of $\Delta^*$, i.e., $\bfx_1^*$ (resp. $\bfx_2^*$, resp. $\bfx_3^*$). Finally, as an illustrative example of left-first or right-first bundles, we observe that the central bundle is right-first in $\Delta^*$.}\label{fig:2}
\end{figure}

\newpage
\subsection{Figure \ref{fig:3}}
\begin{figure}[h!]
\centering
%scale 0.5 for ijcai
\begin{tikzpicture}[scale=0.5]

%%sinistra
\foreach \y in {0,1,2,3,4,5,6,7,8} {

\foreach \x in {1.5*6/9,4.5*6/9, 7.5*6/9} {
    \draw[fill=lightgray, opacity=0.5] (\x, \y) ellipse (1.3*6/9 and 0.2);
}

\fill[white] (3.7*6/9, \y-0.3) rectangle (5.3*6/9, \y+0.3);

\draw[thick] (0,\y) -- (3*6/9,\y);
\draw[dashed, thick] (3*6/9,\y) -- (6*6/9,\y);
\draw[thick] (6*6/9,\y) -- (9*6/9,\y);

 %Aggiungere 10 linee verticali grige equidistanti per la prima figura
\foreach \x in {0,1,2,3,6,7,8,9} 
    \draw[gray, thick] (\x*6/9, \y-0.3) -- (\x*6/9, \y+0.3);

}

%%destra

\foreach \y in {1,3,4,7} {

\foreach \x in {8.2+1.5*6/9,8.2+4.5*6/9, 8.2+7.5*6/9} {
    \draw[fill=lightgray, opacity=0.5] (\x, \y) ellipse (1.3*6/9 and 0.2);
}

\fill[white] (8.2+3.7*6/9, \y-0.3) rectangle (8.2+5.3*6/9, \y+0.3);

\draw[thick] (8.2+0,\y) -- (8.2+3*6/9,\y);
\draw[dashed, thick] (8.2+3*6/9,\y) -- (8.2+6*6/9,\y);
\draw[thick] (8.2+6*6/9,\y) -- (8.2+9*6/9,\y);

 %Aggiungere 10 linee verticali grige equidistanti per la prima figura
\foreach \x in {0,1,2,3,6,7,8,9} 
    \draw[gray, thick] (8.2+\x*6/9, \y-0.3) -- (8.2+\x*6/9, \y+0.3);

}

%% Aggiungere i coltelli (triangoli neri) alle posizioni richieste e label

%coltello, label e allocazione 1

%\label 
\node[black] at (-0.8, 9-1) {{\small 1:}};
%not availability

%coltello
\draw[fill=black] (0*6/9, 8) -- (0.1*6/9, 8.3) -- (-0.1*6/9, 8.3) -- cycle;
\draw[fill=black] (6*6/9, 8) -- (6.1*6/9, 8.3) -- (5.9*6/9, 8.3) -- cycle;

%rettangolo rosso di allocazione
\draw[red, very thick, opacity=0.5] (0*6/9+0.15, 8-0.3) rectangle (6*6/9-0.15, 8+0.3);

%coltello, label e allocazione 2.i

%label
\node[black] at (-0.8, 9-2) {{\small 2(i):}};

%non availability of last item
\fill[pattern=north east lines, pattern color=black, opacity=0.5] (6*6/9, 7-0.3) rectangle (9*6/9, 7+0.3);

%coltelli
\draw[fill=black] (0*6/9, 7) -- (0.1*6/9, 7.3) -- (-0.1*6/9, 7.3) -- cycle;
\draw[fill=black] (7*6/9, 7) -- (7.1*6/9, 7.3) -- (6.9*6/9, 7.3) -- cycle;

%rettangolo rosso di allocazione
\draw[red, very thick, opacity=0.5] (0*6/9+0.15, 7-0.3) rectangle (6*6/9-0.15, 7+0.3);

%coltello, label e allocazione 2.ii (+8)

%label
\node[black] at (8.2-0.8, 9-2) {{\small 2(ii):}};

%%non availability of last item
%\fill[pattern=north east lines, pattern color=black, opacity=0.5] (6*6/9, 7-0.3) rectangle (8+9*6/9, 7+0.3);

%coltelli
\draw[fill=black] (8.2+0*6/9, 7) -- (8.2+0.1*6/9, 7.3) -- (8.2-0.1*6/9, 7.3) -- cycle;
\draw[fill=black] (8.2+7*6/9, 7) -- (8.2+7.1*6/9, 7.3) -- (8.2+6.9*6/9, 7.3) -- cycle;

%rettangolo rosso di allocazione
\draw[red, very thick, opacity=0.5] (8.2+0*6/9+0.15, 7-0.3) rectangle (8.2+9*6/9-0.15, 7+0.3);

%coltello, label e allocazione 3

%label
\node[black] at (-0.8, 9-3) {{\small 3:}};

%non availability of last item
%\fill[pattern=north east lines, pattern color=black, opacity=0.5] (6*6/9, 6-0.3) rectangle (9*6/9, 6+0.3);

%coltelli
\draw[fill=black] (0*6/9, 6) -- (0.1*6/9, 6.3) -- (-0.1*6/9, 6.3) -- cycle;
\draw[fill=black] (8*6/9, 6) -- (8.1*6/9, 6.3) -- (7.9*6/9, 6.3) -- cycle;

%rettangolo rosso di allocazione
\draw[red, very thick, opacity=0.5] (0*6/9+0.15, 6-0.3) rectangle (9*6/9-0.15, 6+0.3);

%coltello, label e allocazione 4

%label
\node[black] at (-0.8, 9-4) {{\small 4:}};

%non availability of last item
%\fill[pattern=north east lines, pattern color=black, opacity=0.5] (6*6/9, 6-0.3) rectangle (9*6/9, 6+0.3);

%coltelli
\draw[fill=black] (1*6/9, 5) -- (1.1*6/9, 5.3) -- (0.9*6/9, 5.3) -- cycle;
\draw[fill=black] (6*6/9, 5) -- (6.1*6/9, 5.3) -- (5.9*6/9, 5.3) -- cycle;

%rettangolo rosso di allocazione
\draw[red, very thick, opacity=0.5] (0*6/9+0.15, 5-0.3) rectangle (6*6/9-0.15, 5+0.3);

%coltello, label e allocazione 5i

%label
\node[black] at (-0.8, 9-5) {{\small 5(i):}};

non availability of last item
\fill[pattern=north east lines, pattern color=black, opacity=0.5] (6*6/9, 4-0.3) rectangle (9*6/9, 4+0.3);

%coltelli
\draw[fill=black] (1*6/9, 4) -- (1.1*6/9, 4.3) -- (0.9*6/9, 4.3) -- cycle;
\draw[fill=black] (7*6/9, 4) -- (7.1*6/9, 4.3) -- (6.9*6/9, 4.3) -- cycle;

%rettangolo rosso di allocazione
\draw[red, very thick, opacity=0.5] (0*6/9+0.15, 4-0.3) rectangle (6*6/9-0.15, 4+0.3);

%coltello, label e allocazione 5ii

%label
\node[black] at (8.2-0.8, 9-5) {{\small 5(ii):}};

%non availability of last item
%\fill[pattern=north east lines, pattern color=black, opacity=0.5] (8.2+6*6/9, 4-0.3) rectangle (8.2+9*6/9, 4+0.3);

%coltelli
\draw[fill=black] (8.2+1*6/9, 4) -- (8.2+1.1*6/9, 4.3) -- (8.2+0.9*6/9, 4.3) -- cycle;
\draw[fill=black] (8.2+7*6/9, 4) -- (8.2+7.1*6/9, 4.3) -- (8.2+6.9*6/9, 4.3) -- cycle;

%rettangolo rosso di allocazione
\draw[red, very thick, opacity=0.5] (8.2+3*6/9+0.15, 4-0.3) rectangle (8.2+9*6/9-0.15, 4+0.3);

%coltello, label e allocazione 6i

%label
\node[black] at (-0.8, 9-6) {{\small 6(i):}};

%non availability of last item
%\fill[pattern=north east lines, pattern color=black, opacity=0.5] (6*6/9, 4-0.3) rectangle (9*6/9, 3+0.3);

%coltelli
\draw[fill=black] (1*6/9, 3) -- (1.1*6/9, 3.3) -- (0.9*6/9, 3.3) -- cycle;
\draw[fill=black] (8*6/9, 3) -- (8.1*6/9, 3.3) -- (7.9*6/9, 3.3) -- cycle;

%rettangolo rosso di allocazione
\draw[red, very thick, opacity=0.5] (3*6/9+0.15, 3-0.3) rectangle (9*6/9-0.15, 3+0.3);

%left-first
\draw[blue, line width=1.5pt, opacity=0.5] (1*6/9, 3.20) circle [radius=0.25];

%coltello, label e allocazione 6ii

%label
\node[black] at (8.2-0.8, 9-6) {{\small 6(ii):}};

%non availability of last item
%\fill[pattern=north east lines, pattern color=black, opacity=0.5] (8.2+6*6/9, 4-0.3) rectangle (8.2+9*6/9, 4+0.3);

%coltelli
\draw[fill=black] (8.2+1*6/9, 3) -- (8.2+1.1*6/9, 3.3) -- (8.2+0.9*6/9, 3.3) -- cycle;
\draw[fill=black] (8.2+8*6/9, 3) -- (8.2+8.1*6/9, 3.3) -- (8.2+7.9*6/9, 3.3) -- cycle;

%rettangolo rosso di allocazione
\draw[red, very thick, opacity=0.5] (8.2+0*6/9+0.15, 3-0.3) rectangle (8.2+9*6/9-0.15, 3+0.3);

%right-first
\draw[blue, line width=1.5pt, opacity=0.5] (8.2+8*6/9, 3.20) circle [radius=0.25];

%coltello, label e allocazione 7

%label
\node[black] at (-0.8, 9-7) {{\small 7:}};

%non availability of last item
%\fill[pattern=north east lines, pattern color=black, opacity=0.5] (6*6/9, 4-0.3) rectangle (9*6/9, 3+0.3);

%coltelli
\draw[fill=black] (2*6/9, 2) -- (2.1*6/9, 2.3) -- (1.9*6/9, 2.3) -- cycle;
\draw[fill=black] (6*6/9, 2) -- (6.1*6/9, 2.3) -- (5.9*6/9, 2.3) -- cycle;

%rettangolo rosso di allocazione
\draw[red, very thick, opacity=0.5] (3*6/9+0.15, 2-0.3) rectangle (6*6/9-0.15, 2+0.3);

%coltello, label e allocazione 8.i

%label
\node[black] at (-0.8, 9-8) {{\small 8(i):}};

%non availability of last item
\fill[pattern=north east lines, pattern color=black, opacity=0.5] (6*6/9, 1-0.3) rectangle (9*6/9, 1+0.3);

%coltelli
\draw[fill=black] (2*6/9, 1) -- (2.1*6/9, 1.3) -- (1.9*6/9, 1.3) -- cycle;
\draw[fill=black] (7*6/9, 1) -- (7.1*6/9, 1.3) -- (6.9*6/9, 1.3) -- cycle;

%rettangolo rosso di allocazione
\draw[red, very thick, opacity=0.5] (3*6/9+0.15, 1-0.3) rectangle (6*6/9-0.15, 1+0.3);

%coltello, label e allocazione 8.ii (+8)

%label
\node[black] at (8.2-0.8, 9-8) {{\small 8(ii):}};

%%non availability of last item
%\fill[pattern=north east lines, pattern color=black, opacity=0.5] (6*6/9, 7-0.3) rectangle (8+9*6/9, 7+0.3);

%coltelli
\draw[fill=black] (8.2+2*6/9, 1) -- (8.2+2.1*6/9, 1.3) -- (8.2+1.9*6/9, 1.3) -- cycle;
\draw[fill=black] (8.2+7*6/9, 1) -- (8.2+7.1*6/9, 1.3) -- (8.2+6.9*6/9, 1.3) -- cycle;

%rettangolo rosso di allocazione
\draw[red, very thick, opacity=0.5] (8.2+3*6/9+0.15, 1-0.3) rectangle (8.2+9*6/9-0.15, 1+0.3);

%coltello, label e allocazione 8.i

%label
\node[black] at (-0.8, 9-9) {{\small 9:}};

%%non availability of last item
%\fill[pattern=north east lines, pattern color=black, opacity=0.5] (6*6/9, 0-0.3) rectangle (9*6/9, 0+0.3);

%coltelli
\draw[fill=black] (2*6/9, 0) -- (2.1*6/9, 0.3) -- (1.9*6/9, 0.3) -- cycle;
\draw[fill=black] (8*6/9, 0) -- (8.1*6/9, 0.3) -- (7.9*6/9, 0.3) -- cycle;

%rettangolo rosso di allocazione
\draw[red, very thick, opacity=0.5] (3*6/9+0.15, 0-0.3) rectangle (9*6/9-0.15, 0+0.3);

%rettangoli grandi
\draw[gray, dashed] (-1.5,6-0.4) rectangle (8.2+9*6/9+0.2,8+0.4);
\node[gray] at (6.6+9*6/9,8) {{\tiny $a_j\equiv 0$ (LV)}};

\draw[gray, dashed] (-1.5,3-0.4) rectangle (8.2+9*6/9+0.2,5+0.4);
\node[gray] at (6.6+9*6/9,5) {{\tiny $a_j\equiv 1$ (BV)}};

\draw[gray, dashed] (-1.5,0-0.4) rectangle (8.2+9*6/9+0.2,2+0.4);
\node[gray] at (6.6+9*6/9,2) {{\tiny $a_j\equiv 2$ (RV)}};
\end{tikzpicture}
\caption{
Given $j \in [n]$, we describe how the fractional bundle $\tilde{A}_j = [a_j, b_j]$ of the main allocation of $\Delta^*$ can be rounded to obtain the integral bundle $A_j$ in each of the following nine cases, assuming that $A_{j+1}, \ldots, A_n$ have already been determined:\\
1:  $a_j\equiv 0,b_j\equiv 0$: $A_j\gets \lpar a^-_j,b^-_j\rpar$;\\
2: $a_j\equiv 0,b_j\equiv 1$: (i) $A_j\gets \lpar a^-_j,b^-_j\rpar$ if $b^+_j\in A_{j+1}$, and (ii) $A_j\gets \lpar a^-_j,b^+_j\rpar$ if $b^+_j\not\in A_{j+1}$;\\
3:  $a_j\equiv 0,b_j\equiv 2$: $A_j\gets \lpar a^-_j,b^+_j\rpar$;\\
4: $a_j\equiv 1,b_j\equiv 0$: $A_j\gets \lpar a^-_j,b^-_j\rpar$;\\
5: $a_j\equiv 1,b_j\equiv 1$: (i) $A_j\gets \lpar a^-_j,b^-_j\rpar$ if $b^+_j\in A_{j+1}$, and (ii) $A_j\gets \lpar a^+_j,b^+_j\rpar$ if $b^+_j\not\in A_{j+1}$;\\
6: $a_j\equiv 1,b_j\equiv 2$: (i) $A_j\gets \lpar a^+_j,b^+_j\rpar$ if $A_j$ is left-first, and (ii) $A_j\gets \lpar a^-_j,b^+_j\rpar$ if $A_j$ is right-first;\\
7: $a_j\equiv 2,b_j\equiv 0$: $A_j\gets \lpar a^+_j,b^-_j\rpar$;\\
8: $a_j\equiv 2,b_j\equiv 1$: (i) $A_j\gets \lpar a^+_j,b^-_j\rpar$ if $b^+_j\in A_{j+1}$, (ii) $A_j\gets \lpar a^+_j,b^+_j\rpar$ if $b^+_j\not\in A_{j+1}$;\\
9: $a_j\equiv 2,b_j\equiv 2$: $A_j\gets \lpar a^+_j,b^+_j\rpar$.\\
The figure illustrates each of the nine cases as follows: each item and its three associated fractionality levels are represented by an ellipse divided into three parts (similarly to Figure \ref{fig:2} in Appendix \ref{app:pictures}); the two black triangles in each case represent the positions ($a_j$ and $b_j$) of the two knives that determine the $j$-th fractional bundle $\tilde{A}_j=[a_j,b_j]$ (by possibly cutting the two board items of the considered bundle); the red rectangle encloses all items that are fully included in the $j$-th bundle $A_j$ of the rounded (integral) allocation $\mathcal{A}$; the right-hand item, when marked with crossed lines, indicates that it was included in bundle $A_{j+1}$ during the previous step of the rounding procedure; the blue circle on the left (resp. right) knife indicates that bundle $[a_j,b_j]$ is left-first (resp. right-first) in $\Delta^*$.
}\label{fig:3}
\end{figure}
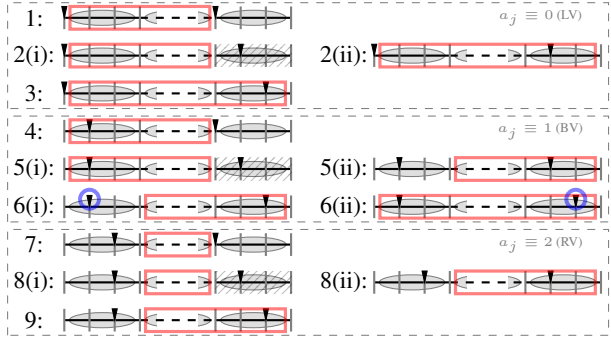
\newpage
\subsection{Figure \ref{fig:4}}
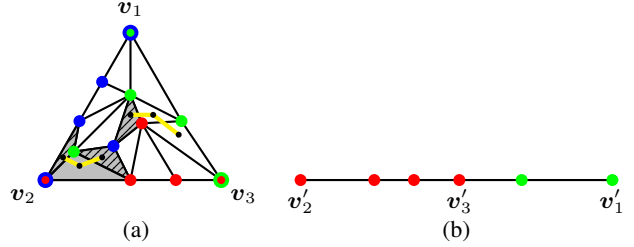
\begin{figure}[h!]
\centering
%scale 0.8 for ijcai
\begin{tikzpicture}[scale=0.75]
%figura 1
% Definiamo i vertici del triangolo grande (simplesso 2D)
\coordinate (A) at (0,0);
\node[below left] at (A) {\small $\vv_2$}; % Etichetta del vertice A

\coordinate (B) at (3.1,0);
\node[below right] at (B) {\small $\vv_3$}; % Etichetta del vertice B

\coordinate (C) at (1.5,2.598);
\node[above, yshift=+0.1cm] at (C) {\small $\vv_1$}; % Etichetta del vertice C

% Definiamo i vertici della triangolazione
\coordinate (q) at (1.2,0.6);
\coordinate (u) at (0.5,0.5);
\coordinate (v) at (0.6,1.0392);
\coordinate (s) at (1,1.732);
\coordinate (w) at (1.5,1.5);
\coordinate (x) at (1.5,0);
\coordinate (y) at (2.3,0);
\coordinate (z) at (1.7,1);
\coordinate (t) at (2.4,1.0392);

% Colleghiamo tutti i punti per formare la triangolazione
\draw[thick] (A) -- (u);
\draw[thick] (A) -- (v);
\draw[thick] (A) -- (x);
\draw[thick] (u) -- (v);
\draw[thick] (u) -- (x);
\draw[thick] (u) -- (w);
\draw[thick] (v) -- (w);
\draw[thick] (v) -- (s);
\draw[thick] (w) -- (z);
\draw[thick] (w) -- (s);
\draw[thick] (w) -- (t);
\draw[thick] (w) -- (q);
\draw[thick] (w) -- (C);
\draw[thick] (s) -- (C);
\draw[thick] (t) -- (C);
\draw[thick] (t) -- (z);
\draw[thick] (t) -- (B);
\draw[thick] (z) -- (x);
\draw[thick] (z) -- (y);
\draw[thick] (z) -- (B);
\draw[thick] (y) -- (B);
\draw[thick] (y) -- (x);
\draw[thick] (q) -- (x);
\draw[thick] (q) -- (u);
\draw[thick] (q) -- (z);

% Coloriamo internamente i fully-colored simplices con pattern grigi
\fill[lightgray] (q) -- (z) -- (w) -- cycle;
\fill[pattern=north east lines] (q) -- (z) -- (w) -- cycle;
%\fill[lightgray] (w) -- (C) -- (t) -- cycle;
\fill[lightgray] (u) -- (x) -- (q) -- cycle;
\fill[pattern=north east lines] (u) -- (x) -- (q) -- cycle;

\fill[lightgray] (u) -- (x) -- (A) -- cycle;

\fill[lightgray] (u) -- (v) -- (A) -- cycle;
\fill[pattern=north east lines] (u) -- (v) -- (A) -- cycle;

\draw[thick] (q) -- (z) -- (w) -- cycle;
%\draw[thick] (w) -- (C) -- (t) -- cycle;
\draw[thick] (u) -- (x) -- (q) -- cycle;
\draw[thick] (u) -- (x) -- (A) -- cycle;
\draw[thick] (u) -- (v) -- (A) -- cycle;

% Coloriamo i vertici del triangolo grande (simplesso 2D)
\fill[blue] (A) circle (4pt); 
\fill[red] (A) circle (2pt); % Sovrapposto leggermente più piccolo

\fill[green] (B) circle (4pt); 
\fill[red] (B) circle (2pt); 

\fill[blue] (C) circle (4pt);
\fill[green] (C) circle (2pt);

% Coloriamo i vertici della triangolazione con le nuove regole
\fill[blue] (q) circle (3pt); % Lato sinistro -> blu
\fill[green] (u) circle (3pt); % Lato sinistro -> blu
\fill[blue] (v) circle (3pt); % Lato sinistro -> blu

\fill[blue] (s) circle (3pt); % Base -> rosso
\fill[red] (x) circle (3pt); % Base -> rosso
\fill[red] (y) circle (3pt); % Base -> rosso

\fill[green] (w) circle (3pt); % Lato destro -> verde
\fill[red] (z) circle (3pt); % Lato destro -> verde
\fill[green] (t) circle (3pt); % Lato destro -> verde

\node at (1.6,-0.9) {\small (a)};

%figura 2

\coordinate (AA) at (4.5,0);
\node[below] at (AA) {\small ${\vv}_2'$};

\coordinate (aa1) at (5.8,0);
\coordinate (aa2) at (6.5,0);

\coordinate (aa3) at (7.3,0);
\node[below] at (aa3) {\small ${\vv}_3'$};

\coordinate (aa4) at (8.4,0);

\coordinate (aa5) at (10,0);
\node[below] at (aa5) {\small ${\vv}_1'$};

\draw[thick] (AA) -- (aa1) -- (aa2) -- (aa3) -- (aa4) -- (aa5);

\fill[red] (AA) circle (3pt);
\fill[red] (aa1) circle (3pt);
\fill[red] (aa2) circle (3pt);
\fill[red] (aa3) circle (3pt);
\fill[green] (aa4) circle (3pt);
\fill[green] (aa5) circle (3pt);

\node at (7.25,-0.9) {\small (b)};

\coordinate (P2) at (2.35,0.8);

\coordinate (P3) at (1.9,1.15);

\coordinate (P4) at (1.5,1.15);

\draw[ultra thick, yellow] (P2) -- (P3) -- (P4);
\fill[black] (P2) circle (1.5pt);
\fill[black] (P3) circle (1.5pt);
\fill[black] (P4) circle (1.5pt);

\coordinate (Q1) at (0.31,0.4);
\coordinate (Q2) at (0.6,0.25);
\coordinate (Q3) at (1,0.4);

\draw[ultra thick, yellow] (Q1) -- (Q2) -- (Q3);

\fill[black] (Q1) circle (1.5pt);
\fill[black] (Q2) circle (1.5pt);
\fill[black] (Q3) circle (1.5pt);
\end{tikzpicture}

\caption{The figure on the left (a) illustrates an application of the Multi-coloring Sperner's lemma to a triangulation \( T \) of a 2-simplex \( \Delta = \text{conv}(\vv_1, \vv_2, \vv_3) \). The multi-coloring function $\LL$ assigns a non-empty subset from \{1 (Red), 2 (Green), 3 (Blue)\} to each vertex \( \bfx \in V(T) \), and it is special; furthermore, in this specific case, $\LL$ satisfies the assumption done w.l.o.g. in the proof of Theorem \ref{multiSperner}, that is: \( \LL(\bfx) = \{i \in [3] : \bfx \in F_i\} \) holds for any vertex $\bfx$ located on the boundary of $\Delta$, and \( | \LL(\bfx) | = 1 \) holds for any internal vertex \( \bfx \) not located on the boundary. The four gray triangles represent the fully-colored simplices in $T$ (w.r.t. the multi-coloring function $\LL$); we observe that, differently from the standard Sperner's lemma, their number is not necessarily odd.\\
The minimal restriction \( L: V(T) \rightarrow [3] \) of \( \LL \), as defined in the proof of Theorem \ref{multiSperner}, assigns the same color as \( \LL \) when that color is unique, and assigns an internal color when it is not (i.e., Green for \( \vv_1 \), and Red for \( \vv_2 \) and \( \vv_3 \)).\\
The figure on the right (b) illustrates the $1$-simplex $\Delta'$ obtained from the instance of figure (a), following the steps outlined in the proof of Theorem \ref{multiSperner}. $\Delta'$ is homeomorphic to the topological space $F'=[\vv_2,\vv_3]\cup [\vv_2,\vv_3]$ obtained by the union of the $1$-dimensional faces $[\vv_2,\vv_3]$ and $[\vv_3,\vv_1]$ of $\Delta$. The figure also illustrates the triangulation (equivalent to) $T'$ and the vertices of $V(T')$ are colored following the minimal restriction $L$ of $\LL$.\\
Finally, the two yellow paths connecting some elementary 2-simplices in the left figure (a)  form the graph \( G = (V, E) \) used in the proof of Theorem \ref{multiSperner}  to establish the existence of an odd number of fully-colored $2$-simplices under the minimal restriction \( L \), namely, the three gray triangles filled with black lines. These triangles are also fully colored with respect to the multi-coloring function 
$\LL$; we observe that, in this case, there are four such triangles, which are shown in gray.
}\label{fig:4}
\end{figure}
\newpage
\subsection{Figure \ref{fig:pir}}
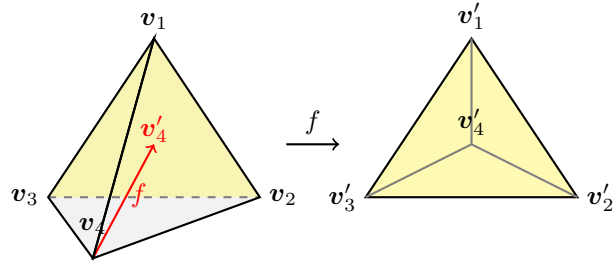
\begin{figure}[h!]
    \centering
    \begin{tikzpicture}[scale=0.7]
        % Definizione dei vertici della piramide
        \coordinate (v3) at (-2,-1,0);
        \coordinate (v2) at (2,-1,0);
        \coordinate (v1) at (0,2,0);
        \coordinate (v4) at (0,-1,3);
        \coordinate (v5) at (0,0,0);

        % Disegno della piramide
        
  %colore faccia F_4
\draw[draw=none, fill=yellow, fill opacity=0.3] (v3) -- (v1) -- (v2) -- cycle;

\draw[draw=none, fill=lightgray, fill opacity=0.1] (v4) -- (v3) -- (v2) -- cycle;

\draw[draw=none, fill=lightgray, fill opacity=0.1] (v3) -- (v1) -- (v4) -- cycle;

\draw[draw=none, fill=lightgray, fill opacity=0.1] (v1) -- (v2) -- (v4) -- cycle;

        \draw[thick, gray, dashed] (v3) -- (v2); % Linea v1-v2 tratteggiata e grigia
        
        \draw[thick, ->, red] (v4) -- (v5);
                
        \draw[thick] (v1) -- (v4) -- (v3) -- cycle;
        \draw[thick] (v1) -- (v4) -- (v2) -- cycle;

        % Etichette dei vertici
        \node[left] at (v3) {$\vv_3$};
        \node[right] at (v2) {$\vv_2$};
        \node[above] at (v1) {$\vv_1$};
        \node[above, yshift=0.2cm] at (v4) {$\vv_4$}; % Etichetta v4 spostata
        \node[below,yshift=-0.4cm,xshift=-0.2cm,red] at (v5) {$f$};
                \node[above,yshift=-0.1cm, red] at (v5) {$\vv_4'$};

        % Mappatura in un triangolo con v_4 interno
        \begin{scope}[shift={(6,0)}]
            \coordinate (w3) at (-2,-1);
            \coordinate (w2) at (2,-1);
            \coordinate (w1) at (0,2);
            \coordinate (w4) at (0,0);
            
            %colorazione
            \draw[draw=none, fill=yellow, fill opacity=0.3] (w3) -- (w1) -- (w2) -- cycle;

            % Disegno del triangolo di proiezione
            \draw[thick] (w3) -- (w2) -- (w1) -- cycle;
            \draw[thick, gray] (w3) -- (w4) -- (w2);
            \draw[thick, gray] (w4) -- (w1);

            % Etichette
            \node[left] at (w3) {$\vv_3'$};
            \node[right] at (w2) {$\vv_2'$};
            \node[above] at (w1) {$\vv_1'$};
            \node[above] at (w4) {$\vv_4'$};
        \end{scope}

        % Freccia di mappatura
        \draw[thick,->] (2.5,0) -- (3.5,0) node[midway, above] {$f$};
    \end{tikzpicture}
    \caption{The figure illustrates the projection \( f \), where \( \Delta \) is a 3-simplex. The left image shows the original simplex \( \Delta \), along with the topological space of \( F' = F_1 \cup F_2 \cup F_3 \). Specifically, the 2-face \( F_4 = \text{conv}(\vv_1, \vv_2, \vv_3) \) opposite to $\vv_4$ is represented with a yellow filling, while the other 2-faces \( F_1, F_2, F_3 \) (i.e., those constituting \( F' \)) are filled with light gray.
The right image depicts the 2-simplex \( \Delta' \), which is obtained by applying the projection \( f \) of \( F' \) onto \( F_4 \). The red arrow on the left indicates the axis and orientation associated with the projection.}\label{fig:pir}
\end{figure}
\newpage
\section{Pseudo-code of the Algorithms}\label{app:algorithms}

\begin{algorithm}[H]
    \caption{Dynamic Programming Algorithm (for Non-negative Valuations).}
    \label{alg4}
    \textbf{Input}: An allocation instance $I=(N,M,(v_i)_{i\in N})$ with non-negative valuations.
    
    \textbf{Output}: An $\EQOP$ allocation.
    
    \begin{algorithmic}[1] %[1] enables line numbers
    	\STATE Initialize a set of integers $C_v\gets \{0\}$;
    	\FOR{$i=1,\ldots, n$}
		 \FOR{$l=1,\ldots, m$}
		 \FOR{$j=1,\ldots, a$}
		 \STATE $c\gets v_i(\lpar l,j\rpar)$ and add $c$ to $C_v$;
		 \ENDFOR
		 \ENDFOR
		 \ENDFOR
         \FOR{$c\in C_v$}
          \STATE Initialize a $\{1,\ldots, n\}\times \{0,\ldots, m\}$ boolean matrix $B$, setting all entries to FALSE;
          \FOR{$j=0,\ldots, m$}
          \IF{$v_1^+(\lpar 1,j\rpar)\geq c\geq v_1^-(\lpar 1,j\rpar)$}
          \STATE $B[1][j]\gets$TRUE;
          \ENDIF
          \ENDFOR
         \FOR{$i=2,\ldots, n$}
         \FOR{$j=0,\ldots, m$}
          \FOR{$\ell=1,\ldots, j+1$}
          \IF{$v_i^+(\lpar \ell,j\rpar)\geq c\geq v_i^-(\lpar \ell,j\rpar)$ and\\ \ \ \ $B[i-1][\ell-1]==\text{TRUE}$}
          \STATE $B[i][j]\gets$TRUE;
          \STATE break;
          \ENDIF
          \ENDFOR
         \ENDFOR
         \ENDFOR
         \IF{$B[n][m]==$TRUE}
         \STATE break;
         \ENDIF
         \ENDFOR
      \STATE Initialize $\mathcal{A}=(A_1,\ldots, A_n)$ arbitrarily;
     \STATE $j\gets m$;
    \FOR{$i=n,\ldots, 2$}
    \FOR{$\ell=1,\ldots, j+1$}
    \IF{$v_i^+(\lpar \ell,j\rpar)\geq c\geq v_i^-(\lpar \ell,j\rpar)$ and\\ \ \ \ $B[i-1][\ell-1]==\text{TRUE}$}
    \STATE $A_i\gets \lpar \ell,j\rpar $;
        \STATE $j\gets\ell-1$
    \STATE break;
    \ENDIF
    \ENDFOR
    \ENDFOR
    \STATE $A_1\gets \lpar 1,j\rpar $
    \STATE return $\mathcal{A}$;
    \end{algorithmic}
\end{algorithm}

\begin{algorithm}[H]
    \caption{Local-search Alg. (for objective valuations).}
    \label{alg1}
    \textbf{Input}: An allocation instance $I=(N,M=G\cup C,(v_i)_{i\in N})$ with objective valuations.
    
    \textbf{Output}: An $\EQXR$ allocation.
    
    \begin{algorithmic}[1] %[1] enables line numbers
        \STATE Initialize $\mathcal{A}=(A_1,\ldots, A_n)$ by setting $A_1\gets M$ and $A_i\gets \emptyset$ for any $i\in N\setminus \{1\}$;
         \IF {$v_1(A_1)>0$}
        \STATE  $i\gets n$;
        \WHILE{there exist $j\in N$ and a good $g\in A_j\cap G$ such that $ v_i(A_i)<v_j(A_j\setminus\{g\})$}
        \STATE $A_i\gets A_i\cup \{g\}$ and $A_j\gets A_j\setminus \{g\}$;
        \STATE  $i\gets \arg\min_{i\in N} v_i(A_i)$;
        \ENDWHILE
                        \STATE \textbf{return} $\mathcal{A}$;
        \ENDIF
        \IF {$v_1(A_1)<0$}
        \STATE $j\gets n$;
        \WHILE{there exist $i\in N$ and a chore $c\in A_i\cap C$ such that $v_i(A_i\setminus \{c\})<v_j(A_j)$}
        \STATE $A_i\gets A_i\setminus \{c\}$ and $A_j\gets A_j\cup \{c\}$;
         \STATE  $j\gets \arg\max_{j\in N} v_j(A_j)$;
        \ENDWHILE
        \ENDIF
        \STATE \textbf{return} $\mathcal{A}$;
    \end{algorithmic}
\end{algorithm}

\begin{algorithm}[H]
    \caption{Greedy Algorithm (for Objective Valuations).}
    \label{alg2}
    \textbf{Input}: An allocation instance $I=(N,M=G\cup C,(v_i)_{i\in N})$ with objective valuations.
    
    \textbf{Output}: An EQ1 allocation.
    
    \begin{algorithmic}[1] %[1] enables line numbers
        \STATE Initialize the partial allocation $\mathcal{A}=(A_1,\ldots, A_n)$ by setting $A_i\gets \emptyset$ for any $i\in N$;
        \FOR{any good $g\in G$}
        \STATE let $h\in \arg\min_{h\in N}v_h(A_h)$;
        \STATE $A_h\gets A_h\cup \{g\}$;
        \ENDFOR
         \FOR{any chore $c\in C$}
        \STATE let $h\in \arg\max_{h\in N}v_h(A_h)$;
        \STATE $A_h\gets A_h\cup \{c\}$;
        \ENDFOR
        \STATE \textbf{return} $\mathcal{A}$;
    \end{algorithmic}
\end{algorithm}

\begin{algorithm}[H]
    \caption{Strongly-greedy Algorithm (for Objective Additive Valuations)}
    \label{alg3}
    \textbf{Input}: An allocation instance $I=(N,M=G\cup C,(v_i)_{i\in N})$ with additive objective valuations.
    
    \textbf{Output}: An $\EQXR$ allocation.
    
    \begin{algorithmic}[1] %[1] enables line numbers
        \STATE Initialize the partial allocation $\mathcal{A}=(A_1,\ldots, A_n)$ by setting $A_i\gets \emptyset$ for any $i\in N$;
        \STATE initialize $G'\gets G$ and $C'\gets C$;
        \WHILE{$G'\neq \emptyset$}
        \STATE let $h\in \arg\min_{h\in N}v_h(A_h)$;
         \STATE let $g\in \arg\max_{g\in G'}v_h(\{g\})$;
        \STATE $A_h\gets A_h\cup \{g\}$; $G'\gets G'\setminus \{g\}$;
        \ENDWHILE
		\WHILE{$C'\neq \emptyset$}
        \STATE let $h\in \arg\max_{h\in N}v_h(A_h)$;
         \STATE let $c\in \arg\min_{c\in C'}v_h(\{c\})$;
        \STATE $A_h\gets A_h\cup \{c\}$; $C'\gets C'\setminus \{c\}$;
        \ENDWHILE
        \STATE \textbf{return} $\mathcal{A}$;
    \end{algorithmic}
\end{algorithm}

%\newpage \ \newpage
%\input{ReproducibilityChecklist}


\end{document}